\documentclass[11pt, letterpaper,dvipsnames, table]{article}

\title{Dynamic Edge Coloring of Forests}

\author{
    Haim {Kaplan}%
    \thanks{The Blavatnik School of Computer Science and AI, Tel Aviv University, Tel Aviv, Israel.
    Emails: \texttt{\{haimk, dnaori, yanivsadeh\}@tau.ac.il}. This research was supported by ISF grant number 1156/23 and the Blavatnik research Foundation.} 
    \and
    David {Naori}\footnotemark[1]
    \and
    Yaniv {Sadeh}\footnotemark[1]
}

\date{}

\usepackage[margin=1in]{geometry}
\usepackage{amsfonts} 
\usepackage{amsmath}
\usepackage{amsthm}
\usepackage{amssymb}
\usepackage{thmtools}
\theoremstyle{plain}
\newtheorem{theorem}{Theorem}[section]
\newtheorem{lemma}[theorem]{Lemma}

\newtheorem{proposition}[theorem]{Proposition}

\theoremstyle{definition}
\newtheorem{definition}[theorem]{Definition}

\theoremstyle{remark}
\newtheorem{remark}[theorem]{Remark}
\newtheorem*{remarkunnumbered}{Remark} 

\newenvironment{remark*}
  {\begin{remarkunnumbered}}
  {\end{remarkunnumbered}}

\newtheorem{observation}[theorem]{Observation}

\usepackage{svg} 
\usepackage[
    colorlinks=true,
    linkcolor=RoyalBlue,
    citecolor=ForestGreen,
    urlcolor=Maroon,
    bookmarks=true
]{hyperref} 
\usepackage{changepage} 
\usepackage{enumerate}
\usepackage[nocompress]{cite} 

\usepackage[lined,ruled,boxed, vlined, linesnumbered]{algorithm2e} 

\usepackage{subcaption}

\usepackage{pgfplots} 
\pgfplotsset{compat=1.6}

\usepackage{tabularx}
\usepackage{makecell}
\newcolumntype{L}{>{\raggedright\arraybackslash}X}
\usepackage{array}
\usepackage{hhline}
\usepackage{colortbl}

\usepackage{multirow}
\usepackage{graphicx}

\usepackage[math]{cellspace}
\usepackage[framemethod=TikZ]{mdframed}
\newenvironment{wrapper}[1][]%
{%
  \begin{mdframed}[
    hidealllines=true,
    backgroundcolor=gray!20,
    innerleftmargin=0.4cm,
    innerrightmargin=0.4cm,
    innertopmargin=0.4cm,
    innerbottommargin=0.4cm,
    roundcorner=8pt,
    #1
  ]%
}%
{%
  \end{mdframed}%
}

\newcommand{\ceil}[1]{\lceil #1 \rceil}
\newcommand{\floor}[1]{\lfloor #1 \rfloor}
\newcommand{\expect}[1]{\mathbb{E}[#1]}

\newcommand{\funcColor}[1]{\textcolor{olive}{#1}}
\SetKwFunction{funcUpdateForestColored}{\funcColor{UpdateForest}}
\SetKwFunction{funcFixForbiddenColored}{\funcColor{FixForbidden}}
\SetKwFunction{funcRootToChildColored}{\funcColor{RootToChild}}
\SetKwFunction{funcChildToRootColored}{\funcColor{ChildToRoot}}
\SetKwFunction{funcUpdateForestColored}{\funcColor{UpdateForest}}
\SetKwFunction{funcRecursiveStepColored}{\funcColor{RecursiveStep}}
\SetKwFunction{funcChooseColorColored}{\funcColor{ChooseColor}}

\SetKwFunction{funcUpdateGraph}{UpdateGraph}
\SetKwFunction{funcFixForbidden}{FixForbidden}
\SetKwFunction{funcRootToChild}{RootToChild}
\SetKwFunction{funcChildToRoot}{ChildToRoot}
\SetKwFunction{funcUpdateGraph}{UpdateGraph}
\SetKwFunction{funcRecursiveStep}{RecursiveStep}
\SetKwFunction{funcChooseColor}{ChooseColor}

\newcommand{\greedy}{{\textsc{Greedy}}}
\newcommand{\greedyshift}{{\textsc{GreedyShift}}}
\newcommand{\greedypath}{{\textsc{GreedyPath}}}
\newcommand{\colorfulpath}{{\textsc{ColorfulPath}}}
\newcommand{\distalg}{{\textsc{DistMaint}}}

\newcommand{\colorf}{\mathsf{color}}
\newcommand{\parent}{\mathsf{parent}}
\newcommand{\nrchild}{\ell}
\newcommand{\grandparent}{\mathsf{grandparent}}
\newcommand{\size}{\mathsf{size}}
\newcommand{\ch}{\mathsf{children}}

\newcommand{\pal}{\kappa}
\newcommand{\extra}{c}

\newcommand{\heavy}{\mathsf{heavy}}
\newcommand{\light}{\mathsf{light}}

\newcommand{\invExtraPlusSqrt}{\frac{1}{\extra + \sqrt{\Delta}}}
\newcommand{\halfPalette}{\frac{\Delta+\extra}{2}}

\newcommand{\poly}{\mathsf{poly}}

\newcommand{\col}[1]{%
  \ifcase#1\relax
  \or\ensuremath{\alpha}
  \or\ensuremath{\beta}
  \or\ensuremath{\gamma}
  \or\ensuremath{\delta}
  \else\textbf{?}\fi
}

\newcommand{\freecol}{\alpha}
\newcommand{\colset}[1]{%
  \ifcase#1\relax
  \or\ensuremath{X}
  \or\ensuremath{Y}
  \or\ensuremath{Z}
  \or\ensuremath{W}
  \else\textbf{?}\fi
}

\newif\iffullversion
\fullversionfalse 

\begin{document}

\maketitle

\begin{abstract}
In the \emph{dynamic edge coloring} problem, one has to maintain a graph of maximum degree $\Delta$ with at most $\Delta+\extra$ colors, under edge updates. A prominent objective is to minimize the \emph{recourse}, namely the number of edges that are recolored. We study this problem on forests, arguably the simplest graph class that already captures much of the complexity of the problem.  We consider both the \emph{incremental} model, where edges are only inserted and the \emph{fully dynamic} model where edges may also be deleted.

In the deterministic setting, we focus on the natural greedy algorithm. We show that it achieves $O(\invExtraPlusSqrt)$ amortized recourse in the incremental model, and that this is tight up to tie-breaking. In contrast, in a fully dynamic forest, greedy can be forced to have $\Omega(\log_\Delta n)$ amortized recourse. To partially overcome this limitation of greedy within the deterministic setting, we give an optimal non-greedy algorithm with $O(1)$ amortized recourse
for \emph{rooted} fully dynamic forests and $\extra=\Delta-2$. In the randomized setting, we give a natural distribution-maintaining algorithm. In the incremental model, it achieves $\Theta(\frac{1}{\Delta})$ expected amortized recourse, and we show that this is optimal for every constant $\extra$. In the fully dynamic model, the same algorithm achieves $\Theta(\min \{ \frac{\Delta}{\extra}, \log_{\Delta} n \})$ expected recourse for $\extra > 0$, and $\Theta(\log_{\Delta} n)$ for $\extra = 0$. We show that this is optimal for $\extra = 0$, and prove an $\Omega(1)$ lower bound for every constant $\extra$. 
\end{abstract}

\newpage

\setcounter{tocdepth}{3}
\tableofcontents
\pagebreak

\section{Introduction}
\label{section_introduction}

In the classical edge coloring problem, the goal is to color the edges of a graph so that no two incident edges receive the same color, while using as few colors as possible. Any graph of maximum degree $\Delta$ requires at least $\Delta$ colors, since all edges incident to a maximum-degree vertex must receive distinct colors. Vizing's theorem~\cite{Vizing1964} shows that $\Delta+1$ colors always suffice, while K\"onig's theorem~\cite{konig1916graphen} shows that $\Delta$ colors suffice for bipartite graphs, and hence for forests.

In the online version of the problem, the edges are revealed one by one, and each edge must be colored irrevocably upon arrival using a palette of $\pal$ colors. The goal is again to keep $\pal$ as small as possible. The classical online greedy algorithm achieves this with $\pal = 2 \Delta-1$ colors, and the lower bound of Bar-Noy, Motwani, and Naor~\cite{BarNoyOnline1992} demonstrates that, in general, no online algorithm can use fewer colors.

\looseness=-1
Dynamic edge-coloring differs from the online setting by allowing \emph{recourse}, that is, the recoloring of previously colored edges. 
More precisely, we need to maintain a proper $\pal$-edge-coloring under a sequence of updates. Updates consist of edge insertions and deletions, while the maximum degree remains at most $\Delta$ throughout. The case where edges can only be inserted (as in the online setting) is called the \emph{incremental model}, while the case allowing both insertions and deletions is called the \emph{fully dynamic model}. 
As before, we want to keep $\pal$ as small as possible, and, additionally, we want to minimize the number of recolored edges per update.

It is convenient to write the palette size $\pal$ as $\pal= \Delta + \extra$, where $\extra \geq 0$. Thus, with $\extra \geq \Delta-1$, the online greedy algorithm maintains a proper coloring with zero recourse. On the other hand, by the online lower bound above, with $\extra < \Delta - 1$, recourse is necessary, even without deletions. Therefore, the goal is to design dynamic algorithms that maintain a proper $(\Delta+\extra)$-edge-coloring for $\extra \leq \Delta - 2$ while minimizing the recourse per update. Ultimately, a desirable objective is to obtain per-update recourse bounds that are independent of the graph size $n$. This leads to the following guiding question:

\begin{wrapper}
    Can we maintain a proper $(\Delta+\extra)$-edge-coloring with per-update recourse independent of $n$ for $\extra \leq \Delta -2$? How does the answer depend on $\extra$?
\end{wrapper}

Without additional assumptions, the general picture remains largely unsettled. The known positive results on this question require $\Delta$ to be sufficiently large as a function of $n$. Interestingly, a large $\Delta$ assumption also makes it possible to beat the online $2\Delta - 1$ barrier and use fewer colors with no recourse at all~\cite{blikstad2024onlineedgecoloringnearly}. Thus, such results may not capture the power of recourse itself.
On the other hand, existing lower bounds \cite{LowerBoundRecourse2018,sadeh2026vizingchainsimprovedrecourse} do not rule out a positive answer without such assumptions. Existing lower bounds are state-based and worst-case: they show that certain colored states require high recourse to extend, but not that every algorithm can be forced into such a state, nor that this can happen often enough to imply large amortized recourse. 

We study this question on forests. Forests are arguably the simplest graph class, yet the known landscape remains essentially the same as in the general setting. The online lower bound already holds on forests, so even the incremental setting remains nontrivial (requiring recourse). Additionally, since forests are $\Delta$-colorable, we can ask the strongest form of our guiding question, namely, what can be achieved with no extra colors at all ($\extra=0$)? 
Moreover, offline edge coloring of forests is straightforward, and thus restricting our focus to this class cleanly isolates the algorithmic difficulties that arise from the dynamic setting. From another angle, a recent line of work has studied dynamic edge-coloring in low-arboricity graphs~\cite{ArboricitySWAT2024_sayan,ArboricitySWAT2024_Christiansen,sadeh2026vizingchainsimprovedrecourse}. This naturally raises the question of what can be achieved in the most basic case of arboricity $1$, namely forests.

\looseness=-1
We study both deterministic and randomized algorithms for this problem, focusing on two of the simplest and most natural algorithmic approaches. On the deterministic side, we focus on the greedy algorithm, henceforth \greedy{}, which restores a proper coloring after each update by recoloring as few edges as possible. On the randomized side, we focus on a \emph{distribution-maintenance} algorithm, henceforth \distalg{}, which  maintains a highly random proper edge-coloring, with the goal of reducing the probability of reaching hard-to-extend states in hopes of achieving small expected recourse. As our results show, these approaches admit a rich range of recourse bounds across the dynamic models we study.

\medskip
\noindent
\textbf{\large Our results.}
On the deterministic side, we show a separation between the incremental and fully dynamic models for all $0 \leq \extra \leq \Delta-2$. In the incremental model, we show that \greedy{} achieves sub-constant amortized recourse $O(\invExtraPlusSqrt)$, and that this is tight for \greedy{} under some tie-breaking rule. In contrast, we show that in the fully dynamic model, any deterministic algorithm must incur $\Omega(1)$ amortized recourse even with $\extra = \Delta-2$. We also demonstrate that \greedy{} performs poorly in the fully dynamic model, and can be forced to incur $\Omega(\log_{\Delta}{n})$ amortized recourse, even in the more restricted rooted-forest setting. To contrast the limitations of \greedy{}, we give a non-greedy deterministic algorithm that achieves $O(1)$ amortized recourse with $\extra=\Delta-2$ in the rooted-forest setting, where the $\Omega(1)$ lower bound still applies, and is therefore asymptotically optimal.

On the randomized side, we obtain stronger upper bounds that are optimal in several regimes. Here too, we show a separation between the incremental and fully dynamic models for every constant $\extra$.
In the incremental setting, we show that \distalg{} achieves $\Theta(1/\Delta)$ expected amortized recourse, and prove a matching lower bound for every constant $\extra$. In the fully dynamic model, we show that \distalg{} achieves $\Theta(\min\{\Delta/\extra, \log_{\pal-1}{n}\})$ expected recourse per-update, which, for $\extra=0$, becomes $\Theta(\log_{\Delta-1}n)$. We prove a matching lower bound for $\extra = 0$, and an $\Omega(1)$ lower bound for every constant $\extra$. 

Together, our results answer the first part of the guiding question on forests. In the incremental model, recourse independent of $n$ is achievable even with no extra colors. In the fully dynamic model, this is impossible for $\extra = 0$, and becomes achievable once extra colors are available. The dependence on $\extra$, however, remains only partially understood. Table~\ref{table_recourse_results_forests} summarizes our main results.

\begin{table}[h]
    \small
    \centering
    \begin{tabularx}{\textwidth}{|Sc|c|c|c|c|Sc|X|}
        \hline
        & Model & Thm. & Algorithm & Extra Colors & Recourse $R$ & Notes \\
        \hline
        \hline

        \multirow{5}{*}{\rotatebox[origin=c]{90}{Deterministic\ \ }}
        & 
        &  \ref{theorem_amortized_insertion_only_improved_UB}
        &  \greedy{}
        &  $\extra \geq 0$
        &  $O(\invExtraPlusSqrt)$
        &  Any tie-breaking. \\
        \cline{3-7}

        &
         \multirow{-2}{*}{Inc.}
        &  \ref{theorem_insert_only_LB}
        &  \greedy{}
        &  $\extra \geq 0$
        &  $\Omega(\invExtraPlusSqrt)$
        &  Some tie-breaking. \\
        \cline{2-7}

        &
        \multirow{3}{*}{Full}
        & \ref{theorem_greedy_fails_LB}
        & \greedy{}
        & $\extra \geq 0$
        & $\Omega(\log_{\Delta} n)$
        & Any tie-breaking. \\
        \cline{3-7}

        &
        & \ref{theorem_amortized_DplusC_LB}
        & Any deterministic
        & $\extra \geq 0$
        & $\Omega(1)$
        & \\ 
        \cline{3-7}

        &
        & \ref{theorem_rooted_trees_recourse_constant_UB}
        & \colorfulpath{}
        & $\extra=\Delta-2$
        & $O(1)$
        & Rooted forests. \\
        \hline
        \hline

        \multirow{5}{*}{\rotatebox[origin=c]{90}{Randomized\ \ }}
        & 
        &  \ref{theorem_unrooted_randomized_LB_UB_short_version}
        &  \distalg{}
        &  $\extra \geq 0$
        &  $\Theta(1/\Delta)$
        &   \\
        \cline{3-7}

        &
         \multirow{-2}{*}{Inc.}
        &  \ref{thm:incremental_randomized_lb}
        &  Any
        &  $\extra = O(1)$
        &  $\Omega(1/\Delta)$
        &   \\
        \cline{2-7}

        &
        \multirow{3}{*}{Full}
        & \ref{theorem_recourse_randomized_short_version}
        & \distalg{}
        & $\extra \geq 0$
        & $\Theta\!\left(\min\left\{\frac{\Delta}{\extra},\log_{\pal-1} n\right\}\right)$ 
        & For $\extra\! = 0$: $\Theta(\log_\Delta n)$. \\
        \cline{3-7}

        &
        & \ref{theorem_trees_delta_colors_LB}
        & Any
        & $\extra=0$
        & $\Omega(\log_{\Delta-1}n)$
        &  \\
        \cline{3-7}

        &
        & \ref{thm:fully_randomized_lb}
        & Any 
        & $\extra = O(1)$
        & $\Omega(1)$
        & \\ 
        \hline
    \end{tabularx}
    \caption{\small{A summary of amortized recourse bounds for $\Delta + \extra$ colors where $\Delta \ge 3$ and $\extra \in [0,\Delta-2]$. ``Inc.'' and ``Full'' stand for the incremental and fully dynamic models, respectively.
    All randomized bounds are in expectation against an oblivious adversary. }}\label{table_recourse_results_forests}
\end{table}

This paper focuses on recourse minimization. 
Recourse minimization is naturally aligned with settings in which changing past decisions requires a fundamentally different treatment than making new decisions. At one extreme is the online setting, where past decisions cannot be changed at all. More generally, changing a past decision may require aborting and restarting an existing commitment, paying a penalty for violating a previous allocation, or performing an expensive external operation, such as updating a database or physically reconfiguring a link~\cite{KMatchings2022,sadeh2024CachingInMatchingsICALP,bMatchingOpticalSwitches2025}. In such settings, recourse captures the dominant cost of processing an update, rather than the computation needed to decide which changes to perform.

However, the same questions can also be asked with update-time minimization as the objective. Under the update-time objective, recoloring an edge is treated like any other computational operation, and the goal is to process each update as quickly as possible. To the best of our knowledge, dynamic edge coloring of forests has not previously been studied separately from either perspective.

Update time and recourse are different but related objectives. If an update incurs recourse $R$, then executing it requires $\Omega(1+R)$ time: we must at least process the update and perform the $R$ recolorings. Thus,  lower bounds on recourse also imply lower bounds on update-time. On the algorithmic side, although we do not directly optimize running time, our strongest upper bounds come from the randomized distribution-maintenance algorithm, which can be implemented in the fully dynamic model in $\Theta(1+R)$ time per update, where $R$ is the recourse incurred by the update. Therefore, even under the update-time objective, recourse is the bottleneck: any asymptotically faster update-time bound would require an algorithm with a stronger recourse bound. 

Naturally, since we restrict attention to forests, the update-time guarantees obtained here are stronger than the known upper bounds for dynamic edge coloring in more general graph classes. We also note that all algorithms studied in this paper have polynomial update time. To avoid conflating the two objectives, we give efficient implementations of our algorithms and discuss update-time bounds in Appendix~\ref{appendix_running_times_analysis}.

\medskip
\noindent
\textbf{\large Paper organization.}
\looseness=-1
The remainder of this section discusses related work (Section~\ref{section_related_work_short}) and introduces notations (Section~\ref{section_preliminaries}).
In Section~\ref{section_technical_overview_det_and_rand}, we provide a technical overview of our deterministic and randomized approaches. Section~\ref{section_deterministic_mostly_greedy_extended} then provides a complete formal account for the deterministic setting, and Section~\ref{section_randomized_analysis} provides the same for the randomized setting. We conclude in Section~\ref{section_conclusions} with open questions.
Appendix~\ref{appendix_running_times_analysis} is dedicated to running time analyses of our algorithms, and Appendix~\ref{section_technical_analysis_appendix} contains deferred technical details from Section~\ref{section_deterministic_mostly_greedy_extended}.

\subsection{Related Work}
\label{section_related_work_short}

\looseness=-1
The study of online edge coloring was initiated by Bar-Noy, Motwani, and Naor~\cite{BarNoyOnline1992}. Their lower bound was shown on an incremental forest, but to the best of our knowledge, no previous works on online edge-coloring study forests directly. See~\cite{blikstad2024onlineedgecoloringnearly, FOCS2025ONlineEdgeColoring} for recent results and further references.

The line of work on dynamic edge coloring is more recent, and in contrast to the online setting, where the graph is incremental by definition, the dynamic edge-coloring works that we are aware of consider only the fully dynamic setting. The works~\cite{DynamicEdgeColoring2019DuanEtal,Bernshteyn2022SmallRecoloring,MultistepVizing2023,2025LinearEdgeColoring} study multi-Vizing chains, which give worst-case recourse bounds, which are  poly-logarithmic in $n$ and polynomial in $\Delta$.  Other techniques include the Nibbling Method~\cite{NibblingCycles2024,NibblingMethod2021} and the Shift Tree~\cite{sadeh2026vizingchainsimprovedrecourse}. 
We cover a few recent results: The algorithm of \cite{2025LinearEdgeColoring} achieves $O(\epsilon^{-4} \cdot \log n)$ worst-case recourse per update for $(1+\epsilon)\Delta$-edge-coloring, in $O(m \epsilon^{-4} \log (1/\epsilon))$ time where $m$ is the number of edges.\footnote{The paper focuses on offline coloring, and this running time suffices to color a whole graph from scratch. There is no explicit refined time per update.} If poly-logarithmic time is sought for, \cite{MultistepVizing2023} gives $(1+\epsilon)\Delta$-edge-coloring with $O(\Delta^6 \log^2 n)$  worst-case recourse in $O(\poly(\epsilon,\log n))$ time per update. For expected recourse, \cite{NibblingCycles2024} maintains $(1+\epsilon)\Delta$-edge-coloring with high probability, with $O(\epsilon^{-4})$ recourse in expectation, in $O(\epsilon^{-9} \log^4 (1/\epsilon))$ time. This result is restricted to $\Delta = \Omega(\poly(\epsilon) \cdot \log n)$. See Table 1 of \cite{sadeh2026vizingchainsimprovedrecourse} for a summary of these and related results. Our focus on forests allows us to design and analyze simpler algorithmic approaches, and to obtain stronger recourse and runtime bounds without the large-$\Delta$ assumption that appear in some previous works.

Edge coloring has also been extensively studied in the distributed model, where vertices locally compute consistent colors for their incident edges. Here the objective is typically to minimize the number of communication rounds (which is related to the recourse in a dynamic settings). See~\cite{LowerBoundRecourse2018} for a summary of this rich literature, and~\cite{SuVuTruncatedVizing2019,Bernshteyn2022SmallRecoloring,MultistepVizing2023,Davies2023DistributedEdgeColoring} for recent advances. 

Beyond the models discussed above, edge coloring has been widely studied both as a fundamental combinatorial question and within other computational frameworks~\cite{GraphEdgeColoring2012Book, BookGraphColoring2024Book}. Two additional settings are adjacent to our context. The first is the \emph{offline centralized model}, the original focus of Vizing's foundational work~\cite{Vizing1964,Vizing1965_10up,Vizing1965_8up}, which has seen significant recent progress~\cite{STOC2025OfflineColoring, 2025LinearEdgeColoring}. The second is the \emph{streaming model}, an online framework where the primary objective is memory efficiency~\cite{StreamingChechikEtal2024ICALP,GhoshEtalWStreamingEdgeColoringICALP2024,StreamingEdgeColoringICALP2025}. Across all these models, a substantial body of work also focuses on specific graph classes, such as bipartite, planar, and generalizations of planar graphs.

\subsection{Preliminaries}
\label{section_preliminaries}

We study a dynamic forest  $F=(V,E)$ where $|V|=n$ is a fixed set of vertices. Its set of edges $E$ is initially empty and changes dynamically, by insertions, and/or deletions. The maximum degree of any vertex, at any time, is bounded by $\Delta$ which is known in advance. 
 We assume $\Delta \ge 3$. The simple case of $\Delta=2$ (trees are paths) is addressed in Section~\ref{section_Delta_2}.

In the \emph{incremental model}, only insertions are allowed, thus the sequence has at most $n-1$ insertions regardless of $\Delta$, unlike general graphs, which could accommodate $\Theta(n \cdot \Delta)$ insertions. In the \emph{fully dynamic model}, edges may also be deleted, so we are still constrained to at most $n-1$ existing edges at any given time, but deletions remove the hard bound on the number of updates and give flexibility to change the forest's topology.\footnote{The \emph{decremental model} is irrelevant to edge coloring, which remains proper under deletions.}

\looseness=-1
Moving on from the topology to the coloring, we are given in advance a palette of size $\pal$ colors to color the edges, where $\pal = \Delta + \extra$ for $0 \leq \extra \leq \Delta-2$. We say that a color $\alpha$ is \emph{available} or \emph{free} at a vertex $v$ if no edge incident to $v$ is colored  $\alpha$. Similarly, $\alpha$ is \emph{available} for an edge $e=(u,v)$ if it is available both at $u$ and $v$.
Formally, we denote the set (sub-palette) of colors available at a vertex $v$ by $A(v)$, and for an edge $e=(u,v)$ we have $A(e) = A(u) \cap A(v)$. We refer to the recoloring of a previously colored edge as \emph{recourse}. In particular, coloring a new edge is not considered recourse.

We measure recourse with respect to a single update. For deterministic algorithms, the \emph{worst-case} recourse is the maximum recourse due to a single update, while \emph{amortized recourse} is the average recourse over the entire update sequence. For randomized algorithms, the recourse of an update is a random variable. We distinguish between \emph{expected recourse per-update}, which bounds the expected recourse of each fixed update in the sequence, and \emph{expected amortized recourse}, which is the expected average recourse over the whole sequence. Throughout the paper, update sequences start from the empty forest.





\section{Technical Overview}

We now give a technical overview of our results, state the main theorems, and sketch selected proofs. We begin with the deterministic results, where we discuss the performance of \greedy{} in both the incremental and fully dynamic models, and the deterministic lower bound. We then turn to the randomized results, where we describe \distalg{}, outline its analysis in both models, and present the randomized lower bounds.

\label{section_technical_overview_det_and_rand}
\iffullversion
\section{Deterministic Algorithms: A Formal Account}
\label{section_deterministic_mostly_greedy_extended}
\else
\subsection{Deterministic Algorithms}
\label{section_deterministic_mostly_greedy_short}
\fi

In the deterministic setting we study the natural greedy algorithm, \greedy{}, that minimizes the recourse of each update.\footnote{Technically \greedy{} is a family of algorithms, that differ by their tie-breaking rules.}
\iffullversion
    For completeness, Algorithm~\ref{algorithm_greedy} is a formal description of \greedy{}. In section~\ref{section_appendix_greedy_variants_differences} we also study a related variant named \greedyshift{}. Table~\ref{table_recourse_results_forests_extras} details additional results (of the complete analysis) that were omitted from Table~\ref{table_recourse_results_forests}.

    \begin{table}[h]
        \small
        \centering
        \begin{tabularx}{\textwidth}{|Sc|c|c|c|c|Sc|X|}
            \hline
            & Model & Thm. & Algorithm & Extra Colors & Recourse $R$ & Notes \\
            \hline
            \hline
    
            \multirow{7}{*}[-2.4ex]{\rotatebox[origin=c]{90}{Deterministic\ \ }}
            &
            & \ref{theorem_Delta_is_2}
            &  \greedy{}
            &  $\Delta = 2$, $\extra = 0$
            &  $O(\log n)$
            &  Optimal. \\
            \cline{3-7}
            &
            
            &  \ref{theorem_insertion_only_Czero_high_LB}
            &  Any shift-based
            &  $\extra = 0$
            &  $\Omega(\frac{1}{\Delta} \cdot \log \frac{n}{\Delta^3})$
            &  \greedyshift{} as an example. \\
            \cline{3-7}
    
            &
             \multirow{-3}{*}{Inc.}
            &  \ref{theorem_insertion_only_Czero_high_shift_UB}
            &  \greedyshift{}
            &  $\extra = 0$
            &  $O(\frac{1}{\Delta} \log n)$
            &  \\
            \cline{2-7}
    
            &
            \multirow{4}{*}[-2.5ex]{Full}
            & \ref{theorem_trees_delta_colors_sublinear_WC_UB}
            & Specific Alg.
            & $\extra = 0$
            & $O \Big ( (\frac{\sqrt{\lg n}}{\Delta} + 1) \cdot 2^{\sqrt{2 \lg n}} \Big )$
            & Worst-case; Used to prove Theorem~\ref{theorem_amortized_insertion_only_improved_UB} for $\extra = 0$. \\
            \cline{3-7}
    
            &
            & \ref{theorem_Delta_is_2}
            & \greedy{}
            & $\Delta = 2$, $\extra = 0$
            & $O(n)$
            & Optimal. \\
            \cline{3-7}
    
            &
            & \multirow{2}{*}[-0.8ex]{\ref{theorem_greedy_fails_LB_shift_based}}
            & \multirow{2}{*}[-0.8ex]{\greedyshift{}}
            & $\extra \ge 1$
            & $\Omega(\log_{\extra+1} {\frac{n}{\Delta - \extra}})$
            & \multirow{2}{\linewidth}{Compare to \cite{sadeh2026vizingchainsimprovedrecourse}, Theorem~B.11:   $O(\log_{\extra+1} {\frac{n}{\Delta+\extra^2}})$ worst-case for $\extra \ge 1$.} 
            \rule{0pt}{0ex} \\ 
            \cline{5-6}
            &
            & 
            & 
            & $\extra = 0$
            & $\Omega(n/\Delta)$
            & \rule{0pt}{5.0ex} \\ 
            \hline
        \end{tabularx}
        \caption{\looseness=-1 \small{A summary of additional results that appear in the full analysis (Section~\ref{section_deterministic_mostly_greedy_extended}), for $\Delta \ge 3$ and $\extra \in [0,\Delta-2]$, unless explicitly stated for $\Delta = 2$. The model can be incremental or fully dynamic. \greedy{} minimizes the recourse on each update, and \greedyshift{} is a greedy variant that only shifts colors (see Definition~\ref{definition_greedy_algorithms}). The recourse is amortized unless stated worst-case. Recall that $\lg$ is the logarithm in base $2$.}}
        \label{table_recourse_results_forests_extras}
    \end{table}

\else
    In section~\ref{section_appendix_greedy_variants_differences} we also study a related variant named \greedyshift{}.
\fi
\iffullversion
    \begin{algorithm}[h]
        
        \DontPrintSemicolon
        \KwIn{
             A sequence of edge updates over a forest $F$ with a maximum degree $\Delta$.
        }
    
        \KwOut{
            Maintains a proper coloring of $F$ after each update.
        }
        
        \SetKwProg{Fn}{Function}{:}{}
        
        \Fn{\funcUpdateForestColored{Forest $F$, edge $e=(u,v)$, insert/delete}}{
            \uIf{insert $e$}{
                Insert $e$ to $F$, and determine $H$: a smallest subset of edges of $F$ that can be recolored to let $e$ have an available color. Recolor $H$ and $e$ accordingly.
            }\Else{
                Delete $e$ from $F$, keep the coloring.
            }
        }
    
        \caption{\greedy{}: a greedy dynamic edge coloring.}
        \label{algorithm_greedy}
    \end{algorithm}
\else
\fi
\greedy{} can be computed easily on forests, as detailed in Appendix~\ref{appendix_running_times_analysis}, and thus it is a natural algorithm to consider. We also show that \greedy{} is asymptotically optimal in the degenerate case of $\Delta=2$ where trees are paths. This further motivates the study of \greedy{}, for $\Delta \ge 3$.

We study the recourse in two models. First, in the incremental model, where edges are only inserted, we show that a greedy approach achieves a very low amortized recourse.  Second, in the fully dynamic model, we demonstrate that this greedy strategy performs poorly. To address this, we introduce a non-greedy algorithm, \colorfulpath{}, that is optimal up to a constant factor, within a restricted setting.

\iffullversion
\subsection{Incremental Forests (Insertions Only)}
\else
\subsubsection{Incremental Forests (Insertions Only)}
\fi
A forest has at most $n-1$ edges regardless of $\Delta$, unlike in general graphs that could accommodate $\Theta(n \cdot \Delta)$ edges. Therefore, the average degree of a vertex in a forest is typically smaller which intuitively makes edge coloring easier. Nevertheless, it is known that the worst-case recourse can be large even in trees~\cite{sadeh2026vizingchainsimprovedrecourse}. Interestingly, Theorem~\ref{theorem_amortized_insertion_only_improved_UB} shows that the amortized recourse is sub-constant.

\iffullversion
    \theoremAmortizedInsertionOnlyImprovedUB*
\else
    \begin{restatable}[]{theorem}{theoremAmortizedInsertionOnlyImprovedUB}
    \label{theorem_amortized_insertion_only_improved_UB}
    Let $\Delta \geq 3$ and $0 \le c \le \Delta-2$. Assume only insertions to a forest with maximum degree $\Delta$. \greedy{} has amortized recourse $O\!\left(\invExtraPlusSqrt \right)$.
    \end{restatable}
\fi

\iffullversion
    To prove the Theorem~\ref{theorem_amortized_insertion_only_improved_UB}, we require the following additional facts. Theorem~\ref{theorem_logCp1_recourse_shiftable_UB} and Theorem~\ref{theorem_trees_delta_colors_sublinear_WC_UB} give us worst-case upper bounds that we use to argue about the recourse per update that we amortize on different edges, for the cases $\extra \ge 1$ and $\extra=0$, respectively. Edges are charged multiple times, and Lemma~\ref{lemma_calculus_summation} is used to bound the total (accumulated) charge per edge. We defer the proofs of Theorem~\ref{theorem_trees_delta_colors_sublinear_WC_UB} and Lemma~\ref{lemma_calculus_summation} to Appendix~\ref{section_technical_analysis_appendix}.

    \begin{theorem}[Theorem B.11 of \cite{sadeh2026vizingchainsimprovedrecourse}]
    \label{theorem_logCp1_recourse_shiftable_UB}
    Consider a dynamic forest $F$ of maximum degree $\Delta$ that is edge colored by $\Delta+\extra$ colors for $\extra \ge 1$. When a new edge $e$ is inserted, the coloring can be extended to $e$ with recourse of at most $\log_{\extra+1} \frac{N}{\Delta+\extra^2} + 2$, that recolors a path, where $N$ is the number of edges in the smaller tree that $e$ links. The update time is $O(N)$.
    \end{theorem}
    
    For our purposes, it suffices to weaken Theorem~\ref{theorem_logCp1_recourse_shiftable_UB} and assume  $O(\log_{\extra+1} N)$ recourse, which is a cleaner expression. The intuition for the weaker expression is that it bounds the length of the cheapest recoloring  path (see also \greedypath{} in Definition~\ref{definition_greedy_algorithms}), because at every vertex the algorithm has (at least) $\extra+1$ recoloring choices that lead the path to a disjoint subtree. This path can always choose to proceed to the smallest subtree among its available options, and it stops at a leaf at the latest.
        
    \begin{restatable}[]{theorem}{theoremSubLinearRecourseTreesDeltaColors}
    \label{theorem_trees_delta_colors_sublinear_WC_UB}
    Consider a dynamic forest $F$ of maximum degree $\Delta \ge 3$ that is edge colored by $\Delta$ colors. When a new edge $e$ is inserted, the coloring of $F$ can be extended to $e$ with recourse of at most $O \Big ( (\frac{\sqrt{\lg N}}{\Delta-2} + 1) \cdot 2^{\sqrt{2 \lg N}} \Big )$ where $N$ is the number of edges in the smaller tree that $e$ links. The update time is $O(N + \Delta \cdot 2^{\sqrt{2 \lg N}})$. 
    \end{restatable}

    \begin{restatable}[]{lemma}{lemmaSummations}
    \label{lemma_calculus_summation}
    Let $\{ x_i \}_{i \ge 1}$ be a sequence of integers such that $x_1 \ge \frac{\Delta}{2} \ge 1$ and $\forall i: x_{i+1} \ge 2 \cdot x_i$. Define the functions: $f(x) \equiv \frac{\lg x}{x}$
    and $g(x) \equiv \frac{\sqrt{2 \lg x} \cdot 2^{\sqrt{2 \lg x}}}{x}$.
    Then $\sum_{i=1}^{\infty} {f(x_i)} = O(f(\Delta))$, and $\sum_{i=1}^{\infty} {g(x_i)} = O(g(\Delta))$.
    \end{restatable}
    
    \begin{proof}[Proof of Theorem~\ref{theorem_amortized_insertion_only_improved_UB}]
    We analyze and charge the amortized recourse of an edge $e$ by considering two types of insertions, referred to as \emph{heavy} and \emph{light}. An insertion is \emph{heavy} if it links trees both having at least $\halfPalette$ edges, and \emph{light} otherwise.    
    
    When an insertion of edge $e=(u,v)$ links two trees $T_u$ and $T_v$, if the insertion is heavy we can charge the recourse immediately, equally to every edge in the smaller tree, without loss of generality, say, $T_v$. When the insertion is light, we cannot charge the recourse immediately because $T_v$ is too small. Therefore we charge it to $u$, which only by the end of the sequence re-distributes its charge among its own edges and the edges of $T_v$ and other trees that linked to $u$ in a light insertion when $u$ was on the large side. The goal is to show that each edge receives a total charge $O(\invExtraPlusSqrt)$ over the whole sequence.
    
    We begin with a trivial but important observation: In order to incur recourse, $e$ must have at least $\Delta+\extra$ neighbours (edges).  Note that the endpoint with the higher degree has at least $\ceil{\halfPalette}$ edges before inserting $e$.

    \noindent
    \textbf{Heavy insertions.} Denote the size of the connected component of $e$ at the $i$th time we charge it by $S_e(i)$. Observe that $S_e(1) \ge \halfPalette$ and that $S_e(i+1) \ge 2 \cdot S_e(i)$, so $\{S_e(i) \}_{i \ge 1}$ satisfy the conditions of Lemma~\ref{lemma_calculus_summation}. If $\extra \ge 1$, we can loosen Theorem~\ref{theorem_logCp1_recourse_shiftable_UB} (applied for $N = S_e(i)$) and deduce a recourse of $O(\log_{(\extra+1)} S_e(i))$, then $e$ is charged its equal share: $O(\frac{\log_{(\extra+1)} S_e(i)}{S_e(i)})$. Now by Lemma~\ref{lemma_calculus_summation} for $f(x)$, the total charge of $e$ is:
    \begin{align*}
    \heavy(e) \le O\!\left(\frac{1}{\lg (\extra+1)} \cdot \frac{\lg \Delta}{\Delta}\right) = O\!\left(\frac{\log_{(\extra+1)} \Delta}{\Delta}\right).    
    \end{align*}
    If $\extra=0$, we analyze similarly using Theorem~\ref{theorem_trees_delta_colors_sublinear_WC_UB}. Assuming $N = S_e(i) \ge 2$ or there is no recourse, $(\frac{\sqrt{\lg N}}{\Delta - 2} + 1) = O(\sqrt{2\lg N})$ and thus the $i$th charge of $e$ is $\frac{1}{S_e(i)}$-fraction out of $O(\sqrt{2 \lg S_e(i)} \cdot 2^{\sqrt{2\lg S_e(i)}})$, and by Lemma~\ref{lemma_calculus_summation} for $g(x)$ the total charge is now bounded by $\heavy(e) = O(\frac{\sqrt{2 \lg \Delta} \cdot 2^{\sqrt{2\lg \Delta}}}{\Delta}) = o(\frac{1}{\sqrt{\Delta}})$.
    
    We emphasize that even though the worst case recourse of a heavy insertion could be large, for example for $\extra \ge 1$, $O(\log_{(\extra+1)} S_e(i))$, overall this recourse is amortized over multiple edges so the amortized recourse is small as we claim.
    
    \smallskip
    \noindent
    \textbf{Light insertions.} By the definition of light, $\deg(v) \le \size(v) < \halfPalette$. In this case we say that $e$ \emph{blames} $u$ for the recourse, and that $u$ \emph{swallowed} $v$. Observe that the recourse is at most $1$ because we can color $e$ by some color used by $v$ on an edge $e'$, uncolor $e'$ and recolor it with an available color because $|A(e')| \ge 1$ ($\size(v) < \Delta+\extra-1$, so every edge in the component of $v$ prior to inserting $e$ has an available color).
    
    When an insertion of $e=(u,v)$ blames $u$ for recourse, the existence of recourse implies $\deg(v) \ge \Delta + \extra - \deg(u)$. Let $\ell \ge 1$ be the number of edges that blame $u$. We charge the recourse blamed on $u$ to the edges of $u$ (all of them) and to the edges of the neighbours it swallowed, at the time of these swallows. The total number of charged edges is at least:
    \begin{align*}
    \deg(u) + \sum_{i=1}^{\ell}{(\Delta+\extra-(\deg(u)-i))}
    =  \deg(u) +
    (\Delta+\extra - \deg(u)) \cdot \ell + \frac{\ell(\ell+1)}{2}.    
    \end{align*}
    
    To clarify, this edge count is minimized when the last $\ell$ edges of $u$ are those that swallow neighbours and blame $u$ for recourse (being the latest edges of $u$ means that the swallowed neighbours can have less edges while still incurring recourse). We charge the recourse blamed on $u$ equally, to get $\frac{1}{\Delta+\extra - \deg(u) + \frac{\deg(u)}{\ell} + \frac{\ell+1}{2}}$ charge per edge. To upper bound this expression, it suffices to lower bound the denominator, the function $f(d,\ell) \equiv \Delta+\extra-d + \frac{d}{\ell} + \frac{\ell+1}{2}$. We assumed some blame of $u$, then $\ell \ge 1$, then $f$ is non-increasing in $d$ ($\frac{\partial f}{\partial d} = \frac{1}{\ell} - 1$), so choosing $d = \Delta$ minimizes $q$ with respect to $d$. Furthermore, given a fixed $d$, $f$ attains minimum at $\ell = \sqrt{2d}$ ($\frac{\partial f}{\partial \ell} = -\frac{d}{\ell^2} + \frac{1}{2}$; $\frac{\partial^2 f}{{\partial \ell}^2} = \frac{2d}{\ell^3} > 0$). Overall, we get that
    $f(d,\ell)
    \ge 
    f(\Delta,\sqrt{2\Delta})
    \ge 
    \extra + \frac{1}{2} + \sqrt{2\Delta}
    $.
    
    Then the light charge due to blames of $u$ is upper bounded by $\frac{1}{1+ (\extra-\frac{1}{2}) + \sqrt{2\Delta}} = \Theta(\invExtraPlusSqrt)$. Observe that every edge can be charged with recourse up to thrice: once per endpoint (in the role of $u$), and once for belonging to a small tree that is being swallowed, so overall, the charge of every edge satisfies $\light(e) = O(\invExtraPlusSqrt)$.
    
    To conclude, we showed for each edge $e$ that $\light(e)+\heavy(e) = O(\invExtraPlusSqrt)$.
    \end{proof}    
    
\else
    The proof relies on a somewhat delicate charging argument that splits into cases according to the size of the trees being linked, small or large, and it is deferred to Section~\ref{section_deterministic_mostly_greedy_extended}.
\fi

Theorem~\ref{theorem_amortized_insertion_only_improved_UB} shows that greedy coloring is  efficient in terms of recourse when edges are only inserted: The recourse ranges from $O(\frac{1}{\sqrt{\Delta}})$ for $\extra = O(\sqrt{\Delta})$ down to $O(\frac{1}{\Delta})$ for $\extra = \Omega(\Delta)$. Theorem~\ref{theorem_insert_only_LB} below shows that this analysis is tight for some tie-breaking rule. 

\iffullversion
    \theoremInsertOnlyLB*
\else
    \begin{restatable}[]{theorem}{theoremInsertOnlyLB}
    \label{theorem_insert_only_LB}
    Let $\Delta \geq 3$ and $0 \le c \le \Delta-2$. There exists an insertion sequence on a forest of maximum degree $\Delta$, and a tie-breaking rule, such that \greedy{} has amortized recourse $\Omega(\invExtraPlusSqrt)$.
    \end{restatable}
\fi

\iffullversion
    \begin{proof}
    Let $1 \le \ell \le \Delta-1-\extra$ be a parameter that we choose later so as to optimize the resulting bound. Fix a vertex $u$, and insert $\Delta-\ell$ leaf edges incident to $u$. Break ties so that these edges receive the highest colors: $\ell+\extra+1,\ell+\extra+2,\dots,\Delta+\extra$. Thus, initially, the available colors at $u$ are $A(u)=\{1,2,\dots,\ell+\extra\}$.
    
    Now, for each $i=1,2,\dots,\ell$, create a fresh star with center $v_i$ and $\ell+\extra+1-i$ leaves, and break ties so that the edges incident to $v_i$ receive exactly the colors $i,i+1,\dots,\ell+\extra$. Then insert the edge $(u,v_i)$. We show by induction on $i$ that there is a tie-breaking rule, such that when the edge $(u,v_i)$ is inserted, there is no free color for it (i.e., $A(u)\cap A(v_i)=\emptyset$), \greedy{} colors $(u,v_i)$ by the color $i$, and recolors one existing leaf-edge of $v_i$.
    
    The base case ($i=1$) is immediate since $u$ and $v_1$ together use all the colors, so we choose a tie-breaking rule that colors $(u,v_1)$ with the color $1$, and recolors the unique leaf-edge of $v_1$ whose color is $1$ to some other free color. Next, fix $i \geq 2$. By the induction hypothesis, just before inserting $(u,v_i)$, the previous $i-1$ such insertions have used the colors $1,2,\dots,i-1$ at $u$, while the original leaf edges at $u$ still use the colors $\ell+\extra+1,\dots,\Delta+\extra$. Hence, $A(u)=\{i,i+1,\dots,\ell+\extra\}$. On the other hand, the colors used at $v_i$ are exactly $\{i,i+1,\dots,\ell+\extra\}$. Therefore, $A(u)\cap A(v_i)=\emptyset$, so, recourse $0$ is impossible. As before, we choose a tie-breaking rule that colors $(u,v_i)$ by the color $i$, and recolors the unique leaf-edge of $v_i$ whose color is $i$ to some other available color. This concludes the induction.
    
    Thus the algorithm pays recourse of exactly $1$ for each of the $\ell$ insertions $(u,v_i)$. The total recourse is therefore exactly $\ell$. We insert $\Delta$ edges for $u$ and $\ell+\extra+1-i$ edges for $v_i$, so the total number of insertions is $\Delta+\sum_{i=1}^{\ell}(\ell+\extra+1-i)
    =
    \Delta+\frac{\ell^2+(2\extra+1)\ell}{2}$.
    Hence, the amortized recourse is $f(\ell) \equiv \frac{\ell}{\Delta + (\ell^2 + (2\extra+1)\ell)/2} = \frac{2}{\frac{2\Delta}{\ell} + \ell + (2\extra+1)}$. $f(\ell)$ is maximized when the denominator is minimized, which is when $\ell = \sqrt{2\Delta}$. Since $\ell$ should be an integer, we get that the amortized recourse is maximized for either $\ell^* = \ceil{\sqrt{2\Delta}}$ or $\ell^* = \floor{\sqrt{2\Delta}}$, and it is $\Omega(\invExtraPlusSqrt)$. The choice $\ell = \ell^*$ might not be feasible, because we are constrained to $\ell \le \Delta-1-\extra$. If this condition is violated, then $\Delta-1-\extra < \ell^* < \sqrt{2\Delta} + 1$, therefore $\extra > \Delta-2 - \sqrt{2\Delta} \Rightarrow \extra = \Theta(\Delta)$. Then choosing $\ell = 1$ instead of $\ell = \ell^*$ gives $f(1) = \frac{1}{\Delta+\extra+1}$ which is also $\Omega(\invExtraPlusSqrt)$ since $\extra = \Theta(\Delta)$.
    \end{proof}
\else
\fi

\iffullversion
\subsection{Fully Dynamic Forests (Insertions and Deletions)}
\else
\subsubsection{Fully Dynamic Forests (Insertions and Deletions)}
\fi

When the graph is fully dynamic, the problem is harder. In general, $o(1)$ amortized recourse is no longer possible as Theorem~\ref{theorem_amortized_DplusC_LB} below shows. Concretely,
for a sufficiently large $n$, the amortized recourse is $\Omega(1)$. Details of what is ``sufficiently large'' appear in Section~\ref{section_dependence_of_n}. 

\iffullversion
    \theoremDynamicLB*

    \begin{remark*}
    The value of $n_0$ is studied in detail in Section~\ref{section_dependence_of_n}.
    \end{remark*}
\else
    \begin{restatable}[]{theorem}{theoremDynamicLB}
    \label{theorem_amortized_DplusC_LB}
    Let $\Delta \geq 3$ and $0 \le c \le \Delta-2$. There exists $n_0=n_0(\Delta,\extra)$ such that for every deterministic algorithm $ALG$, every $n \ge n_0$, and every $\epsilon > 0$, there is an update sequence on an $n$-vertex forest, such that the amortized recourse of $ALG$ on this sequence is greater than $1/2 - \epsilon$. The amortized recourse is already $\Omega(1)$ after $O(n_0)$ updates.
    \end{restatable}
\fi

\begin{proof}
    Let $\mathcal{P}$ be the family of sub-palettes of $[\pal]$ of size $\Delta-1$. For each such sub-palette $P \in \mathcal{P}$, we associate a unique \emph{owner} vertex $o_P$. From the remaining vertices, we, who play the role of the adversary, create $N \equiv (\extra+1)\cdot |\mathcal{P}| + 1$ disjoint stars of degree $\Delta-1$. $n_0$ is the number of vertices needed for this construction.
    
    For a star $S$, we refer to the set of colors $P \in \mathcal{P}$ used on its leaf-edges as the \emph{palette of $S$}. We now repeat the following \emph{step}: take an isolated star and connect its center to the owner of its palette. Then, while there is a connected star whose palette has changed, we disconnect it from its owner. Observe that after each of our steps, at most $\extra+1 < \Delta$ stars can be connected to each owner. This is because all edges incident to $o_P$ must use distinct colors outside $P$, and there are only $\extra+1$ such colors. Additionally, since $N > (\extra+1)\cdot|\mathcal{P}|$, we can repeat this step indefinitely.
    
    At each step, we connect one star and may disconnect several others. Since there are at most $\extra+1$ stars connected to each owner, after $s$ steps we must have disconnected at least $s - (\extra+1)\cdot|\mathcal{P}|$ stars. 
    Since we  disconnect a star only if its palette has changed, it follows that at least $s - (\extra+1)\cdot|\mathcal{P}|$ edges have been recolored by the algorithm. The total number of updates we make in these $s$ steps is at most $2s$, as each edge we delete was inserted in a previous step. In addition, we use at most $n_0$ updates for the initial creation of the stars. Hence, the amortized recourse is at least $\frac{s - (\extra+1)\cdot|\mathcal{P}|}{2s+n_0}$, which is $\Omega(1)$ already for $s = 2 \cdot n_0$ and approaches $1/2$ as $s \rightarrow \infty$.
\end{proof}

The $\Omega(1)$ lower bound of Theorem~\ref{theorem_amortized_DplusC_LB} may appear small (a mere constant),
but we should compare it to the subconstant incremental amortized recourse of \greedy{} (Theorem~\ref{theorem_amortized_insertion_only_improved_UB}), which is $O(\Delta^{-\epsilon})$ for $\epsilon \in [\frac{1}{2},1]$, depending on $\extra$. Perhaps surprisingly, the same greedy algorithm that performs so well when there are only insertions, performs poorly in the fully dynamic model, incurring $\Omega(\log_\Delta {n})$  amortized recourse.

\iffullversion
    \theoremGreedyRecourse*

    \begin{remark*}
    The value of $n_1$ is studied in detail in Section~\ref{section_dependence_of_n}.
    \end{remark*}
\else
    \begin{restatable}[Greedy recourse]{theorem}{theoremGreedyRecourse}
    \label{theorem_greedy_fails_LB}
    Let $\Delta \geq 3$ and $0 \le c \le \Delta-2$. There exists a forest and coloring such that from this state onward, the amortized  recourse of \greedy{} is $\Omega(\log_\Delta n)$. There exists $n_1 = n_1(\Delta,\extra)$ such that for every $n \ge n_1$, this bad state can be reached from the empty graph.
    \end{restatable}

\fi
\iffullversion
    We postpone the proof of Theorem~\ref{theorem_greedy_fails_LB} to Section~\ref{section_greedy_high_recourse}.
\else 
    \begin{proof}[Proof Sketch of Theorem~\ref{theorem_greedy_fails_LB}]
    We partition the palette into two complementary sub-palettes $P$ and $\overline{P}$, such that $P$ has size $\Delta-1$. The proof relies on layered trees whose coloring alternate between these two sub-palettes, layer by layer. Vertices at even depth have $|P|$ child-edges colored from $P$, and vertices at odd depth have $|\overline{P}|$ child-edges colored from $\overline{P}$.
    We show that when an edge is inserted between the roots of a $P$-layered and a $\overline{P}$-layered trees, \greedy{} recolors exactly a  path from a root to a leaf $v$, and that the resulting combined tree, when we delete the edge of $v$, is itself a $P$-layered or a $\overline{P}$-layered (depending on parity) rooted at the former parent of $v$.

    \looseness=-1
    We start with two perfect layered trees, and perform a repeating cycle of $3$ deletions and $3$ insertions such that the path that \greedy{} recolors goes back-and-forth on the same edges, effectively undoing the changes so that by the end of the cycle we get back to the initial state. The topology, coloring and the cycle of updates are illustrated in Figure~\ref{figure_adversarial_cycle_of_requests} and explained in its caption. The depths of the trees are $\Omega(\log_\Delta n)$ and therefore this expression lower bounds the amortized recourse. The construction of this initial state from an empty forest is deferred to Section~\ref{section_greedy_high_recourse}.
    \end{proof}

    \begin{figure*}[t]
        \centering
        \includegraphics[width=0.82\textwidth]{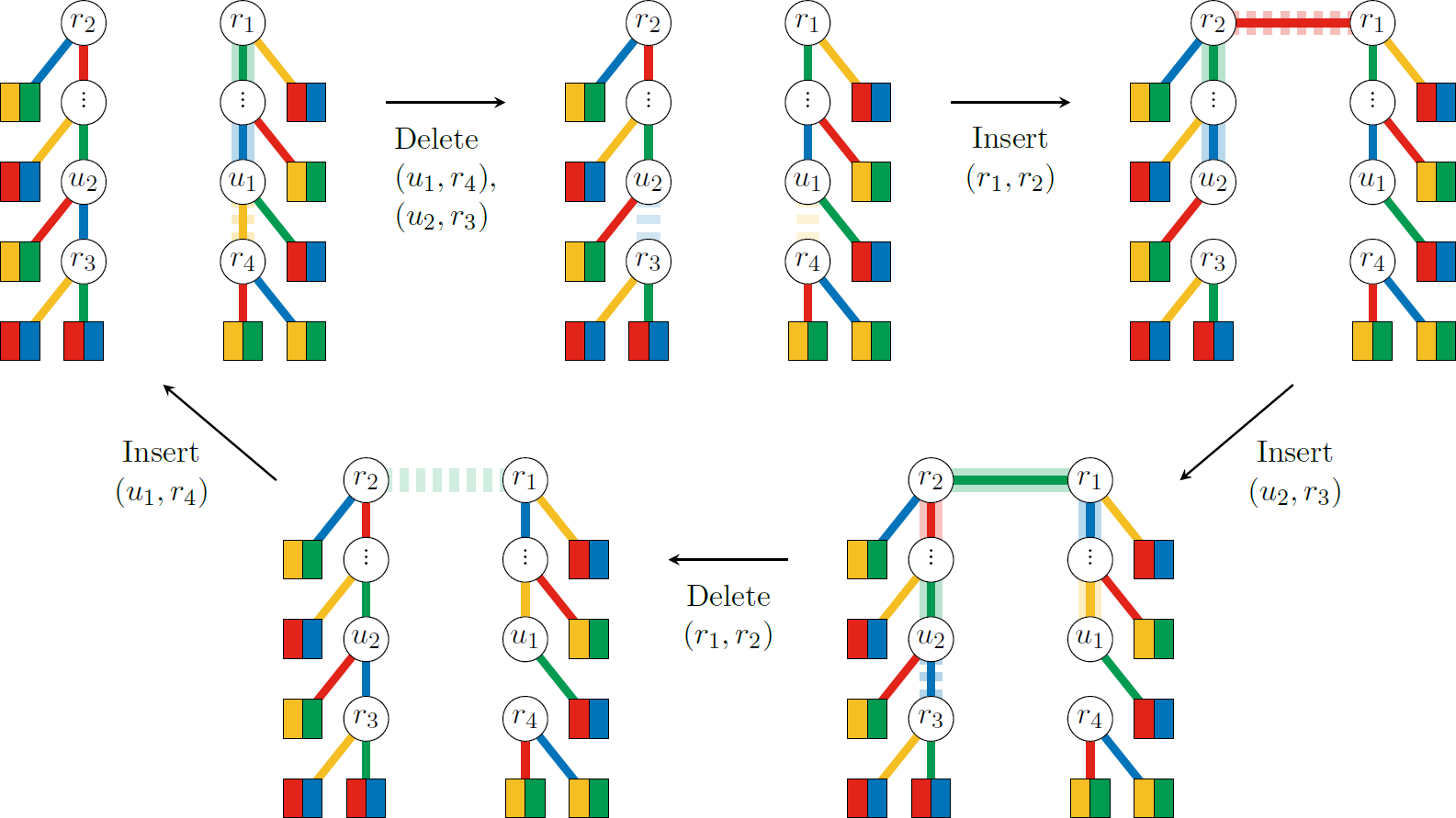}
        \vspace{0.2cm}
        
    	\caption{\small{Visualization for the proof of Theorem~\ref{theorem_greedy_fails_LB}, with sub-palettes $P=\{\text{\color{Dandelion} yellow}, \text{ \color{OliveGreen} green}\}$, and $\overline{P} = \{\text{\color{red} red}, \text{ \color{blue} blue}\}$. Bicolored squares abstract layered subtrees where the two colors are the colors used in the next layer. Note that the depth of these subtrees is \emph{not} illustrated in this figure. The actual depths are as follows: In the initial state, top left, we have two perfect layered trees where each leaf is of depth $d$ for $d = \Theta(\log_{\Delta} n)$. The distance from $r_i$ to $u_i$ in edges is $\floor{\frac{d}{3}} - (i-1)$ for $i \in \{ 1,2 \}$ (chosen differently to prevent ties). Updated edges are highlighted with a dashed background, and recolored edges are highlighted with a solid background. The cycle of updates is as follows: Starting from the initial state (top left), delete $(u_1,r_4)$ and $(u_2,r_3)$ to get $4$ trees, then insert $(r_1,r_2)$, insert $(u_2,r_3)$, delete $(r_1,r_2)$, and finally insert $(u_1,r_4)$. This costs $\Theta(\log_\Delta n)$ recourse and returns to the same initial state.}}
        \vspace{-0.5cm}
    	\label{figure_adversarial_cycle_of_requests}
    \end{figure*}
\fi

\smallskip

If the lower bound of Theorem~\ref{theorem_amortized_DplusC_LB} is tight, we should aim for a deterministic algorithm with $O(1)$ amortized recourse.  Theorem~\ref{theorem_greedy_fails_LB} rules out \greedy{}. Interestingly, if the forest is \emph{rooted} (see Definition~\ref{definition_rooted_forest} below), and if the palette has $2\Delta-2$ colors, then a simple non-greedy deterministic algorithm, which we call \colorfulpath{}, achieves $O(1)$ recourse. The exact claim is Theorem~\ref{theorem_rooted_trees_recourse_constant_UB} below. Note that the $\Omega(1)$ lower bound of Theorem~\ref{theorem_amortized_DplusC_LB}, and the lower bound on the recourse of \greedy{} from Theorem~\ref{theorem_greedy_fails_LB} also apply to rooted forests. 
\iffullversion
\else We defer the definition and analysis of the \colorfulpath{} to Section~\ref{section_deterministic_mostly_greedy_extended}.
\fi

\iffullversion
    \definitionRootedForest*
\else
    \begin{restatable}[]{definition}{definitionRootedForest}
    \label{definition_rooted_forest}
    A dynamic forest is \emph{rooted} if it is a collection of trees, each with a fixed assigned root. An insertion must be of the form $(p,r)$ where $r$ is a root of its tree, that becomes a child of $p$. (If both are roots, the parent must be specified.) A deletion of $(u,v)$ creates a new tree, such that if (without loss of generality) $u$ remains reachable from the root of the tree that contained $(u,v)$, then $v$ becomes the root of the new tree.
    \end{restatable}
\fi

\iffullversion
    \theoremRootedTreesDynamicTwoDeltaDeterministicConstant*
\else
    \begin{restatable}[]{theorem}{theoremRootedTreesDynamicTwoDeltaDeterministicConstant}
    \label{theorem_rooted_trees_recourse_constant_UB}
    Let $\Delta \ge 3$, and let $F$ be a dynamic \emph{rooted} forest, initially empty. \colorfulpath{} maintains a proper $(2\Delta-2)$-edge-coloring on $F$ with $O(1)$ amortized recourse.
    \end{restatable}
\fi

\iffullversion
    \begin{algorithm}[t]
        
        \DontPrintSemicolon
        \KwIn{
             A sequence of edge updates over a forest $F$ with a maximum degree $\Delta$.
        }
    
        \KwOut{
            Maintains a proper coloring of $F$ after each update.
        }
        
        \SetKwProg{Fn}{Function}{:}{}
        
        \Fn{\funcUpdateForestColored{Forest $F$, edge $e=(u,v)$, insert/delete}}{
            If insert $e$: Add $e$ to $F$, uncolored (yet) and call \funcRecursiveStepColored{$F$, $e$, $u$}; Otherwise delete $e$ from $F$ and keep the coloring.
        }
    
        \Fn{\funcRecursiveStepColored{Forest $F$, uncolored edge $e=(u,v)$, entry vertex $u$}}{
            \textbf{if} $A(u) \cap A(v) \ne \emptyset$ \textbf{then} Color $e$ by any $\col{1} \in A(u) \cap A(v)$. \Return.
    
            \uIf {$\exists w \in \ch(v)$ such that $\colorf(v,w) \in A(u)$ and $A(w) \cap A(v) \ne \emptyset$}{
                    Color $e$ by $\colorf(v,w)$, uncolor $(v,w)$, and call \funcRecursiveStepColored{$F$, $(v,w)$, $v$}.
            }\Else{
                Let $\col{1} = \funcChooseColorColored{F,e,u}$ and denote the edge of $v$ colored by $\col{1}$ by $(v,w)$.
                Color $e$ by $\col{1}$, uncolor $(v,w)$, and call \funcRecursiveStepColored{$F$, $(v,w)$, $v$}.
            }
                
        }
    
        \Fn{\funcChooseColorColored{Forest $F$, uncolored edge $e=(u,v)$, entry vertex $u$}}{
            Let $\col{2} = \colorf(\grandparent(u),\parent(u))$ if such edge exists, otherwise $\col{2} = null$. Let $w$ be an arbitrary child of $v$ subject to $\col{1} \equiv \colorf(v,w) \ne \col{2}$. \Return $\col{1}$.
        }
        
        \caption{\colorfulpath{}: Rooted-forest dynamic edge coloring, using $2\Delta-2$ colors.}
        \label{algorithm_rooted_forest}
    \end{algorithm}
\else
    
\fi

\iffullversion
    We defer the proof of Theorem~\ref{theorem_rooted_trees_recourse_constant_UB} to Section~\ref{subsection_rooted_TwoDelta}. The complete pseudo-code description of \colorfulpath{} is in Algorithm~\ref{algorithm_rooted_forest}, it is exemplified in Figure~\ref{figure_shifted_rooted_e_ready}, and its high level idea is as follows.

    When we insert an edge $e=(p,r)$ that connects $p$ as the parent of a root $r$, we shift colors upwards in the tree, along a downward path from $r$, as follows. The edge $e = (p,r)$ is initially uncolored. If it has an available color (i.e.\ $A(p)\cap A(r) \not= \emptyset$) we just use it and finish. Otherwise, if $r$ has a child $v$ such that $(r,v)$ has an available color $\col{1} \in A(r)\cap A(v)$ then it must be that the current color $\col{2}$ of $(r,v)$ is in $A(p)$ (this follows since we have $2\Delta - 2$ colors and no free color for $(p,r)$). So we shift $\beta$ from $(r,v)$ to $(p,r)$, recurse on $(r,v)$ and finish in the next step (by coloring it $\col{1}$). If no edge $(r,v)$ has a free color, we pick an edge $(r,v)$ whose color is different from the color of the edge $e'' \equiv (\grandparent(p),\parent(p))$ (if such $e''$ exists), shift its color to $(p,r)$ (which again is possible since we have $2\Delta - 2$ colors and no free color for $(p,r)$), and recurse on $(r,v)$ which is now uncolored.

    \looseness=-1
    The choice $\colorf{}(e) \ne \colorf{}(e'')$ essentially guarantees that we have an available color for the edge between them, $e'$, when we reach it next time: Suppose that we pick the color $\col{1}$ for an edge $e = (v,w)$ whose parent is $e' = (u,v)$ and its grandparent is $e'' = (x,u)$, in a recoloring that is not last or next to last (as part of the whole recoloring due to the recent insertion). Then by coloring $e$ with a color $\col{1} \ne \col{2}$ where $\col{2} = \colorf{}(e'')$, we guarantee that $\col{2} \in A(v)$ (it was surely in $A(v)$ when we shifted $\col{2}$ from $e'$ to $e''$). Since trees are rooted, next time we reach  $e'$ in a preceding step we will shift $\col{2}$ from $e''$ upwards, vacating it so that $\col{2} \in A(u)$ and therefore $\col{2} \in A(e')$. Thus the update would terminate before reaching $e$ again, assuming that no additional updates occur in the neighbourhood of $e'$. In the proof we formalize this using a potential argument showing that if the recourse is high, the algorithm guarantees that all but $O(1)$ of the recolored edges will have an available color in a future re-visit.
    
    \vspace{0.3cm}
    \begin{figure}[h!]
        \centering
        \begin{subfigure}{0.35\textwidth}
        \includegraphics[width=1\linewidth]{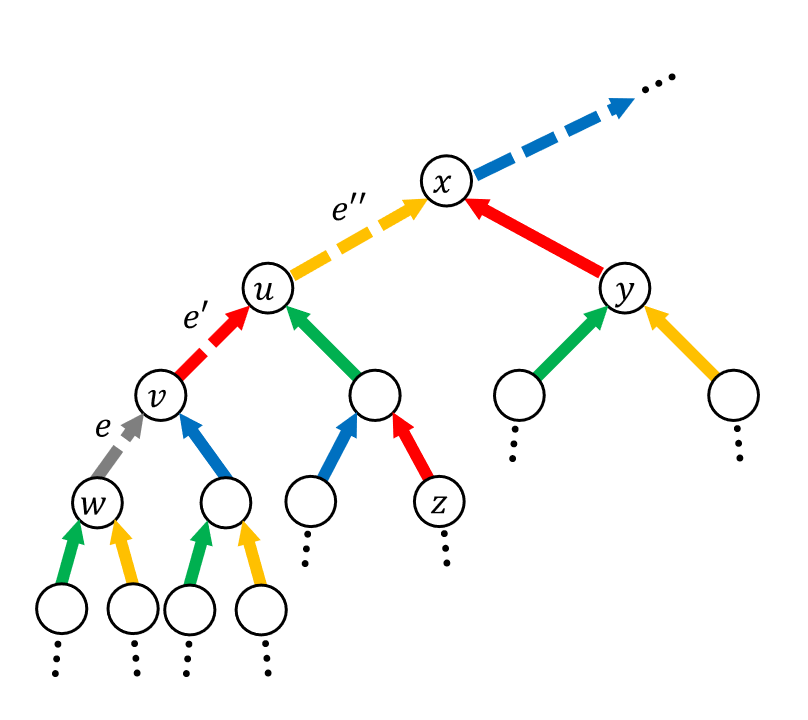}
        \subcaption{}\label{fig:first}
        \end{subfigure}
        \hspace{1cm}
        \begin{subfigure}{0.35\textwidth}
        \includegraphics[width=1\linewidth]{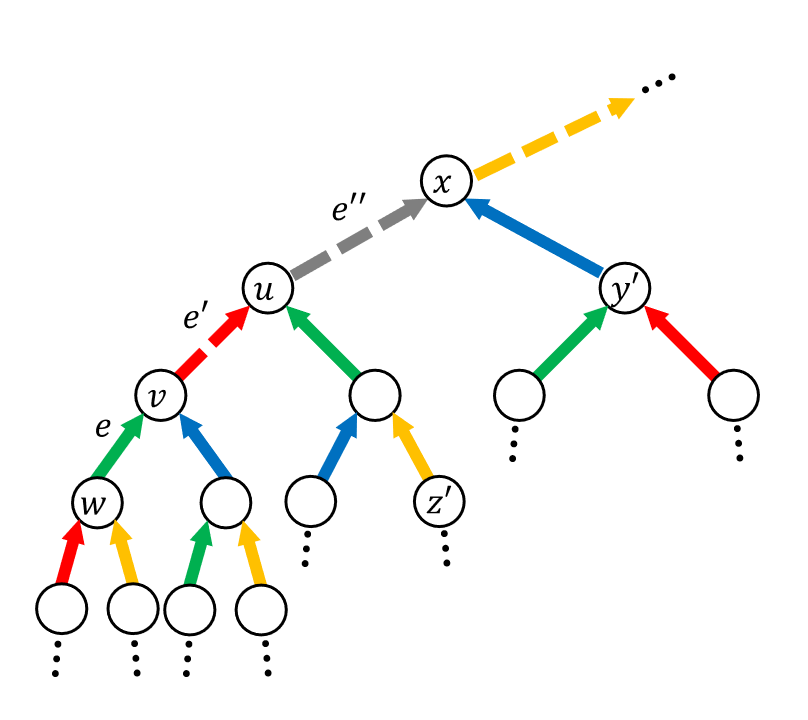}
        \subcaption{}\label{fig:second}
        \end{subfigure}\\[0.5cm]
        \begin{subfigure}{0.35\textwidth}
        \includegraphics[width=1\linewidth]{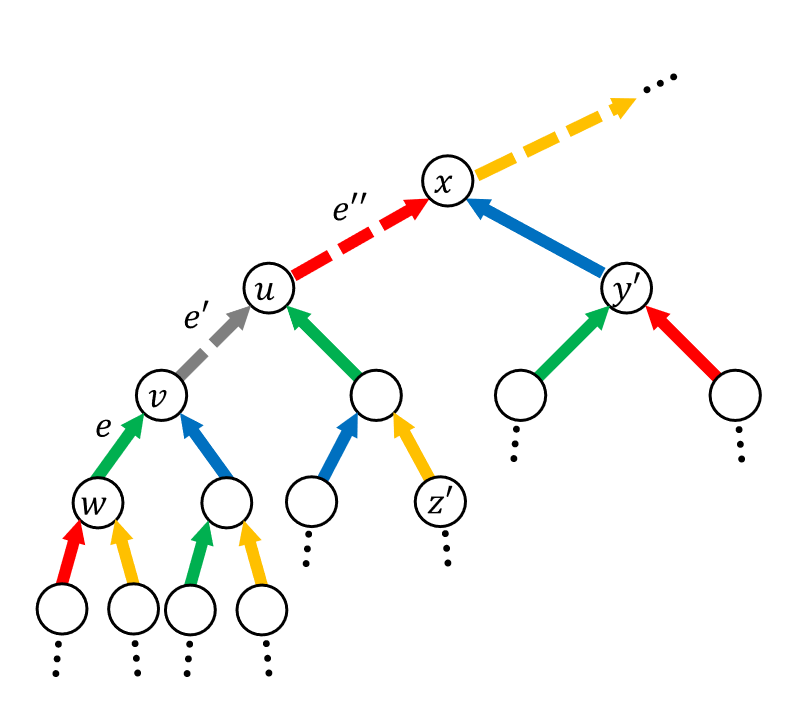}
        \subcaption{}\label{fig:third}
        \end{subfigure}
        \caption{Illustration of the recoloring process of \colorfulpath{}. $\Delta=3$, and arrows indicate the rooting of the tree. Each sub-figure is mid-way of applying \colorfulpath{} and dashed colors are colors that were already shifted upwards. The gray edge is the currently (unique) uncolored edge, and the palette is $\{\text{\color{blue} blue}, \text{ \color{red} red},  \text{ \color{Dandelion} yellow}, \text{ \color{OliveGreen} green}\}$. In (a) a recoloring path reached $e$. Assuming that we cannot stop on $e$ (as in the figure) or on any of its children, it takes not-yellow from its children because $\colorf{}(e'') = \text{yellow}$. In this example $e$ takes green, and in (b) we see it in solid green. The update in (b) happens sometime in the future, after $(x,y)$ was deleted, $(x,y')$ was inserted, $z'$ was attached instead of $z$, and some other updates possibly happened but not around $e'$. Now, when the recoloring path arrives at $e''$ it cannot stop, but it does have a single look-ahead towards $e'$. So $e''$ takes the color of $e'$, and since $\text{yellow} \in A(e')$, in (c) the recoloring can stop on $e'$. We ensure $\text{yellow} \in A(v)$ by not shifting yellow to $e$ in the application of \colorfulpath{} in (a), and yellow got into $A(u)$ by the following application that reaches $e'$ in (c). Overall the effect was that in the second application of \colorfulpath{} we had $\text{yellow} \in A(e')$ so we could stop.} \label{figure_shifted_rooted_e_ready}\vspace{-0.5cm}
    \end{figure}
    
\else
\fi

\iffullversion
\subsubsection{Greedy Logarithmic Amortized Recourse}
\label{section_greedy_high_recourse}

In this sub-section we prove that \greedy{} has amortized recourse that is logarithmic in $n$. Before we can prove Theorem~\ref{theorem_greedy_fails_LB}, we need to build up towards it with multiple lemmas that study the behavior of the \emph{$P$-layered trees} gadget.

\begin{restatable}[$P$-layered tree, $P$-star]{definition}{definitionLayeredTree}
    \label{definition_layered_tree}
    We say that a $\pal$-colored rooted tree $T$ ($\Delta \ge 3$) is \emph{$P$-layered} for a sub-palette $P \subseteq [\pal]$ if: (1) $T$ is full in the following sense: every non-leaf vertex at an even depth has $|P|$ children, and every non-leaf vertex at an odd depth has $|\overline{P}|$ children where $\overline{P}$ is the complementary sub-palette;\footnote{In particular, this means that $|P|+|\overline{P}| = \pal$ and that $|P|,|\overline{P}| \le \Delta-1$; So also, $|P|,|\overline{P}| \ge \extra+1$.} (2) The edge coloring of $T$ is such that the edges from a vertex to its children are colored from $P$ if its depth is even, and from $\overline{P}$ if its depth is odd. See Figure~\ref{figure_layered_tree} for an example. If the tree is of depth $1$ (i.e., a star graph) we refer to it as a \emph{$P$-star}.
\end{restatable}

\begin{figure*}[t]
    \centering
    \includegraphics[width=0.3\textwidth]{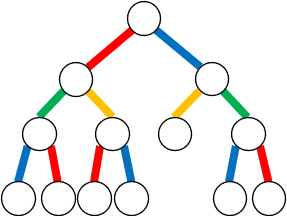}
    \caption{\small{An example for Definition~\ref{definition_layered_tree} of a \emph{layered tree}, for $\Delta=3$, and $\pal = 4$ with sub-palettes $P=\{\text{\color{blue} blue}, \text{ \color{red} red}\}$, and $\overline{P} = \{\text{\color{Dandelion} yellow}, \text{ \color{OliveGreen} green}\}$. In \emph{layered} coloring, the palette used by a vertex on edges towards its children uniquely determine the palette used by its children towards their children (the complement) and vice versa (unless it is a leaf).}}
    \label{figure_layered_tree}
\end{figure*}

We first show that when \greedy{} merges two complementing layered trees, it recolors a path to a leaf $v$, and that the resulting tree is itself a larger layered tree with a different root if we delete the edge of $v$.

\begin{lemma}
\label{lemma_must_recolor_path}
Let $T_1$ be an $P$-layered tree and let $T_2$ be a $\overline{P}$-layered tree. If we insert a new edge between their roots, then a proper coloring must recolor at least a path to some leaf. Recoloring only a path is sufficient.
\end{lemma}

\begin{proof}
First, recoloring a path to a leaf is sufficient. Indeed, color the newly inserted edge by a free color of the root of $T_2$, and uncolor the conflicting edge. As long as there is an uncolored edge, pick its color based on the vertex that belongs to the connected component of colored edges on the side of $T_2$ (which is growing) and uncolor a conflicting edge if such exists. This will push the uncolored edge towards a leaf of $T_1$, where the process will end with a proper coloring.

On the other hand, a path is necessary. Assume by contradiction that we can recolor a subset of edges that do not reach a leaf. Then we can recover the palette-coloring of the tree bottom-up, because given $|P|$ siblings that are all using the palette $P$ towards their parent, the shared parent must use a color from $\overline{P}$ towards its parent, and similarly $|\overline{P}|$ siblings that use $\overline{P}$ towards a shared parent $p$ enforce $p$ to use a color from $P$ towards its parent. For example, revisit Figure~\ref{figure_layered_tree} and consider what the four left most leaves imply for their parents.
Then the color reconstruction gives us a layered coloring of $T_1$ and $T_2$, that results in a contradiction since we cannot color the newly inserted edge.
\end{proof}

The condition of layered-trees in Lemma~\ref{lemma_must_recolor_path} is necessary to force recoloring to reach a leaf. See Figure~\ref{figure_non_layered_recolor} for an example where we can stop earlier than a leaf.

\begin{figure*}[t]
	\centering
    \includegraphics[width=0.7\textwidth]{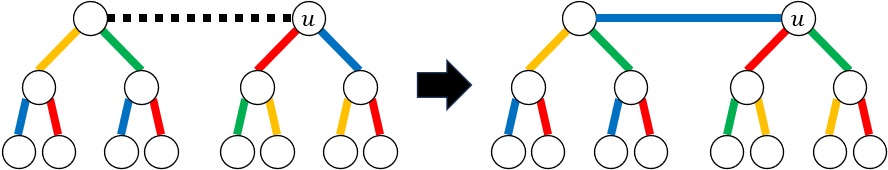}
	\caption{\small{When linking two layered trees of complementing palettes, recoloring must reach a leaf. In this figure, the tree on the right is \emph{not} layered, thus recoloring can be done without reaching a leaf.}}
	\label{figure_non_layered_recolor}
\end{figure*}

\begin{lemma}
\label{lemma_layered_backwards}
Let $T_1$ be a $P$-layered tree and let $T_2$ be a $\overline{P}$-layered tree. Assume that we insert an edge between their roots, and as a result \greedy{} recolored a path towards a leaf $u$, whose parent is $p$. If we delete the edge $(u,p)$, then the new (combined) tree, rooted at $p$, is either $P$-layered or $\overline{P}$-layered depending on the parity of distance to the edge connecting the two original roots.
\end{lemma}

\begin{proof}
\looseness=-1
Without loss of generality, assume that $u$ is a leaf of $T_1$, and denote the nodes on this path as $v_0 = root(T_2), v_1 = root(T_1),v_2\ldots,v_k = u$. See Figure~\ref{figure_merging_layered_tree} for a visualization of the proof. The argument is inductive, we  consider the recoloring of the path as single steps of coloring the uncolored edge and uncoloring its successor. Note that the last step just colors the edge $(u,p)$ to get a proper coloring, so if we end the argument just before coloring $(u,p)$, we indeed argue about the combined tree without the edge $(u,p)$ (as if right after $(u,p)$ is deleted).

To prove the claim notice the invariant that the vertices of the uncolored edge are always roots of complementing-palettes layered trees. This is true initially, for $v_0$ ($\overline{P}$) and $v_1$ ($P$). By definition of the layering, every child of $v_1$ is a $\overline{P}$-layered tree, just like $v_0$. Therefore when we shift a color, from any edge of $v_1$ to the uncolored edge $(v_0,v_1)$, the tree rooted at $v_1$ remains $P$-layered, and the subtree we cut away from it, now rooted at $v_2$, is $\overline{P}$-layered. We can rename the vertices and switch the roles of $P$ and $\overline{P}$, and continue by induction until we reach the end of the path, where eventually $p$ is the root of a $P$-layered or a $\overline{P}$-layered tree that combines all the vertices of $T_1$ and $T_2$ except for $u$.
\end{proof}

\begin{figure*}[t]
	\centering
    \includegraphics[width=0.5\textwidth]{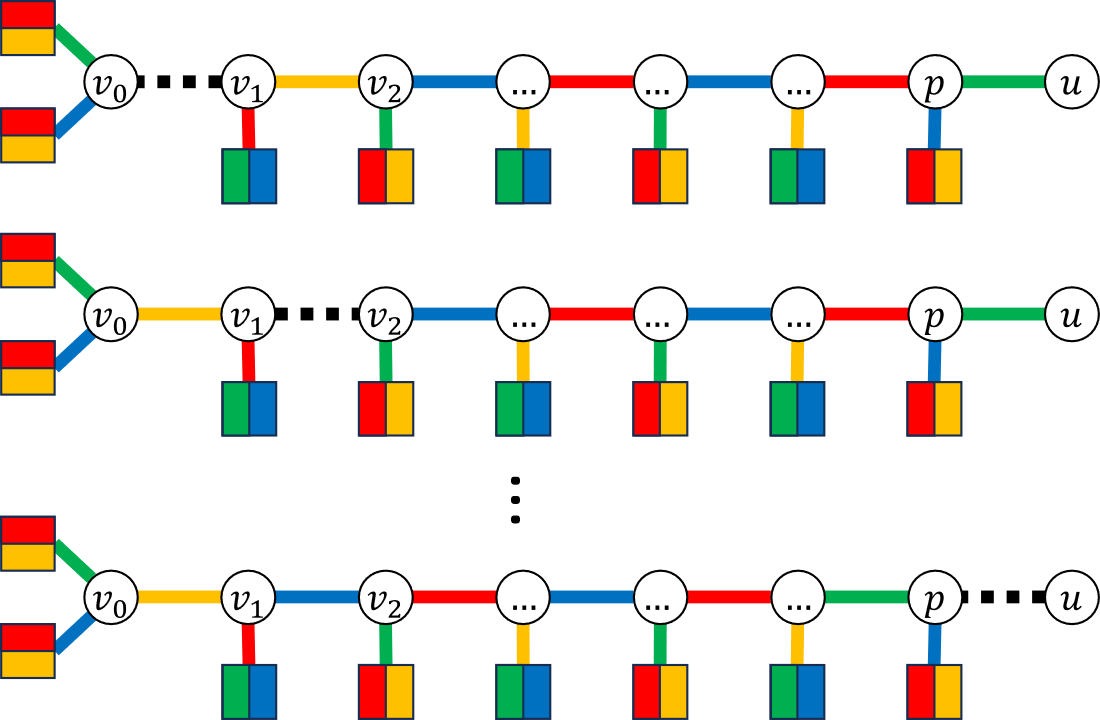}
	\caption{\small{Visualization for Lemma~\ref{lemma_layered_backwards}, for $\Delta=3$, and $\pal = 4$ with sub-palettes $P=\{\text{\color{red} red}, \text{\color{Dandelion} yellow}\}$, and $\overline{P} = \{\text{\color{blue} blue}, \text{ \color{OliveGreen} green}\}$. A square colored by palette $Q \in \{ P,\overline{P}\}$ is an abstract representation of a $Q$-layered subtree. We begin by merging two layered trees rooted in $v_0$ and $v_1$ (top). In the end of recoloring the path from $v_1$ to $u$ (bottom), ignoring the edge $(p,u)$, the combined tree rooted at $p$ is $Q$-layered, where $Q \in \{ P,\overline{P}\}$ depends on the parity of the recourse.}}
	\label{figure_merging_layered_tree}
\end{figure*}

Lemma~\ref{lemma_construct_initial_layered_tree} below constructs layered trees out of the empty graph.

\begin{lemma}
\label{lemma_construct_initial_layered_tree}
Let $\Delta \geq 3$ and $0 \le c \le \Delta-2$. Let $ALG$ be a deterministic algorithm that does not recolor edges unless it is necessary. There exists $n_1 = n_1(\Delta,\extra)$ such that for every 
$n \ge n_1(\Delta,\extra)$ the adversary can construct multiple $P$-layered and $\overline{P}$-layered trees, for $|P| = \extra+1$ (then $|\overline{P}| = \Delta-1$), with  $\Omega(n/\Delta)$ vertices. ($n_1$ is studied in details in Section~\ref{section_dependence_of_n}.)
\end{lemma}

\begin{proof}
We play the role of the adversary. Assume that we have $|P| \cdot |\overline{P}|$ $P$-stars, then we can use them to \emph{replicate} additional $P$-stars as follows: Connect sets of $|\overline{P}|$ stars to a new vertex. No recoloring is necessary, so by assumption none is made, and we get $|P|$ $\overline{P}$-layered trees of depth $2$. Then connect all of these trees to another new vertex $v$, to get a $P$-layered tree of depth $3$. Remove the edges between the original stars and their parents, and we end up with the original $P$-stars, plus a $P$-star centered at $v$. We can also replicate a $\overline{P}$-star out of the existing $P$-stars, by replicating $|\overline{P}|$ ``disposable'' $P$-stars, forming a $\overline{P}$-layered tree of depth $2$, and then deleting its leaf-edges.

Observe that for a sufficiently large $n$, we can bootstrap the replication process by forming sufficiently many stars such that by the pigeonhole principle we get enough stars with some same sub-palette. Once we have $|P| \cdot |\overline{P}|$ $P$-stars we can replicate more $P$-stars and then construct, bottom-up, large \emph{balanced} $P$-layered and $\overline{P}$-layered trees using an $\Omega(\frac{1}{\Delta})$ fraction of them: To increase the depth of a tree by $1$, we connect $|\overline{X}|$ $X$-layered trees to a single shared parent for $X \in \{P,\overline{P} \}$, to form a $\overline{X}$-layered tree. Increasing the depth requires additional vertices, but naively we can dedicate an extra vertex to each $P$-star, and only lose a constant factor because a $P$-star has $|P|+1 = \extra+2$ vertices. We can utilize at least a $\frac{1}{\Delta-1}$ fraction of the stars to construct the desired trees (otherwise, we can increase the depth of each tree by $1$). Accordingly, the construction of our trees utilizes $\Omega(n/\Delta)$ vertices.
\end{proof}

We are ready to prove Theorem~\ref{theorem_greedy_fails_LB}. We restate it for convenience.

\theoremGreedyRecourse*

\begin{proof}
Let the initial state be as follows: Two perfect layered trees $T^*_i$ of depth $d = \Theta(\log_\Delta n)$ (every leaf has depth $d$) with roots $r_i$, for $i \in [2]$. $T^*_1$ is $P$-layered, while $T^*_2$ is $\overline{P}$-layered, for $|P|=\extra+1$ and $|\overline{P}| = \Delta-1$. Assuming that the algorithm's tie-breaking is deterministic (per the setting of this section), we can construct this initial state by Lemma~\ref{lemma_construct_initial_layered_tree}.

Now the cycle of updates begins. The initial two updates delete one edge from each tree. Denote the edges $(u_1,r_4)$ in $T^*_1$ and $(u_2,r_3)$ in $T^*_2$, where in both edges $u$ is the parent of $r$. The distance between $r_i$ and $u_i$ for $i \in [2]$ is $d_i = \floor{\frac{d}{3}} - (i-1)$, this exact choice will be clarified later. Now the forest contains $4$ trees $T_i$ each rooted at $r_i$ for $i \in [4]$.

The next update inserts $(r_1,r_2)$. By Lemma~\ref{lemma_must_recolor_path} \greedy{} recolors a path. Since $d_2 < d_1$, the smallest recourse would recolor over the path from $r_2$ to $u_2$. Then, because $u_2$ has one less child, the algorithm can finish the recoloring, with recourse $d_2$. To see clearly why recoloring can stop at $u_2$, imagine that $u_2$ had another leaf child $u'$ (instead of the former child $r_3$), for $T_2$ to be a proper layered tree. In this case, the recolored path to a leaf would be to $u'$ and include $(u_2,u')$, since any other leaf of $T_2$ is deeper. The color of $(u_2,u')$ would have been exactly the color of the edge $(u_2,r_3)$ that was deleted, and this color is used to color $(\parent(u_2),u_2)$ when the recoloring is applied. Then by Lemma~\ref{lemma_layered_backwards} the new combined tree of $T_1 \cup T_2$ can be thought of as a tree $T'_2$ rooted at $u_2$, such that $T'_2$ is $\overline{Q}$-layered compared to $T_3$ which is $Q$-layered, where $Q \in \{ P , \overline{P} \}$ depends on the parity of $d_i$ (in Figure~\ref{figure_adversarial_cycle_of_requests}, $Q=P$).

Next, insert the edge $(u_2,r_3)$, which again connects two layered trees with complementing palettes. Then again \greedy{} will recolor a path. This time, we want the path to go back through the edge $(r_1,r_2)$ and towards $u_1$. This is satisfies if $d_1+d_2 < (d-d_2)$ (otherwise the shortest recourse goes towards a leaf of $T_3$), which holds by how we defined $d_1$ and $d_2$. By Lemma~\ref{lemma_layered_backwards}, $T_1 \cup T_2 \cup T_3$ is now a $Q'$-layered tree rooted at $u_1$ where $Q' \in \{ P , \overline{P} \}$, and by a similar argument as before, we have that $T_4$ is $\overline{Q'}$-layered (in Figure~\ref{figure_adversarial_cycle_of_requests}, $Q'=P$).

The next insertion would be of $(u_1,r_4)$, but first it is important to delete $(r_1,r_2)$. Note that this deletion recovers $T^*_2$ of the initial state. Now we insert $(u_1,r_4)$, which has the same effect as with the previous insertion: the edge connects $T_1$ that is $Q'$-layered when viewed as rooted at $u_1$, and $T_4$ that is $\overline{Q'}$-layered. To ensure that smallest recourse would be to recolor on a path from $u_1$ to $r_1$, because any other leaf in $T_1 \cup T_4$ is farther from $u_1$, we ensured that $d_1 < (d-d_1)$. (Recall that $r_1$ is not a leaf, but has one less child than its siblings, which is why the recoloring can stop at $r_1$.) By similar arguments as before the recoloring stops such that $T_1 \cup T_4$ becomes a $P$-layered tree when rooted at $r_1$, which is exactly $T^*_1$ of the initial state.

In total, the recourse paid during the cycle was $d_2 + (d_2 + 1 + d_1) + d_1$ subject to the constraints that $d_2 < d_1$, $d_1+1+d_2 < (d-d_2)$ and $d_1 < (d-d_1)$. These constraints are why we chose $d_1 = \floor{\frac{d}{3}}$ and $d_2 = d_1 - 1$. Overall, we get an amortized recourse of $\Omega(d)$. Given $n$ vertices, we can have $d = \Omega(\log_\Delta n)$. To clarify $d$, note that the expansion rate of the layered trees is $(\Delta-1) \cdot (\extra+1)$ per two layers, which is at least $\Delta - 1$ and at most $(\Delta - 1)^2$, per two layers. The adversary can force \greedy{} to reach this state from the empty graph by Lemma~\ref{lemma_construct_initial_layered_tree}, subject to $n \ge n_1(\Delta,\extra)$.
\end{proof}

\begin{remark}[Laziness]
\label{remark_greedy_worst_case}
The proof of Theorem~\ref{theorem_greedy_fails_LB} also implies that any \emph{lazy} algorithm that does not incur recourse when it does not have to, has a worst-case recourse of $\Omega(\log_\Delta n)$ because it allows the adversary to construct the bad configuration. Note that unlike \greedy{}, an arbitrary lazy algorithm might not be forced to maintain this state once it must pay some recourse, therefore the lower bound applies only in worst-case and not amortized. A non-lazy algorithm can theoretically amortize its recourse by distributing it gradually over multiple updates to achieve an overall reduced worst-case recourse.
\end{remark}

\fi

\iffullversion
\subsubsection{\texorpdfstring{$O(1)$}{O(1)} Recourse with \texorpdfstring{$2\Delta-2$}{2Delta-2} Colors on \emph{Rooted} Forests}
\label{subsection_rooted_TwoDelta}

In this section we prove Theorem~\ref{theorem_rooted_trees_recourse_constant_UB}. Note that in this context, we only consider $\extra = \Delta-2$, and only consider a \emph{rooted} dynamic forest (Definition~\ref{definition_rooted_forest}). In this setting we show that \colorfulpath{} (Algorithm~\ref{algorithm_rooted_forest}) has amortized recourse of $O(1)$, which is independent of $n$, ``breaking'' the $\Omega(\log_{\Delta} n)$ lower bounds of \cite{sadeh2026vizingchainsimprovedrecourse} thus showing the importance of amortization,\footnote{Theorem~2.9 in \cite{sadeh2026vizingchainsimprovedrecourse}, when we substitute arboricity $\alpha=1$. This is lower bound that applies for any $\extra \in [0,\Delta-2]$, but it applies in worst-case, for a given graph with a given edge coloring.} and the lower bound of Theorem~\ref{theorem_trees_delta_colors_LB} thus showing the importance of extra colors.

Intuitively, extra colors break ``long-distance correlations''.\footnote{For example, think of a $2$ edge colored path where the colors of every pair of edges are correlated, but with $3$ colors we cannot say as much on the correlation of the colors of non-adjacent edges.} This also shows that an averaged-analysis helps, whether as expected recourse or amortization (or both).

Intuitively speaking, \colorfulpath{} shifts colors over a path originating at the root of one of the trees that the new edge links, downwards until it can stop. The path is made ``as farthest as possible'' from being bicolored in the following sense: In a bicolored path an edge and its grandparent-edge have the same color, and the algorithm explicitly refrain from this. This can be leveraged to show a low amortized recourse.

\begin{definition}[Step]
A \emph{step} is a single coloring of an uncolored edge, possibly uncoloring one of its neighbours. For example, the function \emph{RecursiveStep} in Algorithm~\ref{algorithm_rooted_forest} implements a step. We say that a step \emph{stops} on an edge if it does not recursively call another step, otherwise it \emph{moves} to the next edge we call it on. A step on edge $(u,v)$ exits from $v$ if it moves to another edge $(v,x)$ ($x \ne u$). Then we also say that it enters from $u$, even if $(u,v)$ was just inserted.
\end{definition}

A step either colors a new edge or recolors an edge (that was previously colored). Recall that we defined the recourse as the number of edges that we \emph{recolor}. So in this section we bound the number of steps that recolor edges. The number of steps that color new edges is at most $T$ where $T$ is the number of updates.

\begin{observation}
\label{observation_topology_and_color_independent}
Assume $2\Delta-2$ colors. Consider a coloring step on $(u,v)$ that enters from $u$. Denote the set of colors used in $v$ by $\overline{A}(v)$. Then the step does not stop if and only if $\deg(u)=\deg(v)=\Delta$ and $A(u) = \overline{A}(v)$.

A neat implication of having $\Delta+\extra = 2\Delta-2$ colors is that the whole palette is divided to two disjoint sets of size $\Delta-1$, one for the even and one for the odd numbered steps.
\end{observation}

We restate Theorem~\ref{theorem_rooted_trees_recourse_constant_UB} for convenience, and then prove it.
\theoremRootedTreesDynamicTwoDeltaDeterministicConstant*

\begin{proof}
In simple words, Algorithm~\ref{algorithm_rooted_forest} recolors in steps that lead on a path towards (eventually) a leaf, so it must stop successfully. Each step can be out of three types:
\begin{enumerate}
    \item \label{type1} If a step on edge $e$ can choose an available color for it, do it and stop (first \emph{if} in Algorithm~\ref{algorithm_rooted_forest}).
    \item \label{type2} Otherwise, try to step towards a neighbour $e'$ that has an available color (second \emph{if} in Algorithm~\ref{algorithm_rooted_forest}). By Observation~\ref{observation_topology_and_color_independent}, if $e$ has no available color, we are guaranteed to be able to step to such $e'$ if there is one (it is not guaranteed if $\extra < \Delta-2$). \item \label{type3} Otherwise (and finally), if no such $e'$ exists, just step to $e''$ such that the color we pick for $e$ does not match the color we picked $2$ steps earlier, to ensure that we are not shifting colors over a bicolored path (\emph{else} in Algorithm~\ref{algorithm_rooted_forest}). We can always pick a color for $e$ that is different than $2$ steps before, because we have $\extra+1 \ge 2$ options and at most one color is not allowed. We proceed to analyze the amortized recourse.
\end{enumerate}

Let $v_0,v_1,v_2,v_3,v_4$ such that $v_i = \parent(v_{i+1})$ and denote $e_i = (v_i,v_{i+1})$ for $i \in \{0,1,2,3\}$. We also assume that the edges $e_1,e_2,e_3$ have an initial color $\col{1}_i = \colorf(e_i)$ ($e_0$ may be a newly inserted edge). Consider an update that causes steps sequentially on $e_0$ through $e_3$. If the step on $e_2$ is of type~(\ref{type3}), we say that $e_1$ becomes \emph{ready}. Intuitively, it is ready because the next step on $e_1$ should be able to stop. To see why, notice that the step on $e_0$ freed $\col{1}_1$ from $e_1$ so now $\col{1}_1 \in A(v_2)$. The step on $e_1$ did not change $A(v_2)$, and the step on $e_2$ ensured (by definition of type~(\ref{type3})) that $\col{1}_1 \in A(v_2)$. If another step arrives at $e_1$ from $e_0$ then it means that $\col{1}_1$ is no longer used at $v_1$ (was either deleted from $e_0$ or shifted upward in the tree), therefore $\col{1}_1 \in A(v_1) \cap A(v_2)$ so we can stop the step on $e_1$ by coloring it with $\col{1}_1$. Note that we relied on the fact that the forest is rooted, for the consistency of the recoloring path arriving to $e_1$ from the parent $e_0$. A ready edge does not remain ready forever. We say that an edge becomes \emph{unready} if one of the following occurs:
\begin{enumerate}
    \item An edge $e$ becomes unready when a step is applied on a sibling-edge $e'$ of $e$. To clarify why, note that $e_1$ above became ready because the next shift of colors frees $\col{1}_1$ in $v_1$, but a step on a sibling $e' = (v_1,v'_2)$ of $e_1$ might color $e'$ by $\col{1}_1$, ruining the readiness of $e_1$.
    \item When a step stops on $e$, it makes $e$, the neighbour edges of $e$, and the grandparent edge of $e$ (if it exists) unready. To clarify why a grandparent edge is also affected, notice that above, in order for $e_1$ to become ready we required the step on $e_2$ to be of type~(\ref{type3}). However if the step on $e_3$ stops it might be because the step on $e_2$ stepped towards $e_3$ on purpose, type~(\ref{type2}). So we are being conservative by assuming that the grandparent of $e$ becomes unready when a step stops on $e$.
    
    \item A ready edge may be deleted, and newly inserted edges are always unready. When an edge is inserted, also all of its neighbour-edges, and its grandparent-edge (if it exists), become unready.
\end{enumerate}
Observe that by definition, there can be at most one ready edge between a vertex $u$ to its children. To see why, consider a path of steps that arrives at $u$, such that $v \in \ch(u)$ and $(u,v)$ is ready. Either the step on $(\parent(u),u)$ stops thereby making all the edges of $u$ unready, or the algorithm proceeds by a type~(\ref{type2}) step towards $(u,v)$ or to another edge of $u$ over which the next step stops, again making all the edges of $u$ unready. 

Let $\Phi \ge 0$ be a potential that counts the number of unready edges. $\Phi(t=0) = 0$ because the graph is initially empty. We claim that the amortized recourse is $O(1)$. For deletion of $e$, there is no recourse, and $\Phi$ also does not increase, so the amortized recourse $am \le \Phi_{t+1} - \Phi_{t} + 0 \le 0$. For insertion, we increase the potential by at most $O(1)$ due to the newly inserted edge, and by another $O(1)$ due to the vicinity of the edge we stop on. Each $O(1)$ is at most $4$ 
because there is at most one ready edge among the child-edges of $e$, at most one among sibling-edges of $e$, and the remaining possible $2$ are due to its parent-edge and grandparent-edge. Finally, say that the recourse is $\ell$, on a path over the vertices $u,v_0,v_1,v_2,\ldots,v_\ell$ where $(u,v_0)$ is the newly inserted edge. Then every edge on this path becomes ready except for the newly inserted edge and the last three (the stopping edge, its parent, and its grandparent). On the other hand, for every $i \in [0,\ell-2]$ and $y \ne v_{i+1}$ the edge $e = (v_i,y)$ (if it exists) was not ready before shifting colors, so we did not change ready edges to unready. Indeed, the algorithm did not stop on $(v_i,v_{i+1})$, which means that $v_{i+1}$ was chosen over $y$ by a type~(\ref{type3}) step. However, if $e$ was ready then the algorithm would have picked $y$ instead of $v_{i+1}$ and stop on $e$ (a type~(\ref{type2}) step). If some $(v_{\ell-1},y)$ was ready and becomes unready, this was already charged to the stopping step on $(v_{\ell-1},v_\ell)$. In conclusion, the potential decreases by $\ell - O(1)$, and therefore the amortized recourse of insertion is also $O(1)$. Explicitly, $am \le \Phi_{t+1} - \Phi_{t} + \ell \le (O(1) - \ell) + \ell = O(1)$.
\end{proof}

Theorem~\ref{theorem_rooted_trees_recourse_constant_UB} fails for unrooted trees because when a step enters $(u,v)$ from $u$ and exits through $v$, we are not guaranteed that the next step on $(u,v)$ also enters $(u,v)$ from $u$. This issue can be slightly mitigated if the algorithm directs edges and attempts in some clever way to pass through them only according to their orientation.\footnote{As an example, say that we flip an unbiased coin to decide the orientation of an edge when it is inserted, and that when a step exits via a vertex $v$ we only choose available colors based on outgoing edges of $v$, unless there is at most $1$ so we pick among all of them (to maintain a random choice). Then the probability that we exit $v$ on an edge directed towards $v$ is exponentially small in $\Delta$: Only if $\deg(v)$ is large (at least $\extra+1$, or else there should be an available color to stop), and if at least $\deg(v)-2$ edges of $v$ are oriented towards $v$.} This still does not truly help because our argument actually requires that we enter $(u,v)$ consistently from the same edge of $u$ and not just from any edge of $u$. Furthermore, if we restrict the options of where the next step may proceed from one edge to the next, we introduce a difficulty as in the case $\extra < \Delta - 2$, and then we cannot ensure that a step can proceed to a neighbour edge with a free color. To demonstrate why orientation might be crucial for designing an algorithm with low recourse, recall Theorem~\ref{theorem_greedy_fails_LB} and Section~\ref{section_greedy_high_recourse} where the adversary causes the greedy algorithm to recolor the exact same paths repeatedly in opposite directions.

\fi

\iffullversion
\subsection{Greedy Algorithms: Additional Variants and Discussion}
\label{section_appendix_greedy_main}

So far most of our discussion on greedy algorithms regarded only \greedy{}. This is not the only greedy variant one can think of. It is the most natural definition when considering recourse minimization per update, but in the realm of edge coloring many algorithms restrict their generality and only focus on shifting colors along paths, or along chains of edges. Indeed, any \emph{(multi-) Vizing chain} which is the base of many edge recoloring algorithms~\cite{Vizing1964,DynamicEdgeColoring2019DuanEtal,Bernshteyn2022SmallRecoloring,MultistepVizing2023,2025LinearEdgeColoring}, is such algorithm. Other algorithms also rely on paths and chains~\cite{ArboricitySWAT2024_sayan,ArboricitySWAT2024_Christiansen,sadeh2026vizingchainsimprovedrecourse}. With this perspective in mind, we can define two additional greedy variants. Then in the remainder of this section we explain that each variant is simple (polynomial) on forests, and detail how the results we derived differ for the new variants.

\begin{definition}
\label{definition_greedy_algorithms}
Given a forest $F$ with a proper edge coloring except for one edge $e \in F$, we define three greedy recoloring algorithms (variants) as follows. 
\begin{enumerate}
    \item \greedy{} (re-stated for completeness): This algorithm determines and recolors a smallest set of edges $H \subseteq F \setminus \{ e \}$ needed to be recolored to extend the coloring to $e$. Note that $H \cup \{ e \}$ is a connected component (single, by minimality).
    
    \item \greedyshift{}: This algorithm picks the minimal $H$ from a special subset: iteratively, let $e' = (u,v)$ be the current uncolored edge (initially $e'=e$). Pick a color for $e'$ that is available at either $u$ or $v$ or both. If both, then we are done. Otherwise, uncolor the conflicting edge $e''$, and re-iterate on $e''$.
    
    \item We can define a sub-type of \greedyshift{} which is \greedypath{}, such that the chain of edges over which we shift colors must form a path (that is, we do not allow \emph{fans}).
\end{enumerate}
\end{definition}

Observe that the greedy algorithms in Definition~\ref{definition_greedy_algorithms} are listed in non-decreasing order of computation simplicity and recourse for a given (same) state. It is simpler to compute more restricted options, and more options can only reduce the recourse. Furthermore, each algorithm is in fact a class of algorithms, because we did not define the exact behavior in case of tie breaks. Such tie breaks are always necessary, even when the minimal recourse is unique, because there may be multiple recoloring options. The greedy hierarchy does not immediately imply a hierarchy in terms of amortized recourse, not even in worst-case when starting from the empty graph: Hypothetically, perhaps some clever invariant holds for one version but not for another.

\subsubsection{Greedy Optimality for \texorpdfstring{$\Delta=2$}{Delta=2} }
\label{section_Delta_2}

The case of $\Delta=2$ behaves quite differently than $\Delta \ge 3$. Intuitively, such trees are paths, and their depth is linear in their size rather than logarithmic. Moreover, any connected component is a path, so \greedy{}, \greedyshift{} and \greedypath{} are the same. We briefly state optimal bounds in this mostly trivial case.

\begin{theorem}
\label{theorem_Delta_is_2}
Let $F$ by a dynamic forest with maximum degree $\Delta = 2$, and assume $\Delta$ edge coloring. \greedy{} (Algorithm~\ref{algorithm_greedy}) has $O(\log n)$ amortized recourse in the incremental model, and $O(n)$ amortized recourse in the fully dynamic model, which are both optimal up to a constant factor (in particular, randomization does not help).
\end{theorem}

\begin{proof}
\looseness=-1
Consider the fully dynamic model first. Since the forest is a collection of paths ($\Delta=2$), the worst-case recourse, and in particular the greedy recourse, is $O(n)$. Theorem~\ref{theorem_trees_delta_colors_LB} proves amortized recourse of $\Omega(h)$ where $h = \Theta(\log_{\Delta-1} n)$ for $\Delta \ge 3$, by constructing $(\Delta-1)$-trees of depth $h$. Repeating the proof verbatim for $\Delta=2$, the only difference is that now $h = \Theta(n)$ because the trees are unary. This concludes the analysis in the fully dynamic model.

Next consider the incremental model. When an edge is inserted, it links two paths of lengths $d \le d'$. If recourse is required, \greedy{} pays $d$ which we charge equally to the shorter path ($1$ per edge). For a fixed edge $e$, notice that the path that contains it at least doubles in size when we charge $e$, therefore the total charge of $e$ is bounded by $\lg n$. Then the amortized recourse is $O(\log n)$.

To show a matching lower bound against any algorithm, including a randomized algorithm, observe that when a path has an even length, its endpoints have different free colors. Therefore, given an even path with endpoints $u$ and $v$, when we link it to another path with endpoint to $w$, regardless of the available color at $w$, if we insert the edge $(u,w)$ or the edge $(v,w)$ with equal probability, there is probability of exactly $\frac{1}{2}$ that recourse will be incurred. The recourse must recolor one of the paths, so it is best for the adversary to link paths of comparable lengths. If the algorithm is deterministic, we can guarantee the recourse by choosing $(u,w)$ or $(v,w)$, but henceforth assume a randomized algorithm (harder to fail).

Therefore, the adversary acts as follows. Initially, it inserts edges to form paths of length $1$. Then, it links paths of length $1$ by first extending one of them by an edge so that its length is even ($2$), and then adds another edge, so we get paths of length $4$. Then, pairs of paths of length $4$ are linked to form paths of length $9$. The process continues, such that whenever each path in the collection has even length $d$ we merge pairs to get paths of length $2d+1$ and an expected recourse of (at least) $d$, and whenever each path in the collection has odd length $d$ we extend one path in each pair to length $d+1$, and then merge the pair to form a path of length $2d+2$ with expected recourse of (at least) $d$. The lengths satisfy the recursive relation: $d_1 = 1$ and $d_{i+1} = 2 \cdot d_i + 1 + (d_i \mod 2)$, whose closed form is $d_{i} = \frac{4 \cdot 2^{i} - 4 - (i \mod 2)}{3}$.\footnote{The recursion flips the parity with each element, and one can verify that the odd indices satisfy $d_{2i+1} = 4 d_{2i-1} + 5$ and that the even indices satisfy $d_{2i+1} = 4 d_{2i-1} + 4$, to solve explicitly and separately and get $d_{2i-1} = \frac{4 \cdot 2^{2i-1} - 5}{3}$ and $d_{2i} = \frac{4 \cdot 2^{2i} - 4}{3}$.} In each level $i$ the number of paths is halved, until it is a single long path in the last level $\ell$ where $d_\ell = \Theta(n)$ (by maximality of $\ell$), which implies (by the closed form) that $2^\ell = \Theta(n)$. Then the total expected recourse is at least: $\sum_{i=1}^{\ell}{2^{\ell - i} \cdot d_i}
\ge 
\sum_{i=1}^{\ell}{2^{\ell - i} \cdot 2^{i-1}}
=
\ell \cdot 2^{\ell - 1} = \Theta(n \log n)$. Therefore, the amortized expected recourse of any algorithm is $\Omega(\log n)$.
\end{proof}

\subsubsection{What Changes for the new Greedy Variants}
\label{section_appendix_greedy_variants_differences}
Interestingly, most of the results that we derived for \greedy{} apply to \greedyshift{} and its sub-variant \greedypath{}. We list them briefly, with pointers to what change or doesn't change:
\begin{enumerate}
    \item Theorem~\ref{theorem_amortized_insertion_only_improved_UB}: This theorem applies verbatim for $\extra \ge 1$. To see why, note that the proof relies in this case on Theorem~\ref{theorem_logCp1_recourse_shiftable_UB} which applies to recoloring a path (which is a special case of a chain). The case $\extra=0$ is completely different, as shown in Theorem~\ref{theorem_insertion_only_Czero_high_LB} below.
    
    \item Theorem~\ref{theorem_insert_only_LB} is unchanged. Each recourse in the proof is at most $1$, which is always a path.

    \item Theorem~\ref{theorem_greedy_fails_LB} is unchanged. The recourse in the proof is already due to recoloring a path (by Lemma~\ref{lemma_must_recolor_path}). However, we can strengthen the lower bound against these variants even further from $\Omega(\log_{\Delta} n)$ to $\Omega(\log_{\extra+1} \frac{n}{\Delta-\extra})$ (and even $\Omega(n/\Delta)$ if $\extra = 0$), as shown in Theorem~\ref{theorem_greedy_fails_LB_shift_based} below. This improved lower bound is tight for $\extra \ge 1$ since it matches the upper bound in Theorem~\ref{theorem_logCp1_recourse_shiftable_UB} (up to a constant factor).
\end{enumerate}

\begin{theorem}
\label{theorem_insertion_only_Czero_high_LB}
Let $\Delta \geq 3$ and $\extra=0$. Any \emph{shift-based} deterministic algorithm has amortized recourse of $\Omega(\frac{1}{\Delta} \log(\frac{n}{\Delta^3}))$ even in the incremental model.
\end{theorem}

\begin{proof}
Let the adversary start by dividing the $n$ vertices to sets of $\Delta-1$ vertices, and forming a star of degree $\Delta-2$ out of each set.
The algorithm must color each star using $\Delta-2$ different colors. There are $B = \binom{\Delta}{2} = \Theta(\Delta^2)$ different sub-palettes of $\Delta-2$ colors, denote by $k_i$ for $i \in \{ 1,\ldots,B \}$ the number of stars that use sub-palette $i$. Since each star contains $\Delta-1$ vertices, we have $\sum_{i=1}^{B}{k_i} = \floor{\frac{n}{\Delta-1}} \equiv K = \Theta(\frac{n}{\Delta})$. Next, the adversary acts on each group of same-palette stars separately. It will only link the  centers of the stars, and since the algorithm only shifts colors, the problem is reduced to maintaining a proper $2$-coloring of these newly inserted edges.\footnote{Recall that a shift-based algorithm never colors an edge by a color that is used at both of its ends, because this would cause two edges to become uncolored. In each subset of stars, all the colors except for $2$ are used by both endpoints of the uncolored edge: Initially by definition of the insertion, and then by induction over the shifting of colors.} 

For $2$-coloring, we have shown in the proof of Theorem~\ref{theorem_Delta_is_2} that the total incremental recourse is $\Omega(n \log n)$ for $n$ vertices. Then for each reduced subproblem with $k_i$ vertices, the total recourse is $\Omega(k_i \log k_i)$. By applying the adversarial insertions for each group of stars separately, we get a total recourse of $\Omega(\sum_{i=1}^{B} {k_i \log k_i})$. By convexity of the function $f(x) = x \log x$ and Jensen's inequality, the sum is minimized when the parts are equal, so $\sum_{i=1}^{B} (k_i \log k_i) \ge B \cdot (K/B) \log(K/B) = K \log(K/B)$. Finally, substitute $K = \Theta(\frac{n}{\Delta})$ and $B = \Theta(\Delta^2)$ to get $\Omega(\frac{n}{\Delta} \log(\frac{n}{\Delta^3}))$ total recourse, and amortize it over $\Theta(n)$ insertions. 

Recall that this bound applies only if the algorithm only shifts colors. Other algorithms \emph{could} recolor the stars, to drastically reduce the recourse when merging long paths.
\end{proof}

Next we show that the bound from Theorem~\ref{theorem_insertion_only_Czero_high_LB} is almost tight for \greedyshift{} and \greedypath{}.

\begin{theorem}
\label{theorem_insertion_only_Czero_high_shift_UB}
Assume only insertions to a forest with maximum degree $\Delta$. Given $\Delta$ colors, \greedyshift{} and \greedypath{} each have amortized recourse $O(\frac{1}{\Delta} \log n)$.
\end{theorem}

\begin{proof}
When an edge is inserted between two trees, we charge the recourse equally to each edge of the smaller tree. Each edge is charged $O(\log n)$ times because the size of its tree at least doubles each time we charge it. Thus, to complete the proof we show that the charge due to a single merge is $O(\frac{1}{\Delta})$ per edge.

Instead of analyzing the charge due to recourse of \greedyshift{} or \greedypath{}, it would be easier to use an alternative recourse that bounds them. We define a recolored path as follows: its trajectory is towards a leaf of the smaller tree, and in each single color shift, it moves further away from the root, in the direction of a smallest subtree option. This is clearly a path, and it is guaranteed to finish successfully at a leaf (possibly earlier), therefore the recourse is an upper bound to the recourse of both \greedyshift{} and \greedypath{}.

\looseness=-1
Consider an insertion, and say that $(u,v)$ is uncolored such that the recoloring path proceeds to the tree of $v$ whose size in edges is $E(v)$ (either $(u,v)$ is the new edge, or we are in the process of recoloring the path). The algorithm now chooses a color for $(u,v)$. Unless $A(u) \cap A(v) \neq \emptyset$, the algorithm chooses a color from $A(u)$ which, since $A(u) \cap A(v) = \emptyset$, belongs to some edge $(v,w)$. $|A(u)| = \Delta - (\deg(u) - 1) \ge 1$. Since we choose the smallest subtree option, we get a recursive function $R$ that upper bounds the recourse: $R(E(v)) \le 1 + R(E(w))$ where $E(w) \le \floor{\frac{E(v) - (\deg(v) - 1)}{A(u)}}$. Note that $R$ is monotone because it represents the worst possible recourse given a fixed number of edges. If $z$ is the next vertex after $w$, then:
$E(z) \le
\frac{E(w) - (\deg(w) - 1)}{A(v)}
< 
\frac{E(v) - (\deg(v) - 1)}{A(v)}
= 
\frac{E(v) - (\deg(v) - 1 + A(v)) + A(v)}{A(v)}
=
\frac{E(v) - \Delta}{A(v)} + 1
<
E(v) - \Delta + 1
$.
For every $2$ steps we remove at least $\Delta-1$ edges, thus $R(E(v)) \le 2 \cdot \frac{E(v)}{\Delta-1}$, or, $\frac{2}{\Delta-1}$ per edge.
\end{proof}

We now state a stronger version of Theorem~\ref{theorem_greedy_fails_LB}, against the new greedy variants.

\begin{restatable}[Shift-based Greedy recourse]{theorem}{theoremGreedyRecourseShiftsVersion}
\label{theorem_greedy_fails_LB_shift_based}
Let $\Delta \geq 3$ and $0 \le c \le \Delta-2$. There exist a forest and a coloring such that from this state onward, the amortized recourse of \greedyshift{} and \greedypath{} is $\Omega(\log_{\extra+1} \frac{n}{\Delta-\extra})$ for $\extra \ge 1$ and $\Omega(n/\Delta)$ for $\extra = 0$. There exists $n_2 = n_2(\Delta,\extra)$ such that for every $n \ge n_2$, this bad state can be reached from the empty graph.
\end{restatable}

\begin{remark*}
The value of $n_2$ is studied in detail in Section~\ref{section_dependence_of_n}.
\end{remark*}

\begin{proof}
The proof reduces the case of maximum degree $\Delta$ and extra $\extra$ colors to a problem where the maximum degree is $\Delta' = \extra+2$ colors, such that \greedypath{} and \greedyshift{} behave like \greedy{} on the reduced problem.

To reduce the problem, we need to make some colors unusable. The reduction is similar to other cases where we make colors unusable: We form stars such that the actual dynamic updates are only on edges between centers. Assuming that all the stars use the same sub-palette $P$, the colors of $P$ can never be used by a shift-based algorithm. Each star is of degree $\Delta-(\extra+2)$, has $\Delta-\extra-1$ vertices, so we have $n' = \floor{\frac{n}{\Delta-\extra-1}}$ stars in total.

Provided that we are given these stars to begin with, or that $n \ge n_2(\Delta,\extra)$ so that we can bootstrap them from the empty graph, we now have $n'$ vertices (star centers), which can have a maximum degree $\Delta' = \extra+2$ in the reduced problem. The number of extra colors is $\extra = \Delta'-2$. Finally, if $\extra \ge 1$ then $\Delta' \ge 3$ and by Theorem~\ref{theorem_greedy_fails_LB} we conclude a lower bound of $\Omega(\log_{\Delta' - 1} {n'}) = \Omega(\log_{C+1} {n'})$. If $\extra = 0$ then $\Delta'=2$, and by Theorem~\ref{theorem_Delta_is_2} we conclude a lower bound of $\Omega(n')$.
\end{proof}

\fi

\iffullversion
\section{Randomized Algorithms: A Formal Account}
\label{section_randomized_analysis}
\else
\subsection{Randomized Algorithms}
\label{section_randomized_analysis_sketch}
\fi

We now turn to randomization. 
We consider a simple randomized approach that yields better bounds in both models (incremental and fully dynamic). The basic idea is to keep the edge-coloring highly random as the forest evolves over time.

Concretely, for a forest $F$, we define a simple distribution $\mathcal D (F)$ over the proper $\pal$-edge-coloring of $F$. The algorithm maintains the invariant that after every update, the coloring of the forest $F$ is distributed exactly according to $\mathcal D (F)$. Henceforth, we  use $\mathcal D$ for $\mathcal D (F)$ when $F$ would be clear from the context.

\iffullversion
\subsection{A Randomized Distribution-Maintaining Algorithm}
\else
\subsubsection{A Randomized Distribution-Maintaining Algorithm}
\fi

To define the distribution, and then the algorithm, it is clearest to begin with rooted forests.

\iffullversion
\subsubsection{Rooted Forests}
\else
\vspace{-0.2cm}
\paragraph{Rooted Forests}
\fi

In the rooted setting (Definition~\ref{definition_rooted_forest}), the distribution $\mathcal{D}$ has a simple top-down description. Starting from the root, each vertex assigns distinct colors to its child-edges uniformly at random, with one restriction:  the color of the edge connecting that vertex to its parent is forbidden (and thus excluded). Repeating this recursively down each tree defines the distribution. 
\iffullversion
The formal definition is given in Definition~\ref{def:DF} below.
\else
\fi

We describe how the algorithm maintains this distribution shortly. The point of maintaining this distribution is that it allows us to analyze the expected recourse of the algorithm's maintenance operations using the exact probability to recolor any particular edge.

\iffullversion

\subparagraph{A top--down coloring distribution.}\label{par_topdown}

Fix $\Delta\ge 2$ and $\pal\geq \Delta$.
We define $\mathcal D(F)$ using a single local sampling primitive.

\begin{definition}[Uniform coloring]\label{def:uda}
Let $v$ be a vertex with $\nrchild=|\ch(v)|$ children, and let $ \colset{1} \subseteq [\pal]$ be of size at least $\nrchild$.
A \emph{uniform coloring from $\colset{1}$ of the child edges of $v$} is obtained by choosing an $\nrchild$-subset
$\colset{2}\subseteq \colset{1}$ uniformly at random and then assigning the colors of $\colset{2}$ to the $\nrchild$ child edges of $v$ by a uniformly random permutation.\footnote{Equivalently, it is a uniformly random injective mapping from $\ch(v)$ to $\colset{1}$.}
\end{definition}

For a vertex $v$, let $T_v$ denote the subtree consisting of $v$ and all its descendants.

\begin{definition}[The distributions $\mathcal D$ and $\mathcal D_{\col{1}}$]\label{def:DF}
Let $T$ be a rooted tree with root $r$.

For a color $\col{1} \in [\pal]$, the distribution $\mathcal D_{\col{1}}(T)$ is defined recursively as follows.
Uniformly color the child edges of $r$  from $[\pal]\setminus\{\col{1}\}$.
Then, independently for each child $u\in\ch(r)$, generate the coloring of the subtree rooted at $u$, $T_u$, according to $\mathcal D_{\col{2}}(T_u)$, where $\col{2}=\colorf(r,u)$.

The distribution $\mathcal D(T)$ is defined in the same way, except that the child edges of the root are uniformly colored from $[\pal]$.

For a rooted forest $F$, the distribution $\mathcal D(F)$ is obtained by sampling each rooted tree component $T$ independently from $\mathcal D(T)$.
\end{definition}

We will repeatedly use the following two probabilities, which follow immediately from Definition~\ref{def:uda}.
Let $v$ be a non-root vertex with $\nrchild$ children and parent-edge color $\col{1}=\colorf(\parent(v),v)$.
Under $\mathcal D(F)$, the set of colors on child edges of $v$ is a uniformly random $\nrchild$-subset of $[\pal]\setminus\{\col{1}\}$.
Hence, for any fixed $\col{2} \neq \col{1}$,
\begin{equation}
\Pr[\col{2} \text{ appears on some child edge of }v] \;=\; \frac{\nrchild}{\pal-1}
\;\le\; \frac{\Delta-1}{\pal-1}. \label{eq:local-hit}
\end{equation}
Similarly, for a root $r$ with $\nrchild$ children and any fixed $\col{2}$,
\begin{equation}
\Pr[\col{2} \text{ appears on some child edge of }r] \;=\; \frac{\nrchild}{\pal}
\;\le\; \frac{\Delta}{\pal}. \label{eq:root-hit}
\end{equation}

\fi

\medskip
\noindent
{\textbf{The algorithm.}}
\iffullversion 
\else
The maintenance operations of the algorithm, henceforth \distalg{}, are based on one simple observation: an update can only add (in insertion) or remove (in deletion) a forbidden color at a single vertex. \distalg{} repairs this by possibly recoloring one child-edge, which pushes the same problem one level down the tree.

More concretely, if an edge $e=(p,r)$ is inserted (where $r$ is the root of its tree, and it becomes a child of $p$), the algorithm colors $e$ with a uniformly random color $\col{1}$ available at $p$. This makes $\col{1}$ forbidden at $r$. Thus, if there is a child-edge $(r,u)$ colored $\col{1}$, it must be recolored. The algorithm recolors $(r,u)$ to a uniformly random color $\col{2}$ that is available at $r$. This makes the coloring on child-edges of $r$ distributed correctly, as in $\mathcal{D}$, but now the forbidden color at $u$ changes from $\col{1}$ to $\col{2}$. If $u$ has a child-edge colored $\col{2}$, we apply a repair there as well, recoloring it from $\col{2}$ to $\col{1}$. This way, the repair propagates in a single path down the tree by swapping the colors $\col{1}$ and $\col{2}$ on the $(\col{1},\col{2})$-bicolored path until no forbidden color is violated, and the coloring of the new forest is distributed according to $\mathcal{D}$.

Deletion is similar. If an edge $(p,r)$ colored $\col{1}$ is deleted (so $r$ becomes a new root), then $\col{1}$ should no longer be forbidden at $r$. So, with the appropriate probability, we recolor a uniformly random child-edge $(r,w)$ by $\col{1}$. If this happens, and $(r,w)$ was previously colored $\col{2}$, then the forbidden color of $w$ changes from $\col{1}$ to $\col{2}$. This again creates a repair process down the tree until no forbidden color is violated, and the coloring of the new forest is distributed according to $\mathcal{D}$. Note that the coloring process was defined top-down, which is why we only need to fix the coloring at $r$ (with a possible cascade) but not at $p$ or anywhere else.
\fi
\iffullversion
Our algorithm, Algorithm~\ref{algorithm_randomized_main_rooted}, maintains the following invariant: at every time $t$, the current proper edge-coloring is distributed exactly as $\mathcal D(F_t)$ for the current rooted forest $F_t$.
Each update is handled by a local recoloring that may change the parent-edge color of one child. This triggers a recursive repair at that child, and the same mechanism propagates along a single downward path until no constraint is violated.

\begin{algorithm}[t]
\DontPrintSemicolon
    \KwIn{
         A sequence of edge updates over a forest $F$ with a maximum degree $\Delta$.
    }

    \KwOut{
        Maintains a proper coloring of $F$ after each update.
    }

    \SetKwProg{Fn}{Function}{:}{}

    \Fn{\funcUpdateForestColored{Forest $F$, edge $(p,r)$ where $p = parent(r)$, insert/delete}}{
        \uIf{insert}{
            Insert $(p,r)$ and choose $\colorf(p,r)$ uniformly from $A(p)$ (available colors at $p$).\;
            Call \funcRootToChildColored{$r,\colorf(p,r)$}.\;
          }
          \uIf{delete}{
            Let $\col{1} \leftarrow \colorf(p,r)$ and delete the edge $(p,r)$.\;
            Call \funcChildToRootColored{$r,\col{1}$}.\;
          }
    }
    
    \caption{Top--Down Random Edge Coloring Maintenance (calls Algorithm~\ref{algorithm_randomized_helpers})}\label{algorithm_randomized_main_rooted}
\end{algorithm}

At the top level, an insertion chooses a feasible color for the new edge at the parent endpoint and then repairs the former root that became a child using \funcRootToChild (Algorithm~\ref{algorithm_randomized_helpers}).
A deletion removes the parent edge of a child $r$, after which $r$ becomes a root and we repair it using \funcChildToRoot (Algorithm~\ref{algorithm_randomized_helpers}). We now describe these repair procedures.

\begin{algorithm}[t]
\DontPrintSemicolon
    
    \SetKwProg{Fn}{Function}{:}{}

    \Fn{\funcRootToChildColored{vertex $r$, color $\col{2}$}}{
        // Root $r$ becomes child with parent-edge color $\col{2}$.\;
        If no child edge $(r,w)$ has color $\col{2}$: \Return.\;
        Let $w$ be the unique child with $\colorf(r,w)=\col{2}$.\;
        Choose $\col{1} \in A(r)$ uniformly at random and recolor $(r,w)$ to $\col{1}$.\;
        Call \funcFixForbiddenColored{$w,\col{2},\col{1}$}.\;
    }

    \Fn{\funcChildToRootColored{vertex $r$, color $\col{1}$}}{
        // Vertex $r$ becomes root; parent-edge color $\col{1}$ removed.\;
        $\nrchild \leftarrow |\ch(r)|$.\;
        With probability $(\pal-\nrchild)/\pal$ \Return.\;
        Choose a uniformly random child $w\in\ch(r)$.\;
        $\col{2} \leftarrow \colorf(r,w)$.\;
        Recolor $(r,w)$ to $\col{1}$.\;
        Call \funcFixForbiddenColored{$w,\col{2},\col{1}$}.\;
    }

    \Fn{\funcFixForbiddenColored{non-root vertex $v$, color $\col{1}$, color $\col{2}$}}{
        // Parent-edge color changed $\col{1} \to \col{2}$.\;
        If no child edge $(v,w)$ has color $\col{2}$: \Return.\;
        Let $w$ be the unique child with $\colorf(v,w)=\col{2}$\;
        Recolor $(v,w)$ to $\col{1}$\;
        Call \funcFixForbiddenColored{$w,\col{2},\col{1}$}\;   
    }
    
    \caption{Helper functions for Algorithm~\ref{algorithm_randomized_main_rooted} and Algorithm~\ref{algorithm_randomized_main_unrooted}}
    \label{algorithm_randomized_helpers}
\end{algorithm}

When a non-root vertex $v$ changes its parent-edge color from $\col{1}$ to $\col{2}$, the only new constraint is that $\col{2}$ becomes forbidden on edges from $v$ to its children.
The procedure \funcFixForbidden (Algorithm~\ref{algorithm_randomized_helpers}) repairs this locally: if some child edge uses $\col{2}$, there is exactly one such edge, and it is recolored to $\col{1}$.
This fixes $v$, but it changes the parent-edge color of that child, so the same repair is applied recursively. Observe that since we only exchange the colors $\col{1}$ and $\col{2}$ down a path in the tree, the global effect of applying local repairs is simply to flip the colors on an $(\col{1},\col{2})$-bicolored path down the tree. The local presentation is simpler for the analysis.

When a root $r$ gains a parent edge of color $\col{2}$, the color $\col{2}$ becomes forbidden on child edges of $r$.
The procedure \funcRootToChild repairs $r$ so that the colors of its child-edges have the correct distribution conditioned on forbidding $\col{2}$:
if a child edge uses $\col{2}$, it is recolored to a uniformly random color $\col{1}$ not used on the other child edges of $r$.
This restores the correct distribution at $r$ and reduces to a forbidden-color change at the affected child, which is handled by \funcFixForbidden (The colors $\col{2}$ and $\col{1}$ define the bicolored path over which a flip will be applied down the tree).

When a non-root vertex $r$ becomes a root, its former parent-edge color $\col{1}$ stops being forbidden.
The procedure \funcChildToRoot restores the root distribution by inserting the color $\col{1}$ onto a uniformly random child edge with probability ${\nrchild}/{\pal}$, where $\nrchild=|\ch(r)|$ (see \eqref{eq:root-hit}). If the color $\col{1}$  was inserted, let $e = (r,x)$ be the edge that was chosen, and let $\col{2}$ be its former color. Now $\col{2}$ becomes forbidden for $x$, and this change of constraint is repaired by \funcFixForbidden. Again, from a global perspective, a bicolored path of colors $\col{1}$ and $\col{2}$ is now being flipped from $r$ downwards.
\fi

\iffullversion
\else
\begin{theorem}[Distribution invariant]
\label{theorem_distribution_invariant_high_level}
Against an oblivious adversary, after every update, the random coloring maintained by \distalg{} is distributed according to the top-down distribution $\mathcal{D}$ of the current forest.
\end{theorem}
\fi

\iffullversion
\subparagraph{Distribution invariant.}

We analyze the algorithm against an \emph{oblivious adversary}: the update sequence is fixed in advance and is independent of the randomness of the algorithm.
Let $F_t$ be the rooted forest after $t$ updates, and let $C_t$ be the random coloring maintained by Algorithm~\ref{algorithm_randomized_main_rooted} at time $t$.
The goal of this subsection is to show that $C_t$ is distributed exactly as $\mathcal D(F_t)$ for every $t$.

The proof reduces to three local statements, one for each repair procedure.
Each statement is formulated for an entire subtree and therefore already accounts for the full downward cascade created by the recursive calls.

\begin{theorem}[Distribution invariant]
\label{theorem_distribution_invariant_rooted}
Let $C_t$ be the random coloring maintained by Algorithm~\ref{algorithm_randomized_main_rooted}. For every $t\ge 0$ we have $C_t\sim \mathcal D(F_t)$.
\end{theorem}

\begin{lemma}[\textsc{FixForbidden}]\label{lem:fixforbidden-dist}
Let $v$ be a non-root vertex, and let $\col{1},\col{2}\in[\pal]$ with $\col{1}\neq \col{2}$.
Assume the current coloring restricted to the subtree $T_v$ is distributed as $\mathcal D_{\col{1}}(T_v)$.
After executing \funcFixForbidden{$v,\col{1},\col{2}$} (Algorithm~\ref{algorithm_randomized_helpers}), the resulting coloring restricted to $T_v$ is distributed as $\mathcal D_{\col{2}}(T_v)$.
\end{lemma}

\begin{proof}
We prove the statement by induction on the size of $T_v$.
If $v$ has no children, the procedure makes no changes and $\mathcal D_{\col{1}}(T_v)=\mathcal D_{\col{2}}(T_v)$.

Assume $\nrchild=|\ch(v)|\ge 1$ and that the claim holds for all strict descendant subtrees.
Under $\mathcal D_{\col{1}}(T_v)$, the $\nrchild$ child edges of $v$ are uniformly colored from $[\pal]\setminus\{\col{1}\}$.
The procedure \funcFixForbidden changes these colors only as follows: if no child edge has color $\col{2}$ it does nothing, and otherwise it replaces the unique occurrence of $\col{2}$ by $\col{1}$.
Hence, it defines a bijection from the proper colorings of the child edges from $[\pal]\setminus\{\col{1}\}$ to the proper colorings of the child edges from $[\pal]\setminus\{\col{2}\}$.

Conditioned on the child-edge colors at $v$, Definition~\ref{def:DF} makes the child subtrees independent, and each subtree $T_u$ is distributed as $\mathcal D_{\colorf(v,u)}(T_u)$.
If no recoloring occurs at $v$, no parent-edge color changes and there is no recursive call, so the subtree distributions already match $\mathcal D_{\col{2}}(T_v)$.
Otherwise, let $w$ be the unique child with $\colorf(v,w)=\col{2}$ before recoloring; after recoloring, the parent-edge color of $w$ becomes $\col{1}$, while every other child keeps the same parent-edge color.
The recursive call \funcFixForbidden{$w,\col{2},\col{1}$} is then executed.
Before this call, $T_w$ is distributed as $\mathcal D_{\col{2}}(T_w)$; by the induction hypothesis applied to the strict subtree $T_w$, the call transforms it into $\mathcal D_{\col{1}}(T_w)$. 
All other child subtrees are unchanged.
Together with the previous paragraph, this is exactly the recursive definition of $\mathcal D_{\col{2}}(T_v)$.
\end{proof}

\begin{lemma}[\textsc{RootToChild}]\label{lem:roottochild-dist}
Let $r$ be a root vertex and fix $\col{2} \in[\pal]$.
Assume the current coloring restricted to the tree $T_r$ is distributed as $\mathcal D(T_r)$.
After executing \funcRootToChild{$r,\col{2}$} (Algorithm~\ref{algorithm_randomized_helpers}), the resulting coloring restricted to $T_r$ is distributed as $\mathcal D_{\col{2}}(T_r)$.
\end{lemma}

\begin{proof}
Let $\nrchild=|\ch(r)|$.
Under $\mathcal D(T_r)$, the $\nrchild$ child edges of $r$ are uniformly colored from $[\pal]$.
Let $\colset{2}$ be the set of these $\nrchild$ colors.

If $\col{2}\notin \colset{2}$, the procedure does nothing at $r$.
Conditioned on $\col{2}\notin \colset{2}$, the set $\colset{2}$ is a uniformly random $\nrchild$-subset of $[\pal]\setminus\{\col{2}\}$, and conditioned on $\colset{2}$ the assignment of $\colset{2}$ to the child edges is a uniformly random permutation.

If $\col{2}\in \colset{2}$, write $\colset{2}=\colset{3}\cup\{\col{2}\}$ where $\colset{3}$ is an $(\nrchild-1)$-subset of $[\pal]\setminus\{\col{2}\}$.
Under a uniform coloring from $[\pal]$, the set $\colset{3}$ is uniformly random among $(\nrchild-1)$-subsets of $[\pal]\setminus\{\col{2}\}$, the unique edge colored $\col{2}$ is a uniformly random child edge, and the colors of $\colset{3}$ are assigned to the other $\nrchild-1$ child edges by a uniformly random permutation.
The procedure chooses $\col{1}$ uniformly from $[\pal]\setminus \colset{2}$, which equals $[\pal]\setminus\left( \colset{3}\cup\{\col{2}\} \right) $ and has size $\pal - \nrchild$, and replaces $\col{2}$ by $\col{1}$ on that uniformly random child edge.
Thus, the new set of child-edge colors is $\colset{3}\cup\{\col{1}\}$.
For any fixed $\nrchild$-subset $\colset{1}\subseteq [\pal]\setminus\{\col{2}\}$,
\begin{align*}
 \Pr[\colset{3}\cup\{\col{1}\}=\colset{1} \mid \col{2}\in \colset{2}]
 & =\sum_{\col{1} \in \colset{1}}\Pr[\colset{3}=\colset{1}\setminus\{\col{1}\}]\cdot \Pr[\col{1} \text{ is chosen}\mid \colset{3}] \\
 &= \nrchild \cdot \frac{1}{\binom{\pal - 1}{\nrchild - 1}}\cdot \frac{1}{\pal - \nrchild} =\frac{1}{\binom{\pal-1}{\nrchild}},   
\end{align*}
so $\colset{3}\cup\{\col{1}\}$ is uniform over $\nrchild$-subsets of $[\pal]\setminus\{\col{2}\}$.
Moreover, conditioned on $\colset{3}\cup\{\col{1}\}=\colset{1}$, the assignment of $\colset{1}$ to the $\nrchild$ child edges is a uniformly random permutation: the edge receiving the new color $\col{1}$ is uniformly random (it is the former $\col{2}$-colored edge $(r,w)$), and the remaining $\nrchild-1$ colors are uniformly permuted on the remaining edges.

Therefore, after the local step of \funcRootToChild at $r$, the child edges of $r$ are uniformly colored from $[\pal]\setminus\{\col{2}\}$, as required by $\mathcal D_{\col{2}}$.

Finally, if the procedure recolors $(r,w)$ from $\col{2}$ to $\col{1}$, then the parent-edge color of $w$ changes from $\col{2}$ to $\col{1}$ and the procedure calls \funcFixForbidden{$w,\col{2},\col{1}$}.
Conditioned on the colors on child edges of $r$, by Definition~\ref{def:DF} the child subtree colorings are independent and each $T_u$ is distributed as $\mathcal D_{\colorf(r,u)}(T_u)$.
Thus, before the call, $T_w$ is distributed as $\mathcal D_{\col{2}}(T_w)$, and by Lemma~\ref{lem:fixforbidden-dist} the call transforms it into $\mathcal D_{\col{1}}(T_w)$. All other child subtrees are unchanged.
This matches the recursive definition of $\mathcal D_{\col{2}}(T_r)$.
\end{proof}

\begin{lemma}[\textsc{ChildToRoot}]\label{lem:childtoroot-dist}
Let $r$ be a non-root vertex, and let $\col{1}=\colorf(\parent(r),r)$.
Assume the current coloring restricted to the subtree $T_r$ is distributed as $\mathcal D_{\col{1}}(T_r)$.
After executing \funcChildToRoot{$r,\col{1}$} (Algorithm~\ref{algorithm_randomized_helpers}), the resulting coloring restricted to $T_r$ is distributed as $\mathcal D(T_r)$.
\end{lemma}

\begin{proof}
Let $\nrchild=|\ch(r)|$.
Under $\mathcal D_{\col{1}}(T_r)$, the $\nrchild$ child edges of $r$ are uniformly colored from $[\pal]\setminus\{\col{1}\}$.
Let $\colset{2}$ be the set of these $\nrchild$ colors.

We first show that after the local step at $r$, the set of colors on child edges of $r$ is a uniformly random $\nrchild$-subset of $[\pal]$, and conditioned on this set the assignment to the $\nrchild$ child edges is a uniformly random permutation.

Fix any $\nrchild$-subset $\colset{1}\subseteq [\pal]$. If $\col{1}\notin \colset{1}$, then \funcChildToRoot can output the set $\colset{1}$ only by taking the do-nothing branch (probability $(\pal-\nrchild)/{\pal}$) and having $\colset{2}=\colset{1}$ (probability $1/{\binom{\pal-1}{\nrchild}}$). Hence:
\[
\Pr[\text{output set}=\colset{1}]
=\frac{\pal-\nrchild}{\pal}\cdot \frac{1}{\binom{\pal-1}{\nrchild}}
=\frac{1}{\binom{\pal}{\nrchild}}.
\]

If $\col{1}\in \colset{1}$, write $\colset{1}=\colset{3}\cup\{\col{1}\}$ where $|\colset{3}|=\nrchild-1$.
The procedure outputs $\colset{1}$ only in the recoloring branch (probability ${\nrchild}/{\pal}$). This branch chooses a uniformly random child edge, which is equivalent to choosing a uniformly random element of $\colset{2}$ to remove (since the assignment of $\colset{2}$ to child edges is a uniformly random permutation).
Thus the output set equals $\colset{1}$ exactly when $\colset{2}=\colset{3}\cup\{y\}$ for some $y\in [\pal]\setminus(\colset{3}\cup\{\col{1}\})$ and the removed element is $y$.
There are $\pal-\nrchild$ choices for such $y$, and for each one, $ \Pr[\colset{2}=\colset{3}\cup\{y\}]={1}/{\binom{\pal-1}{\nrchild}}$ and $\Pr[\text{remove }y\mid \colset{2}]={1}/{\nrchild}$.
Therefore:
\[
\Pr[\text{output set}=\colset{1}]
=\frac{\nrchild}{\pal}\cdot (\pal-\nrchild)\cdot \frac{1}{\binom{\pal-1}{\nrchild}}\cdot \frac{1}{\nrchild}
=\frac{\pal-\nrchild}{\pal}\cdot \frac{1}{\binom{\pal-1}{\nrchild}}
=\frac{1}{\binom{\pal}{\nrchild}}.
\]
So in all cases the output set is uniform over $\nrchild$-subsets of $[\pal]$.
Moreover, conditioned on the output set, the assignment to child edges is a uniformly random permutation: in the do-nothing branch nothing changes, and in the recoloring branch the edge that receives $\col{1}$ is chosen uniformly among the $\nrchild$ child edges.

Finally, if a recoloring occurs, let $w$ be the chosen child and let $\col{2}$ be the old color of $(r,w)$.
The procedure calls \funcFixForbidden{$w,\col{2},\col{1}$}.
Conditioned on the colors on child edges of $r$, by Definition~\ref{def:DF} the child subtree colorings are independent and each $T_u$ is distributed as $\mathcal D_{\colorf(r,u)}(T_u)$.
Thus before the call, $T_w$ is distributed as $\mathcal D_{\col{2}}(T_w)$, and by Lemma~\ref{lem:fixforbidden-dist} the call transforms it into $\mathcal D_{\col{1}}(T_w)$. All other child subtrees are unchanged.
This matches the recursive definition of $\mathcal D(T_r)$.
\end{proof}

\begin{proof}[Proof of Theorem~\ref{theorem_distribution_invariant_rooted}]
We prove by induction on $t$ that $C_t\sim \mathcal D(F_t)$.

The base case $t=0$ is the empty forest with the empty coloring. For the inductive step, assume that $C_t\sim \mathcal D(F_t)$ and consider the update made at $t+1$.

\smallskip
\noindent
\emph{Insertion.}
Suppose the update attaches a root $r$ as a child of a vertex $p$ in a different component, and let $\col{2}=\colorf(p,r)$ be the color chosen for the new edge.
Under $\mathcal D(F_t)$ the colorings of different components are independent, so the current coloring of $T_r$ is independent of the choice of $\col{2}$.

Consider the endpoint $p$.
Under $\mathcal D(F_t)$, the child edges of $p$ are uniformly colored from the allowed palette at $p$ (either $[\pal]$ if $p$ is a root, or $[\pal]\setminus\{\colorf(\parent(p),p)\}$ otherwise).
Algorithm~\ref{algorithm_randomized_main_rooted} chooses $\col{2}$ uniformly from the colors in this palette that are not used on the existing child edges of $p$.
This extends the current uniform coloring on the old child edges by one additional child edge using a uniformly random unused color, which yields a uniform coloring on the new set of child edges of $p$ from the allowed colors.

At $r$, the update changes $r$ from a root to a non-root with parent-edge color $\col{2}$.
By Lemma~\ref{lem:roottochild-dist}, the call \funcRootToChild{$r,\col{2}$} transforms the coloring on $T_r$ from $\mathcal D(T_r)$ to $\mathcal D_{\col{2}}(T_r)$, as required.
All other vertices are unchanged.
Therefore the resulting coloring is distributed as $\mathcal D(F_{t+1})$.

\smallskip
\noindent
\emph{Deletion.}
Suppose the update deletes an edge $(p,r)$ where $p=\parent(r)$, and let $\col{1}=\colorf(p,r)$ be its color before deletion.
Removing this edge deletes one child edge of $p$.
Under $\mathcal D(F_t)$, the child edges of $p$ are uniformly colored from the same allowed palette as above.
Restricting a uniform coloring to a subset of the edges preserves uniformity, hence after deleting $(p,r)$ the remaining child edges of $p$ are still uniformly colored from that palette.

At $r$, the update changes $r$ from a non-root with parent-edge color $\col{1}$, to a root.
By Lemma~\ref{lem:childtoroot-dist}, the call \funcChildToRoot{$r,\col{1}$} transforms the coloring on $T_r$ from $\mathcal D_{\col{1}}(T_r)$ to $\mathcal D(T_r)$, as required.
All other vertices are unchanged.
Therefore the resulting coloring is distributed as $\mathcal D(F_{t+1})$.
\end{proof}

\fi 

\iffullversion
\subsubsection{General (unrooted) Forests via Rerooting}
\else
\vspace{-0.2cm}
\paragraph{General (unrooted) Forests via Rerooting}
\fi

\iffullversion
\else

To extend the algorithm to general (unrooted) forests, we show that this top-down distribution  $\mathcal{D}(F)$ is, in fact, the uniform distribution over all proper $\pal$-edge-coloring of each tree in $F$. Thus the top-down view from a root is only a convenient way to describe the uniform distribution. This means that we can reroot any tree at any vertex and still view the coloring as distributed according to the same top-down distribution from the new root (rerooting is conceptual, no actual work is done to reroot the tree). Therefore, by rerooting appropriately before each update, we can use the same maintenance algorithm. 

\fi

\iffullversion
We now extend the algorithm to the standard model for (unrooted) forests, where an insertion adds an edge between two arbitrary vertices in different trees (both may be non-roots), and a deletion removes an existing edge.
The key point is that the distribution $\mathcal D$ from Definition~\ref{def:DF} is not merely a convenient rooted distribution: On a tree it is exactly the \emph{uniform} distribution over all proper $\pal$-edge-colorings.
Consequently, the maintained distribution does not depend on which vertex is chosen as the root of a component, and we may reroot components freely (without recoloring) in order to interpret every update as a rooted update handled by Algorithm~\ref{algorithm_randomized_main_rooted}. The revised unrooted algorithm is given in Algorithm~\ref{algorithm_randomized_main_unrooted}, note that  it is almost identical to Algorithm~\ref{algorithm_randomized_main_rooted}.

\begin{lemma}[Uniformity of $\mathcal D$]\label{lem:D-uniform}
Fix $\pal \geq \Delta$, and let $T$ be a rooted tree. The distribution $\mathcal D(T)$ (Definition~\ref{def:DF}) is the uniform distribution over all proper edge-colorings of $T$ with colors $[\pal]$.
\end{lemma}

\begin{proof}
Let $r$ be the root of $T$.
For a vertex $v$, denote $\nrchild(v)=|\ch(v)|$. 

Fix any proper edge-coloring $\sigma$ of $T$ with colors $[\pal]$.
We compute the probability of producing $\sigma$ under $\mathcal D(T)$.
In the top--down process, each vertex $v$ produces a uniform coloring of its child edges from an allowed palette:
$[\pal]$ at the root, and $[\pal]\setminus\{\col{1}\}$ at a non-root vertex whose parent-edge color is $\col{1}$.
Therefore the probability of realizing $\sigma$ is the product, over all vertices, of the probabilities that each vertex produces the child-edge coloring specified by $\sigma$.

If $v$ is a root, then $\mathcal D(T)$ assigns to its $\nrchild(v)$ child edges a uniform coloring from $[\pal]$.
The number of such assignments is $\frac{\pal!}{(\pal-\nrchild(v))!}$, hence the probability that $v$ produces the assignment specified by $\sigma$ is $\frac{(\pal-\nrchild(v))!}{\pal!}$.
If $v$ is not a root, then the parent-edge color is fixed to some $\col{1}$, and $\mathcal D(T)$ assigns to the $\nrchild(v)$ child edges a uniform coloring from $[\pal]\setminus\{\col{1}\}$.
The number of such colorings is $\frac{(\pal-1)!}{(\pal-1-\nrchild(v))!}$, hence the probability that $v$ produces the child-edge coloring specified by $\sigma$ is $\frac{(\pal-1-\nrchild(v))!}{(\pal-1)!}$.
Multiplying over all vertices, we get the probability
\[
\frac{(\pal-\nrchild(r))!}{\pal!}\cdot
\prod_{v\neq r}
\frac{(\pal-1-\nrchild(v))!}{(\pal-1)!} = \frac{1}{\pal}\cdot \prod_{v}
\frac{(\pal-\deg(v))!}{(\pal-1)!}, \]
where the last equality follows from the fact that $\deg(r) = \ell(r)$ and $\deg(v) = \ell(v)+1$ for all $v \neq r$. This expression depends only on the degrees and is, in particular, independent of $\sigma$.\footnote{Note that it is also invariant to the choice of root $r$.} Therefore, every proper edge-coloring has the same probability, i.e., $\mathcal D(T)$ is uniform.
\end{proof}

We maintain a uniform random proper $\pal$-edge-coloring of an unrooted forest.
The algorithm reroots only to apply the update in the rooted model. Rerooting does not change any edge colors and incurs zero recourse.

\begin{algorithm}[t]
\DontPrintSemicolon
    \KwIn{
         A sequence of edge updates over a forest $F$ with a maximum degree $\Delta$.
    }

    \KwOut{
        Maintains a proper coloring of $F$ after each update.
    }

    \SetKwProg{Fn}{Function}{:}{}

    \Fn{\funcUpdateForestColored{Forest $F$, edge $(u,v)$, insert/delete}}{
        \uIf{insert}{
            Let $T_u$, $T_v$ be the components containing $u$ and $v$, respectively.\;
            
            \uIf{$|E(T_v)| \le \Delta$ \textbf{and} $|E(T_u)| \leq \Delta$}{
            w.l.o.g. let $v$ and $u$ be such that $\deg(v) \le \deg(u)$.
            }
            \uElse{
            w.l.o.g. let $v$ and $u$ be such that $|E(T_v)| \le |E(T_u)|$.
            }
            
            Reroot $T_v$ at $v$, and let $u = \parent(v)$.\;
            Insert $(u,v)$ and choose $\colorf(u,v)$ uniformly among $A(u)$.\;
            Call \funcRootToChildColored{$v,\colorf(u,v)$}\;
          }
          \uIf{delete}{
            w.l.o.g. assume that $u = \parent(v)$ by the current rooting of their tree.\;

            Let $\col{1} \leftarrow \colorf(u,v)$ and delete the edge $(u,v)$.\;
            Call \funcChildToRootColored{$v,\col{1}$}.\;
          }
    }
    
    \caption{\distalg{}: Uniform Edge Coloring for Unrooted Forests (calls Algorithm~\ref{algorithm_randomized_helpers})}\label{algorithm_randomized_main_unrooted}
\end{algorithm}

\begin{theorem}[Uniformity invariant for unrooted forests]
\label{theorem_distribution_invariant_unrooted}
For every $t\ge 0$ the random coloring maintained by Algorithm~\ref{algorithm_randomized_main_unrooted}
is distributed uniformly over all proper edge-colorings of the current (unrooted) forest with colors $[\pal]$.
\end{theorem}

\begin{proof}
By induction on $t$.
The base case is trivial.
Assume the claim holds at time $t$, so the current coloring is uniform over proper colorings of the current (unrooted) forest.

Algorithm~\ref{algorithm_randomized_main_unrooted} may reroot one component (upon insertion), changing only parent/child relations.
By Lemma~\ref{lem:D-uniform}, for any choice of roots, the same uniform distribution can be viewed as the rooted distribution $\mathcal D(F_t)$.
The algorithm then performs exactly the same recoloring steps as Algorithm~\ref{algorithm_randomized_main_rooted}.
By Theorem~\ref{theorem_distribution_invariant_rooted}, after these steps the resulting rooted coloring is distributed as $\mathcal D(F_{t+1})$ for the new rooted forest representation, which by Lemma~\ref{lem:D-uniform} is the uniform distribution over all proper $\pal$-edge-colorings of the updated forest.
\end{proof}

Now that we have established that the randomized algorithm maintains the uniform distribution over proper edge colorings, we can proceed to analyze its recourse. 
\fi

\iffullversion
\subsection{Recourse Analysis}
\else
\subsubsection{Recourse Analysis (Incremental and Fully Dynamic Forests)}
\fi

\iffullversion
\else
Maintaining the distribution $\mathcal{D}$ allows us to compute exact recoloring probabilities, a fundamental component of our analysis. Suppose a repair procedure starts at a vertex $r$, and let $e$ be an edge in the subtree of $r$ at depth $d$. For the repair to reach $e$, it must follow the unique path from $r$ to $e$. This occurs only if each of the $d$ child-edges along this path is colored with one specific color among the $\pal-1$ allowed colors, with one exception at the root, where we have $\pal$ allowed colors. Thus, $e$ is recolored with probability $1/(\pal(\pal-1)^{d-1})$. 
\fi

\iffullversion
We begin with two useful lemmas that quantify the exact probability that a fixed edge is recolored due to \funcFixForbidden and due to an insertion.

\begin{lemma}[Recoloring probability of an edge due to \texttt{FixForbidden}]
\label{lem:fixforbidden_edge_recolor_prob}
Let $v$ be a non-root vertex, let $\col{1},\col{2} \in [\pal]$ be two distinct colors, and assume that the current
coloring of $T_v$ is distributed as $\mathcal D_{\col{1}}(T_v)$.
Let $e \in E(T_v)$ be an edge at depth $d \geq 1$ from $v$.
Then $\Pr\!\left[\text{\funcFixForbidden{$v,\col{1},\col{2}$}} \text{ recolors } e\right]
= {1}/{(\pal-1)^d}$.
\end{lemma}

\begin{proof}
By induction on $d$. For $d=1$, $e$ is a child edge of $v$.
Conditioned on the parent-edge color of $v$ being $\col{1}$, by Definition~\ref{def:DF}, the child
edges of $v$ are uniformly colored from $[\pal] \setminus \{\col{1}\}$.
Therefore, $e$ has color $\col{2}$ with probability $1/(\pal-1)$.
Since \funcFixForbidden{$v,\col{1},\col{2}$} recolors exactly the unique child edge of color $\col{2}$, it follows that $e$ is recolored with probability $1/(\pal-1)$.

Now let $d \ge 2$, and let $e_1=(v,w)$ be the first edge on the unique $v$-to-$e$ path.
For \funcFixForbidden{$v,\col{1},\col{2}$} to recolor $e$, it must first recolor $e_1$.
By the base case, $\Pr[e_1 \text{ is recolored}] = 1/(\pal-1)$.
Now, condition on the event that $e_1$ is recolored. Before the recoloring, $e_1$ had color $\col{2}$, so by Definition~\ref{def:DF}, the subtree $T_w$ is distributed as $\mathcal D_{\col{2}}(T_w)$.
After recoloring $e_1$, the procedure makes the recursive call \funcFixForbidden{$w,\col{2},\col{1}$}.
The edge $e$ is now at depth $d-1$ from $w$, so by the induction hypothesis,
conditioned on recoloring $e_1$, the probability that the recursive call recolors $e$ is
$1/(\pal-1)^{d-1}$.
Hence,
\begin{align*}
& \Pr\!\left[\text{\funcFixForbidden{$v,\col{1},\col{2}$} recolors } e\right] = \\
& \quad = \Pr\!\left[\text{\funcFixForbidden{$v,\col{1},\col{2}$} recolors } e \mid e_1 \text{ is recolored}\right]  \cdot \Pr\!\left[e_1 \text{ is recolored}\right] \\
& \quad = \frac{1}{(\pal-1)^{d-1}}\cdot \frac{1}{\pal-1} = \frac{1}{(\pal-1)^d}.\qedhere
\end{align*} 
\end{proof}

\begin{lemma}[Recoloring probability of an edge due to insertion]
\label{lem:exact_edge_recoloring_prob}
Consider an insertion of an edge $(p,r)$. Let $T_r$ be the tree rooted at $r$ prior to this insertion, and let $e \in E(T_r)$ be an edge at depth $d \ge 1$ from $r$. Then, the probability that $e$ is recolored in this insertion is ${1}/{\left(\pal(\pal-1)^{d-1}\right)}$.
\end{lemma}

\begin{proof}
Let $r=v_0,v_1,\dots,v_d$ be the unique root-to-$e$ path in $T_r$, and let
$e_i=(v_{i-1},v_i)$ for $i\in[d]$, so $e_d=e$.
For $e$ to be recolored, the recoloring walk must follow this path.

Let $\col{2}$ be the color chosen for the new parent-edge $(p,r)$.
By Theorem~\ref{theorem_distribution_invariant_rooted}, before the insertion, the coloring of $T_r$ is distributed as $\mathcal D(T_r)$.
Also, the choice of $\col{2}$ is made in the other component (uniformly among the available colors at $p$), so it is independent of the coloring of $T_r$.
By Definition~\ref{def:DF}, the child edges of $r$ are uniformly colored from $[\pal]$.
Hence, for every fixed $\col{2} \in [\pal]$, the edge $e_1$ has color $\col{2}$ with probability $1/\pal$.
So $\Pr[e_1 \text{ is recolored}] = 1/\pal$.

If $d=1$, we are done.
For $d \geq 2$, condition on the event that $e_1$ is recolored.
Let $\col{2}'$ be the new color assigned to $e_1$ by \funcRootToChild.
Before the recoloring, $e_1$ had color $\col{2}$, so by the recursive definition of $\mathcal{D}(T_r)$ (Definition~\ref{def:DF}), the subtree $T_{v_1}$ is distributed as $\mathcal{D}_{\col{2}}(T_{v_1})$.
The procedure then calls \funcFixForbidden{$v_1,\col{2},\col{2}'$}. We have
\begin{align}
\begin{split}
\Pr\!\left[ e \text{ is recolored}\right]
&= \Pr\!\left[ e \text{ is recolored} \mid e_1 \text{ is recolored} \right] \cdot 
\Pr\!\left[e_1 \text{ is recolored}\right] \\
&= \Pr[\text{\funcFixForbidden{$v_1,\col{2},\col{2}'$} recolors } e] \cdot \frac{1}{\pal}.\label{eq:recolor_e}
\end{split}
\end{align}
Since $e$ is at depth $d-1$ from $v_1$, by Lemma~\ref{lem:fixforbidden_edge_recolor_prob} we have  $\Pr[\text{\funcFixForbidden{$v_1,\col{2},\col{2}'$} recolors } e]
= {1}/{(\pal-1)^{d-1}}$.
Substituting this in~\eqref{eq:recolor_e} concludes the proof.
\end{proof}
\fi

\iffullversion
\subsubsection{Incremental Forests (Insertions Only)}
\else
\vspace{-0.2cm}
\paragraph{Incremental Forests (Insertions Only)}
\fi

\iffullversion
\else
The exact recoloring probabilities computed above give a straightforward bound in the incremental rooted setting. 

\begin{theorem}
Assume only insertions in the rooted forest model (Definition~\ref{definition_rooted_forest}). For $\Delta \geq 3$ and $0 \le c \le \Delta-2$, \distalg{} has expected amortized recourse $\Theta({1}/{\Delta})$.
\end{theorem}

\begin{proof}[Proof Sketch of the Upper Bound.]
Fix an edge $e$. It can be recolored only when the root of its current tree is linked to another tree. Each time this happens, the depth of $e$ in the new tree increases. Since the recoloring probability of $e$ decays exponentially with its depth, the expected number of times $e$ is recolored is bounded by a geometric series, which sums to $O(1/\Delta)$.  This implies an $O(1/\Delta)$ expected amortized recourse.
\end{proof}

\fi

\iffullversion
\else
The general (unrooted) setting is more subtle. Here, the depth of an edge does not necessarily increase every time its component is linked to another tree. Rerooting may bring the same edge close to the root repeatedly, which might increase the number of times it is recolored. We show that bad rerooting choices can lead to worse expected amortized recourse. Nevertheless, we give a rerooting strategy that maintains the tight $\Theta(1/\Delta)$ bound. 

\begin{theorem}
\label{theorem_unrooted_randomized_LB_UB_short_version}
Assume only insertions in the general (unrooted) forest model. For $\Delta \geq 3$ and $0 \le c \le \Delta-2$, \distalg{} has expected amortized recourse $\Theta({1}/{\Delta})$.
\end{theorem}

\begin{proof}[Proof Sketch of the Upper Bound.] 
We choose which component to reroot according to the size of the merged trees. We call a tree \textit{small} if it has at most $\Delta$ edges, and \textit{large} otherwise.

\looseness=-1
When both trees are small, we reroot at the endpoint of smaller degree (and attach it as a child of the other endpoint). In this case, the dominant contribution to the expected recourse comes from edges of depth $1$, so this rerooting choice keeps it small. Concretely, by the edge recoloring probabilities discussed earlier,
each edge of depth at least $2$ contributes at most $\frac{1}{\pal(\pal-1)}$ expected recourse. There are at most $\Delta$ edges in a small tree, so their total contribution is $O(1/\Delta)$. Each edge of depth $1$ contributes $\frac{1}{\pal}$ in expectation, so together they contribute $\frac{1}{\pal}$-fraction of the smaller degree. We use a potential function argument to show that, over the entire sequence, the sum of smaller degrees is linear in the number of insertions. Hence, this case contributes $O(1/\Delta)$ expected amortized recourse. 

When both components are large, we reroot the smaller side by size, and charge the recourse to its edges. We then divide these insertions into buckets according to the size of the smaller side: bucket $j$ contains insertions in which the smaller component has between $2^j$ and $2^{j+1} -1$ edges.  Each edge can be charged at most once in each bucket because, after a charge, the tree containing it at least doubles in size. Therefore, bucket $j$ contains at most $m/2^j$ insertions, where $m$ is the total number of insertions.

To bound the expected recourse of an insertion in a fixed bucket $j$, consider the rerooted smaller tree. It has fewer than $2^{j+1}$ edges. Since the probability that an edge is recolored decays exponentially with its depth, the expected recourse is maximized when the edges are placed as close to the root as possible. So it suffices to consider a tree of depth $h= O(\log_{\Delta-1}(2^{j+1}))$. In such a tree, the recourse is at most $h$ (a root-to-leaf path). Multiplying this by the number of insertions in bucket $j$, and summing over all non-empty buckets $j \geq \lfloor \lg{\Delta}\rfloor$, gives a total contribution of $O(m/\Delta)$, and therefore $O(1/\Delta)$ expected amortized recourse.\footnote{We write $\lg{n}$ for $\log_2{n}$.}

The remaining case is when a small tree is merged into a large one. This is the simplest case. Here too, we reroot the smaller side, and charge the recourse to its edges. By the recoloring probabilities, the charge is at most $O(1/\Delta)$ per edge, and this is a one-time charge, as the resulting tree is large.
\end{proof}
\fi

\iffullversion

We now analyze the expected amortized recourse of the distribution maintenance algorithm in the incremental setting. We show a tight $\Theta(1/\Delta)$ bound, both for rooted and general forests.

We begin with the rooted forests case, which is simpler to analyze. In fact, for this case we prove a stronger guarantee. We show that for any fixed edge $e$, the expected number of times $e$ is recolored throughout the entire insertions sequence is at most $O(1/\Delta)$. 

\begin{theorem}
\label{theorem_rooted_insertion_only_randomized_tight}
Assume only insertions in the rooted-forest model. For $\Delta \geq 3$ and $0 \le \extra \le \Delta-2$, Algorithm~\ref{algorithm_randomized_main_rooted} has expected amortized recourse $\Theta(\frac{1}{\Delta})$.
\end{theorem}

\begin{proof}
\textbf{We first prove the upper bound.} Fix the final rooted forest, and let $e$ be one of its edges.
We show that the expected number of times the color of $e$ changes is $O\!\left({1}/{\Delta}\right)$. Let $e=(p,u)$, where $p$ is the parent of $u$ in the final rooted forest.
The color of $e$ can change only when a rooted tree containing $e$ is attached at an ancestor of $p$. If that ancestor is at distance $d-1$ above $p$, then at that moment $e$ is at depth $d$ in the attached rooted tree. By Lemma~\ref{lem:exact_edge_recoloring_prob}, the probability that this insertion recolors $e$ is exactly $\frac{1}{\pal(\pal-1)^{d-1}}$. Summing over all possible ancestors above $e$, the expected number of recolorings of $e$ is at most

\[
\sum_{d\ge 1}\frac{1}{\pal \cdot (\pal-1)^{d-1}}
< \frac{\pal-1}{\pal \cdot (\pal-2)}
= O\!\left(\frac{1}{\pal} \right) = O\!\left( \frac{1}{\Delta} \right).
\]

\textbf{We now prove the lower bound.} Insert two edges, $(v,w)$ followed by $(u,v)$ where $v$ is a parent of $w$ and a child of $u$. With probability $\frac{1}{\pal}$, $(u,v)$ picks the color of $(v,w)$ and causes a recourse of $1$, thus we get $\Omega(\frac{1}{\Delta})$ expected recourse per insertion. A similar lower bound can be shown for longer paths, and for $d$-ary trees for $d = \Omega(\Delta)$ whose edges are inserted bottom-up.
\end{proof}

\fi

\iffullversion

We now turn to general forests. To use the same distribution maintenance operations as in the rooted case, the algorithm must reroot at least one component before each insertion. Although the algorithm maintains the same distribution, the rooted proof does not apply to the unrooted case. This is because in the rooted setting, a fixed edge $e$ can be recolored only when the root of its current tree is linked to another tree, 
and each time this happens, the depth of $e$ in the merged tree increases. This is not the case in the unrooted setting. For example, several edges incident to $e$ may be inserted, and if the component of $e$ is rerooted several times due to these insertions, $e$ becomes a child-edge of the root many times. This can increase the expected number of times $e$ is recolored. This also suggests that the choice of which component to reroot is important, as the following remark illustrates.

\begin{remark}[Rerooting carefully is essential]\label{remark:rerooting}
    When two small components are linked by an edge $(u,v)$, Algorithm~\ref{algorithm_randomized_main_unrooted} reroots the endpoint with a smaller degree. This choice is essential for the $o(1)$ bound. To demonstrate this, consider a sequence that builds a star of degree $\Delta$ centered at $v$. The smaller degree in every insertion is $0$ (new leaf). 
    
    Rerooting the leaf side causes no recourse. On the other hand, suppose that we reroot the star side upon every insertion, and make it a child of the new leaf. On the $i$-th insertion, $v$ has $i-1$ child-edges. The new parent-edge is colored uniformly at random from a palette of size $\pal$. Hence, with probability $\frac{i-1}{\pal}$, it collides with a color of a child-edge, and recourse is incurred. Thus, the expected total recourse is $\sum_{i=1}^{\Delta} \frac{i-1}{\pal} = \Omega(\frac{\Delta^2}{\pal}) = \Omega(\Delta)$, which implies $\Omega(1)$ expected amortized recourse.
\end{remark}

\begin{theorem}
\label{theorem_unrooted_insertion_only_randomized_tight}
Let $\Delta \geq 3$ and $0 \le c \le \Delta-2$. Assume only insertions to a forest of maximum degree at most $\Delta$. Algorithm~\ref{algorithm_randomized_main_unrooted} has expected amortized recourse $
\Theta\!\left(\frac{1}{\Delta}\right)$.
\end{theorem}

We begin with a simple bound that will later be used in the proof of the theorem.

\begin{proposition}\label{prop:sum_min_deg_potential}
    Let $\{(u_i,v_i)\}^{m}_{i=1}$ be a sequence of insertions and denote by $\deg_{i-1}(\cdot)$ the degree just before the $i$th insertion. Then $\sum_{i=1}^{m} \min\{\deg_{i-1}(u_i),\deg_{i-1}(v_i)\} \leq m$.
\end{proposition}
\begin{proof}
    For a tree $T$, let $\delta(T) \equiv \max_{w\in T} \deg_T(w)$. We define $\Phi(T) \equiv |E(T)| - \delta(T) $, and for an entire forest $F$, $\Phi(F) \equiv \sum_{T \in F}{\Phi(T)}$.

    Let $F_{i-1}$ be the forest just before the $i$th insertion.  Consider the insertion of $e_i=(u_i,v_i)$. Let $T_u$ and $T_v$ be the two trees that are merged (we omit the $i$ for readability), and let $T_e = T_u \cup T_v \cup \{e_i\}$ be the merged tree. The change in potential due to this insertion is
    \begin{align*}
        \Phi(F_i) - \Phi(F_{i-1})  &= \Phi(T_e) - \left( \Phi(T_u) + \Phi(T_v) \right) \\
        &= |E(T_e)| - \delta(T_e) - \left(|E(T_u)| - \delta(T_u)  + |E(T_v)| - \delta(T_v) \right) \\
        &= \delta(T_u) + \delta(T_v) - \delta(T_e) + 1 \\
        &\geq \delta(T_u) + \delta(T_v) - (\max\{\delta(T_u),\delta(T_v)\} + 1) + 1\\
        &= \min\{\delta(T_u),\delta(T_v)\} \geq \min\{\deg_{i-1}(u_i), \deg_{i-1}(v_i)\},
    \end{align*}
    where the first inequality is due to the fact that adding one edge can increase the maximum degree by at most $1$. Therefore, we get that $\Phi(F_m) - \Phi(F_0) = \sum_{i=1}^{m} (\Phi(F_{i}) - \Phi(F_{i-1})) \geq \sum_{i=1}^{m}\min\{\deg_{i-1}(u_i),\deg_{i-1}(v_i)\} $. Also, $\Phi(F_0) = 0$, and $\Phi(F_m) =  \sum_{T \in F_m} (|E(T)| - \delta(T)) \leq \sum_{T \in F_m} |E(T)| = m$. Combining the last two inequalities concludes the proof.
\end{proof}

We are now ready to prove Theorem~\ref{theorem_unrooted_insertion_only_randomized_tight}.

\begin{proof}[Proof of Theorem~\ref{theorem_unrooted_insertion_only_randomized_tight}]
\textbf{We first prove the lower bound.} Proving the lower bound is almost as in Theorem~\ref{theorem_rooted_insertion_only_randomized_tight}. The only difference is that we insert $3$ edges, instead of $2$, in the following order: $(u,v)$, $(w,z)$ and then $(v,w)$. Now, regardless of how Algorithm~\ref{algorithm_randomized_main_unrooted} breaks symmetry and chooses its rooting, the total expected recourse is $\frac{1}{\pal} = \Theta(\frac{1}{\Delta})$, averaged over $3$ insertions. (This argument shows that whenever an edge is inserted between two existing edges, it contributes expected recourse of at least $\frac{1}{\pal}$.)

\medskip
\textbf{We now prove the upper bound.} In what follows, we refer to a component $T$ as \emph{small} if it contains at most $\Delta$ edges, i.e., $|E(T)| \leq \Delta$, and \emph{large} otherwise.
We partition the recourse-causing insertions of edges $e = (u,v)$ into three types according to the sizes of the two components being merged.
Let $T_u$ and $T_v$ be the components containing $u$ and $v$, respectively.

\begin{itemize}
    \item \textbf{Type 1:} both components are small.
    \item \textbf{Type 2:} one component is small and the other is large.
    \item \textbf{Type 3:} both components are large.
\end{itemize}

\textbf{Type 1 analysis:} 
Consider a Type 1 insertion $(u,v)$. The algorithm reroots the endpoint of smaller degree. Assume, without loss of generality that $\deg(v) \leq \deg(u)$. So the algorithm reroots $T_v$ at $v$, and then calls \funcRootToChild on $v$.

We first bound the expected recourse due to the edges at depth at least $2$ from $v$. By Lemma~\ref{lem:exact_edge_recoloring_prob}, each such edge is recolored with probability at most $\frac{1}{\pal(\pal-1)}$. There are at most $\Delta$ edges in $T_v$, so the total contribution of depth $\geq 2$ edges is at most $\frac{\Delta}{\pal(\pal-1)} = O(1/\Delta)$. We therefore charge this part directly to the insertion of $(u,v)$.

It remains to bound the contribution of the edges at depth $1$. Again, by Lemma~\ref{lem:exact_edge_recoloring_prob}, their expected recourse is exactly $\deg(v)/\pal = \min\{\deg(v), \deg(u)\}/ \pal$ . We amortize this quantity over all insertions. By Proposition~\ref{prop:sum_min_deg_potential}, the sum of this quantity over all $m$ insertions is upper bounded by $m/\pal$. Thus, the expected amortized Type 1 recourse is $O(1/\Delta)$ per insertion.

\textbf{Type 2 analysis:}
Consider a Type 2 insertion $(u,v)$, where $v$ belongs to the small component $T_v$. Thus, only edges of $T_v$ can be recolored. By Lemma~\ref{lem:exact_edge_recoloring_prob}, each edge of $T_v$ is recolored with probability at most $1/\pal$. We charge the recourse of a recolored edge directly to that edge. So, the expected charge of $e \in E(T_v)$ due to this insertion is at most $1/\pal$. 

Now, each edge can be charged due to Type 2 insertion at most once. Consider the first time such a charge is made to an edge $e$. One of the merged components is large, so after this insertion, $e$ is part of a large component and since we only perform insertions, it can never again lie in the small side of a Type 2 insertion. Thus, the total expected Type 2 recourse per edge is $O(1/\Delta)$.

\textbf{Type 3 analysis:} When both components are large, the algorithm reroots the smaller component. We define buckets according to the size of the smaller component.
For an integer $j$, bucket $j$ consists of those insertions for which the smaller component has between $2^j$ and $2^{j+1}-1$ edges. Note that since both components are large in this case, the first $\floor{\lg {\Delta}} - 1$ buckets are empty.

Fix a bucket $j$.
If an insertion occurs in bucket $j$, then the smaller component has at least $2^j$ edges. After the insertion, the merged component has at least $2^{j+1}$ edges, and thus leaves bucket $j$. Hence bucket $j$ contains at most ${m}/{2^j}$ recourse-causing insertions.

It remains to bound the expected recourse of one insertion in bucket $j$.
Let $T_v$ be the smaller component, rooted at the attachment vertex, and for $d\ge 1$ let $n_d$ be the number of edges in $E(T_v)$ at depth $d$ from the root. By Lemma~\ref{lem:exact_edge_recoloring_prob}, an edge at depth $d$ is recolored with
probability $\frac{1}{\pal(\pal-1)^{d-1}}$. Hence,
\begin{align}\label{eq:exp_rand}
    \mathbb{E}[\text{recourse}] = \sum_{e\in T_v} \Pr[e \text{ is recolored}] =\frac{1}{\pal}\sum_{d\ge 1}\frac{n_d}{(\pal-1)^{d-1}}.
\end{align}

Every non-root vertex has at most $\Delta-1$ children, so $n_d \le (\Delta-1)^d$ for every $d$.
For fixed $|E(T_v)|$,~\eqref{eq:exp_rand} is maximized by placing as many edges as possible as close to the root as possible.
Thus, if $h=\left\lceil \log_{(\Delta-1)}(|E(T_v)|+1)\right\rceil$,
then:
\[
\mathbb{E}[\text{recourse}] = \frac{1}{\pal}\sum_{d\ge1}\frac{n_d}{(\pal-1)^{d-1}}
\leq
\frac{1}{\pal}\sum_{d=1}^{h}\frac{(\Delta-1)^d}{(\pal-1)^{d-1}}
=
\frac{\Delta-1}{\pal}\sum_{d=0}^{h-1}\left(\frac{\Delta-1}{\pal-1}\right)^d.
\]
Now since $\Delta-1 \le \pal-1$, each term in the last sum is at most $1$. Therefore, 
$\mathbb{E}[\text{recourse}] \leq \frac{\Delta-1}{\pal}\cdot h < h$.

Since $|E(T_v)|<2^{j+1}$, we get $ \mathbb{E}[\text{recourse}] \le \log_{(\Delta-1)}(2^{j+1}) = \frac{j+1}{\lg (\Delta-1)}$. Therefore, every insertion in bucket $j$ has expected recourse at most $\frac{j+1}{\lg(\Delta-1)}$. Summing over all the buckets, and recalling that the first $\floor{\lg \Delta} - 1$ are empty, we get a total contribution of at most:
\[
\sum_{j \ge \floor{\lg \Delta}}
\frac{m}{2^j} \cdot \frac{j+1}{\lg(\Delta-1)}
=
\frac{m}{\lg (\Delta-1)} \cdot \sum_{j \ge \floor{\lg \Delta}}
\frac{j+1}{2^j}
=
\frac{m}{\lg (\Delta-1)} \cdot O\!\left(\frac{\lg \Delta}{\Delta}\right)
=
O\!\left(\frac{m}{\Delta}\right).
\]
In other words, we showed $O({1}/{\Delta})$ Type 3 recourse per edge.

The theorem follows by combining the three types, each contributes $O({1}/{\Delta})$ expected amortized recourse.
\end{proof}
\fi

\smallskip
\noindent
\textbf{A lower bound for constant $c$.}
\iffullversion
Theorem~\ref{theorem_rooted_insertion_only_randomized_tight} and Theorem~\ref{theorem_unrooted_insertion_only_randomized_tight} present tight analyses of Algorithm~\ref{algorithm_randomized_main_rooted} and Algorithm~\ref{algorithm_randomized_main_unrooted}. The following lower bound shows that they are, in fact, optimal for $\extra=O(1)$.

\thmRandomizedIncrementalLB*
\else
To conclude the incremental setting, we show that our algorithm is, in fact, optimal for constant values of $\extra \geq 0$, by proving a matching $\Omega(1/\Delta)$ lower bound for \emph{any} algorithm, even on rooted forests.

\begin{restatable}[]{theorem}{thmRandomizedIncrementalLB}\label{thm:incremental_randomized_lb}
For every constant $\extra \geq 0$ and every $\Delta \geq 3$ such that
$\extra \leq \Delta-2$, any randomized algorithm that maintains a proper $\pal$-edge coloring of an incremental forests of maximum degree $\Delta$ has expected amortized recourse $\Omega(1/\Delta)$ against an oblivious adversary, where the hidden constant depends only on $\extra$.
\end{restatable}
\fi

\iffullversion

\begin{proof}
    We construct an distribution over update sequences on which every deterministic algorithm has expected amortized recourse $\Omega(1/\Delta)$. The theorem then follows from Yao's principle.

    The construction is simple. Let $\ell \geq 2$ be a constant integer that depends only on $\extra$. We choose the value of $\ell$ later. We first construct $N = \ell \Delta$ vertex disjoint stars of degree $\Delta-1$, with centers $v_1,\dots,v_N$. We then draw a uniformly random subset of $\Delta$ star centers, and connect these centers to a fresh isolated vertex $u$. More accurately, we draw a uniformly random subset of indices $I \subseteq [N]$ of size $\Delta$, and insert the edges $(v_i,u)$ for $i\in I$. The total number of updates is $ N \cdot (\Delta-1) + \Delta = O(\Delta^2)$, so it remains to show that the expected total recourse is $\Omega(\Delta)$.

    Let $U_i$ be the set of colors \emph{not} used by the $(\Delta-1)$ leaf edges of $v_i$, and let $U = \cup_{i\in I} U_i$. We consider two particular times in the update sequence. The first is right after constructing the $N$ stars and before inserting the edges incident to $u$. We denote this time $t'$, and, for some set $X$, we use $X'$ to refer to its value at time $t'$. The second time, denoted $t''$, is at the end of the update sequence. Likewise, we write $X''$ to refer to the value of $X$ at time $t''$.
    
    Let $R$ be the total recourse incurred by the algorithm. We first argue that $R \geq \Delta-|U'|$. This is because, at the end of the update sequence, the $\Delta$ edges incident to $u$ must have distinct colors, and the color of the edge $(v_i,u)$ must belong to $U''_i$. Thus, $|U''| \geq \Delta$. Therefore, from time $t'$ to the end of the update sequence, the size of the set $U$ must increase from $|U'|$ to at least $\Delta$. The size of $U$ can increase only due to recolorings of star leaf edges incident to selected centers. Each such recoloring may increase $|U|$ by at most $1$, so at least $\Delta - |U'|$ recolorings must occur.

    Let $M' = [\pal] \setminus U'$. Thus, at time $t'$, each color in $M'$ appears on some leaf edge incident to every drawn center $\{v_i\}_{i \in I}$. Therefore, these colors cannot be used on the inserted edges $\{(v_i,u)\}_{i \in I}$ without incurring recourse. Note that $R  \geq \Delta-|U'| = \Delta - \pal + |M'| = |M'| - \extra$. So, it is enough to lower bound $\expect{|M'|} - \extra$.

   \medskip
   \noindent
    We split into cases according to the size of $\Delta$. When $\Delta$ is large, we show that there are many colors in $M'$, that is, $\expect{|M'|} = \Omega(\Delta)$. When $\Delta$ is small, and bounded as a function of $\extra$ (a mere constant), it is sufficient to show that at least one recoloring happens with constant probability.
    
    \textbf{Case 1: $\Delta$ is large.}
    In particular, $\Delta \ge 8(\extra+1)e^{8(\extra+1)}$. Let $d(\alpha) = | \{ i : \alpha \in U'_i \}| $ be the number of stars in which $\alpha$ is free at time $t'$. Since $|U'_i| = \extra+1$ for all $i$, we have $\sum_{\alpha =1}^{\pal} d(\alpha) = N(\extra+1)$. We call a color $\alpha$ \emph{sparse} if $d(\alpha) \leq 2N(\extra+1)/\pal$. At least half of the colors are sparse, since otherwise, the sum $\sum_{\alpha =1}^{\pal} d(\alpha)$ would be larger than $N(\extra+1)$.

    Next, we bound the probability that a sparse color $\alpha$ belongs to $M'$. We have
    \begin{align*}
        \Pr[\alpha \in M'] = \frac{\binom{N-d(\alpha)}{\Delta}}{\binom{N}{\Delta}} = 
        \frac{ (N-\Delta)!} { N! } \cdot \frac{(N-d(\alpha))!}{(N-d(\alpha) - \Delta)!} = \prod_{j=0}^{\Delta - 1} \frac{N-d(\alpha) - j}{N-j} = \prod_{j=0}^{\Delta - 1} \left(1 - \frac{d(\alpha)}{N-j}\right).
    \end{align*}
    Since $N=\ell\Delta$ and $\ell \geq 2$, we have $N-j > N- \Delta \ge N/2$ for every $0\leq j \leq \Delta -1$, and since $\alpha$ is sparse, we get that 
\[
    \frac{d(\alpha)}{N-j} < \frac{2N(\extra+1)}{\pal}\cdot \frac{2}{N} =
    \frac{4(\extra+1)}{\Delta+\extra} \le \frac{1}{2},
\] where the last inequality follows from the case assumption $\Delta \ge 8(\extra+1)e^{8(\extra+1)} \ge 8(\extra+1)$. Therefore, using the fact that $1-x\ge e^{-2x}$ for $x\in [0,1/2]$, we obtain 
\begin{align*}
    \Pr[\alpha \in M']
    &\geq \exp\left( -2\sum_{j=0}^{\Delta-1}\frac{d(\alpha)}{N-j} \right) \geq \exp\left( -2\Delta  \cdot\frac{d(\alpha)}{N-\Delta}  \right). 
\end{align*}

Now, substituting $N= \ell \Delta$, and using the fact that $\alpha$ is sparse, we have
\begin{align*}
    \frac{2\Delta d(\alpha)}{N-\Delta}
    &\le
    \frac{2\Delta}{N-\Delta}\cdot \frac{2N(\extra+1)}{\pal} =
    \frac{2\Delta}{(\ell-1)\Delta}\cdot \frac{2\ell\Delta(\extra+1)}{\pal} =
    \frac{4\ell}{\ell - 1}\cdot \frac{\Delta}{\pal} \cdot (\extra+1) \le
    8(\extra+1),
\end{align*}
where the last inequality uses $\ell \ge 2$. Thus, each sparse color is in $M'$ with probability at least $e^{-8(\extra+1)}$. Since at least half of the colors are sparse, we have $\expect{|M'|} = \sum_{\alpha=1}^{\pal}\Pr[\alpha \in M'] \geq \frac{\Delta+\extra}{2}\cdot e^{-8(\extra+1)}$. Additionally, since $\Delta\ge 8(\extra+1)e^{8(\extra+1)}$, we have $\extra < \Delta e^{-8(\extra+1)}/8$. Hence
\[
    \expect{R}
    \geq \expect{|M'|} - \extra
    \geq  \frac{\Delta+\extra}{2}\cdot e^{-8(\extra+1)}- \frac{\Delta}{8}\cdot e^{-8(\extra+1)}
    \geq \frac{3}{8}\Delta e^{-8(\extra+1)}.
\]
Thus, in this case, the expected total recourse is $\Omega(\Delta)$.

\textbf{Case 2: Small $\Delta$.} That is, $\Delta < 8(\extra+1)e^{8(\extra+1)} = O(1)$. Let $C = \lceil8(\extra+1)e^{8(\extra+1)}\rceil$. 
Recall that each star with center $v_i$ is associated with a subset $U'_i$ of $\extra+1$ free colors at time $t'$.
Since $\Delta < C$ in this case, the number of possible subsets $U'_i$ is at most $\binom{\Delta+\extra}{\extra+1} \leq \binom{C+\extra}{\extra+1}$.
We now choose $\ell = 2\binom{C+\extra}{\extra+1}$. So $\ell \Delta / \binom{\Delta+\extra}{\extra+1} \geq 2\Delta$ for every $\Delta < 8(\extra+1)e^{8(\extra+1)}$. Thus, by the pigeonhole principle, at least one subset appears at least $2\Delta$ times. Let $J'$ be a set of star center indices of size $|J'| \geq 2\Delta$, such that $U'_i = U'_j$ for all $i,j \in J'$.

The probability that $I \subseteq J'$ is at least ${\binom{2\Delta}{\Delta}}/{\binom{\ell \Delta}{\Delta}}$, which is, again, a positive constant that depends only on $\extra$ (as both $\ell$ and $\Delta$ are bounded by a function of $\extra$ in this case). When this event $\{I \subseteq J'\}$ occurs, the recourse is at least $\Delta - (\extra + 1) \geq 1$ (since $\extra \leq \Delta - 2$). Hence, in this case as well, the total expected recourse is $\Omega(1) = \Omega(\Delta)$.

\end{proof}

\else

\begin{proof}[Proof Sketch.]
For simplicity, we give a sketch for $\extra=1$ and large enough $\Delta$. The remaining case, where $\Delta$ is bounded by a constant, is handled in the full proof by a simple separate argument.

We construct a distribution over incremental update sequences on which every deterministic algorithm has expected amortized recourse $\Omega(1/\Delta)$, the theorem then follows from Yao's principle~\cite{yao1977probabilistic}.
Create $N= 2 \Delta$ vertex-disjoint stars of degree $\Delta - 1$, and keep one isolated vertex $u$. Then, select $\Delta$ star centers uniformly at random, and insert $\Delta$ edges connecting each of the selected star centers to $u$. In total, we have $N(\Delta-1) + \Delta = O(\Delta^2)$ updates. Thus, it suffices to show that any algorithm incurs $\Omega(\Delta)$ expected total recourse on this update sequence.

Consider the time $t'$ just after creating the $N$ stars and before inserting the edges incident to $u$. Let $M'$ be the set of colors that are \emph{not} free at any of the $\Delta$ randomly selected centers at time $t'$. Equivalently, each color in $M'$ appears on some leaf edge incident to every selected star center. 

The $\Delta$ edges incident to $u$ require $\Delta$ distinct colors. At time $t'$, only $\pal - |M'|$ colors are available to be used without recoloring. Each recoloring may increase this number by at most one by freeing a color of $M'$ at some star's center. Therefore, the recourse is at least $\Delta - (\pal-|M'|) = \Delta - (\Delta +1 -|M'|) = |M'| - 1$. It remains to show that $\expect{|M'|} = \Omega(\Delta)$.

At time $t'$, each star center has exactly two free colors. Let $d(\alpha)$ be the number of stars in which $\alpha$ is free. So, $\sum_{\alpha = 1}^{\pal} d(\alpha) = 2N$, and therefore, at least half of the colors $\alpha \in [\pal]$ have $d(\alpha) \leq 4N/\pal$. We call such a color \emph{sparse}. We argue that a sparse color $\alpha$ is in $M'$ with constant probability. There are at most $4N/\pal = 8\Delta/(\Delta+1) \leq 8$ stars in which $\alpha$ is free. Therefore, the probability of avoiding all of them is at least 
$\Pr[\alpha \in M'] \geq \prod_{j=0}^{\Delta-1}\left(1 - \frac{8}{N-j}\right) \geq \left(1 - \frac{8}{N-\Delta}\right)^\Delta = \left(1 - \frac{8}{\Delta}\right)^\Delta$.
For large enough $\Delta$, this probability is lower bounded by a constant, and together with the fact that at least half of the colors are sparse, we obtain $\expect{|M'|}  = \Omega(\Delta)$. 
\end{proof}

\fi

\iffullversion
\subsubsection{Fully Dynamic Forests (Insertions and Deletions)}\label{subsec:randomized_fully}
\else
\vspace{-0.2cm}
\paragraph{Fully Dynamic Forests (Insertions and Deletions)}
\fi

\iffullversion

\else

For the fully dynamic setting, we do not use amortization. Instead, we bound the expected recourse of each update separately. Here we have the same analysis for both rooted and unrooted forests.

\begin{theorem}[Expected recourse per update]
\label{theorem_recourse_randomized_short_version}
Let $\Delta \geq 3$ and $0 \le c \le \Delta-2$. Against an oblivious adversary, the expected recourse per update of \distalg{} is $\Theta\!\left(\min\left\{{\Delta}/{\extra},\,\log_{\pal-1} n\right\}\right)$.
\end{theorem}

\begin{proof}[Proof Sketch of the Upper Bound.]
For simplicity, consider an insertion of an edge $(p,r)$ (deletions are similar). We prove two upper bounds.

\looseness=-1
For the first bound, observe that the repair process begins only if the new color $\alpha$ chosen for $(p,r)$ appears on one of the child edges of $r$. Since the child edges of $r$ are uniformly colored from all $\pal$ colors, and there are at most $\Delta-1$ child edges, this happens with probability at most $(\Delta-1)/\pal$. If this happens, and $(r,u)$ is recolored from $\col{1}$ to $\col{2}$, the repair then continues one level deeper if $\col{2}$ appears on a child edge of $u$, which happens with probability at most $q = \left(\Delta-1\right)/\left(\pal-1\right)$ (the denominator is now $\pal-1$ because $u$ had a forbidden color $\alpha$ before, so its child-edges are colored uniformly from $\pal -1$ colors). Thus, after the first conflict, the repair continues with probability at most $q$ at each step. Thus, if the repair starts, its length is stochastically dominated by a geometric random variable with success probability $1-q$. Therefore, the expected repair length, conditioned on starting, is at most $1/(1-q)$. Multiplying this by the probability of at most $(\Delta-1)/\pal$ that the repair starts, and substituting the value of $q$, gives an $O(\Delta/ \extra)$ bound on the expected recourse for $\extra \geq 1$.

The second bound, which applies to any $\extra \geq 0$, comes from depth considerations. By the exact recoloring probabilities discussed in the beginning of this section, the probability that the repair reaches a fixed edge $e$ at depth $d$ is at most $(\pal-1)^{-d}$. A union bound over all edges gives an $O(\log_{\pal-1} n )$ bound on the expected recourse. 
\end{proof}
\fi

\iffullversion
Next, we bound the expected number of recolorings performed by one update in the fully dynamic setting.

\begin{theorem}[Expected recourse per update]\label{theorem_recourse_randomized}
Let $\Delta \geq 3$ and $0 \le c \le \Delta-2$. 
Against an oblivious adversary, the expected recourse of Algorithm~\ref{algorithm_randomized_main_rooted} and Algorithm~\ref{algorithm_randomized_main_unrooted} per update is at most
\[
\frac{\Delta}{\pal}\cdot \min\left\{\frac{\pal-1}{\extra},\; 2+\left\lceil \log_{\pal - 1} n\right\rceil\right\}
\qquad\text{if \ } \extra\ge 1,
\]
and at most
\[
2+\left\lceil \log_{\Delta-1} n\right\rceil
\qquad\text{if \ } \extra=0.
\]
\end{theorem}

\begin{proof}
The proof analyzes the recourse given a uniformly distributed proper coloring. Since both Algorithm~\ref{algorithm_randomized_main_rooted} and Algorithm~\ref{algorithm_randomized_main_unrooted} maintain it (by Theorem~\ref{theorem_distribution_invariant_rooted} and Theorem~\ref{theorem_distribution_invariant_unrooted} respectively), the proof applies to both. Also note that instead of separating the statement for $\extra=0$ and $\extra \ge 1$, we can regard $\frac{\pal-1}{\extra} = \infty$ and then the statement for $\extra \ge 1$ subsumes $\extra = 0$.

Fix a time $t$ and condition on the current forest $F_t$.
By Theorem~\ref{theorem_distribution_invariant_rooted} (or Theorem~\ref{theorem_distribution_invariant_unrooted}), the current coloring is distributed as $\mathcal D(F_t)$. Let $E$ be the event that the update recolors an existing edge at the vertex $r$ where \funcRootToChild (in an insertion) or \funcChildToRoot (in a deletion) is invoked.
Let $L$ be the number of edges recolored by the ensuing \funcFixForbidden cascade (if any).
Then:
\begin{equation}\label{eq:recourse-split}
\expect{\text{recourse}}
=\Pr[E]\cdot \expect{\,1+ L \mid E\,} = \Pr[E]\cdot \left(1+ \expect{\,L \mid E\,}\right).
\end{equation}
We first bound $\Pr[E]$.
Let $\ell=|\ch(r)|$ at the moment the procedure at $r$ is invoked.
For an insertion, $r$ is a root under $\mathcal D(F_t)$ and $E$ is the event that the parent-edge color $\col{2}$ appears on some child edge of $r$, which by~\eqref{eq:root-hit} has probability exactly $\frac{\ell}{\pal}$.
For a deletion, \funcChildToRoot recolors at $r$ with probability exactly $\frac{\ell}{\pal}$ by its own coin flip.
In both cases, $\Pr[E]=\frac{\ell}{\pal} \le \frac{\Delta}{\pal}$.

We now upper bound $\expect{L\mid E}$ in two ways.

\smallskip
\noindent
\textbf{A geometric bound (for $\extra\ge 1$).}
At a visited vertex $v$, the cascade continues past $v$ if and only if the new forbidden color appears on some child edge of $v$.
By~\eqref{eq:local-hit}, this happens with probability at most $q \equiv \frac{\Delta-1}{\pal-1}$.
Thus $L$ is stochastically dominated by a geometric random variable, with success probability $1-q = \frac{\extra}{\pal-1}$,  that counts the number of failures before the first success. Thus, $\expect{L\mid E}\le
\frac{q}{1-q}
=\frac{\Delta-1}{\extra}$. And so, $1+ \expect{L\mid E} \leq \frac{\pal-1}{\extra}$.

\smallskip
\noindent
\textbf{A depth bound (for all $\extra\ge 0$).}
Conditioned on $E$, the remaining recoloring process is a \funcFixForbidden cascade. Thus, by Lemma~\ref{lem:fixforbidden_edge_recolor_prob}, for any fixed edge at depth $i$ below the start vertex, the probability that it is recolored is at most $(\pal-1)^{-i}$. By a union bound over all vertices, we have $
\Pr[L\ge i\mid E]\le \min\{1,\; n\cdot (\pal-1)^{-i}\}$.
Using $\expect{L\mid E}=\sum_{i\ge 1}\Pr[L\ge i\mid E]$, and since $\pal-1 \ge \Delta - 1 \ge 2$, we get:
\[
\expect{L\mid E}
\le \sum_{i=1}^{\lceil \log_{\pal-1} n\rceil} 1 + \sum_{i>\lceil \log_{\pal-1} n\rceil} n\cdot (\pal-1)^{-i}
\le \lceil \log_{\pal-1} n\rceil + \sum_{j\ge 1} (\pal-1)^{-j}
\le 1+\left\lceil \log_{\pal-1} n\right\rceil.
\]
To conclude, combine the two bounds on $1+\expect{L\mid E}$, and substitute $\Pr[E] \le \frac{\Delta}{\pal}$ in~\eqref{eq:recourse-split}.
\end{proof}

We now show that this analysis is tight.

\begin{theorem}
\label{theorem_randomized_dynamic_tight}
Theorem~\ref{theorem_recourse_randomized} is tight up to constant factors for both Algorithms~\ref{algorithm_randomized_main_rooted} and Algorithm~\ref{algorithm_randomized_main_unrooted}.
\end{theorem}
\begin{proof}
Let $T_h$ be a complete rooted $(\Delta-1)$-ary tree of height $h$.
Consider an adversary which first construct two copies of $T_h$ with roots $r_1$ and $r_2$. Then, it repeatedly inserts and deletes $(r_1,r_2)$.

To get the lower bound, it suffices to consider the insertions of $(r_1,r_2)$. Consider one insertion, and let $r$ be the root that plays the role of the child (where a repair may start). As discussed in the proof of the upper bound (Theorem~\ref{theorem_recourse_randomized}), the repair starts if and only if the color $\alpha$ chosen for $(r_1,r_2)$ appears on one of the $\Delta-1$ child-edges of $r$, which happens with probability $(\Delta-1)/\pal$.  Then, conditioned on the repair starting, at each subsequent level the  \funcFixForbidden cascade continues with probability $p\equiv(\Delta-1)/(\pal-1)$, since the tree is complete.

Since $T_h$ has height $h$, the repair path has expected length $\frac{\Delta-1}{\pal}\sum_{i=0}^{h-1} p^i$, and this is the expected recourse of each $(r_1,r_2)$ insertion. If $\extra \ge 1$, then
$\sum_{i=0}^{h-1} p^i =  \Theta\!\left(
\min\left\{\frac{\Delta}{\extra}, h\right\} \right)$. If $\extra=0$, then $p=1$, so the sum is $\Theta(h)$. Now, since $T_h$ has $\Theta((\Delta-1)^h)$ vertices, we can have $h=\Theta(\log_{\Delta-1} n)=\Theta(\log_{\pal-1} n)$. Thus, for a long enough sequence of updates, the expected amortized recourse per update is $\Omega\!\left(\min\left\{\frac{\Delta}{\extra},\,\log_{\pal-1} n\right\}\right)$.
\end{proof}

\fi

\smallskip
\noindent
\textbf{Lower bounds for constant  $\extra$.}
In the fully dynamic setting, we prove a matching lower bound for $\extra=0$. We also prove an $\Omega(1)$ lower bound for every constant $c$. Both
lower bounds apply to randomized algorithms, even on rooted forests and against
an oblivious adversary. We begin with the case $\extra=0$.

\iffullversion
    \theoremLowerBoundTreesDeltaColors*
\else
    \begin{restatable}[]{theorem}{theoremLowerBoundTreesDeltaColors}
\label{theorem_trees_delta_colors_LB}
    For every $\Delta \geq 3$, any algorithm that maintains a proper $\Delta$-edge-coloring of a fully dynamic forest, has amortized recourse $\Omega\!\left(\log_{\Delta-1} n\right)$ against an oblivious adversary.
    \end{restatable}

\fi

\begin{proof}
    We start by constructing two perfect $(\Delta-1)$-ary trees of depth $h$, with roots $r_1,r_2$. Then, we alternate between linking $r_1$ and $r_2$ directly, and linking them both to the same isolated vertex $u$. Let $\freecol_i$ for $i \in \{1,2\}$ be the unique  color not used by the child edges of $r_i$. When $r_1$ and $r_2$ are linked directly, coloring the edge $(r_1,r_2)$ requires the free colors at the roots to be equal, i.e., $\freecol_1 = \freecol_2$. On the other hand, when $r_1$ and $r_2$ are linked to $u$, coloring the edges $(r_1,u)$ and $(r_2,u)$ requires two distinct colors, and therefore requires $\freecol_1 \neq \freecol_2$. Thus, every switch between the two configurations forces the free color at (at least) one of the roots to change.
    Changing the free color at a root requires recoloring a full root-to-leaf path of length $h$. This is because each internal vertex has degree $\Delta$ and thus, changing the color of its parent edge requires changing the color of (at least) one of its child edges, and this propagates down to a leaf. Therefore, each switch causes $\Omega(h)$ recourse with $O(1)$ updates. Repeating this switch sufficiently many times gives $\Omega(h)=\Omega(\log_{\Delta-1} n)$ amortized recourse per update.
\end{proof}

\iffullversion
\thmFullyRandomizedLB*

\begin{proof}
    We use the same construction from the proof of Theorem~\ref{thm:incremental_randomized_lb}, with the same choice of a constant $\ell$, and $N= \ell \Delta$. First construct $N$ vertex-disjoint stars of degree $\Delta - 1$ with centers $v_1,\dots,v_N$, and keep an isolated vertex $u$.

    We now repeat the following phase $\Delta$ times:  Draw a uniformly random subset $I \subseteq [N]$ of size $\Delta$, and insert the edges $\{(v_i,u)\}_{i \in I}$. Then delete these same edges.

    The proof of Theorem~\ref{thm:incremental_randomized_lb} shows that the expected recourse caused by inserting the edges $\{(v_i,u)\}_{i \in I}$ is $\Omega(\Delta)$. Thus, each phase contributes $\Omega(\Delta)$ expected recourse. Since we repeat this for $\Delta$ phases, the expected total recourse is $\Omega(\Delta^2)$. 

    The total number of updates is $O(\Delta^2)$ as the initial construction of the stars requires $O(\Delta^2)$ insertions, and each of the $\Delta$ phases requires $2\Delta$ updates. Hence the expected amortized recourse per update is $\Omega(1)$. 
\end{proof}

\else

We conclude with a $\Omega(1)$ lower bound for every constant $\extra$. 
\begin{restatable}{theorem}{thmFullyRandomizedLB}\label{thm:fully_randomized_lb}
        For every constant $\extra \geq 0$ and every $\Delta \geq 3$ with
$\extra \leq \Delta-2$, any randomized algorithm that maintains a proper $\pal$-edge coloring of a fully dynamic forests of maximum degree $\Delta$ has expected amortized recourse $\Omega(1)$ against an oblivious adversary, where the hidden constant depends only on $\extra$.
\end{restatable}

The proof follows by reusing the incremental hard distribution from the proof of Theorem~\ref{thm:incremental_randomized_lb}. Deletions allow us to reuse the constructed stars, and hence obtain larger recourse per update. Recall that the incremental lower bound shows that connecting a uniformly random subset of $\Delta$ star centers to an isolated vertex $u$ causes $\Omega(\Delta)$ expected recourse. We then delete these $\Delta$ edges, draw a new random subset of $\Delta$ star centers, connect them to $u$, and repeat. Each such phase requires $2\Delta$ updates and causes expected recourse $\Omega(\Delta)$. Repeating this sufficiently many times gives expected amortized recourse of $\Omega(1)$.
\fi

\fullversiontrue

\newpage

\section{Concluding Remarks}
\label{section_conclusions}
In this paper, we investigated the recourse required to maintain a proper edge-coloring of forests, in both the incremental and fully dynamic models. Although forests are much simpler than general graphs, this problem still poses intricate algorithmic and analytical challenges, and exhibits a rich landscape of recourse bounds. We provided tight analyses for natural algorithms, including \greedy{} up to tie-breaking and the randomized distribution-maintenance algorithm \distalg{}. We also proved several lower bounds for any algorithm, showing separations between the incremental and fully dynamic models in both the deterministic and randomized settings, and establishing the optimality of \distalg{} in several regimes.  Nevertheless, several intriguing questions remain open for future work. The main ones are as follows:

\begin{enumerate}
    \item Our analysis of greedy{} in the incremental model is tight under a particular tie-breaking rule (Theorem~\ref{theorem_insert_only_LB}).  Could the $O\!\left(\invExtraPlusSqrt \right)$ upper bound of Theorem~\ref{theorem_amortized_insertion_only_improved_UB} be improved by a better tie-breaking rule, or by a different deterministic algorithm?

    \item What is the optimal recourse that  fully dynamic deterministic algorithms can achieve?  Theorem~\ref{theorem_amortized_DplusC_LB} gives an $\Omega(1)$ lower bound, and we have a matching upper bound only for rooted forests with $2\Delta-2$ colors (Theorem~\ref{theorem_rooted_trees_recourse_constant_UB}). Can we get $O(1)$ amortized recourse for unrooted forests, or for a wider range of palettes? 

    \item Can one prove a separation between deterministic and randomized algorithms? For example, in the incremental model, \distalg{} achieves an optimal $\Theta(1/\Delta)$ expected amortized recourse for every constant $\extra$, while our best deterministic upper bound is $O(1/\sqrt{\Delta})$. Is randomization necessary to achieve $O(1/\Delta)$ recourse? 
    
    \item What is the optimal recourse when $\extra$ grows with $\Delta$? In particular, can we prove nontrivial lower bounds for $\extra = \Omega(\Delta)$? In the fully dynamic model, can the $\Omega(1)$ lower bound of Theorem~\ref{thm:fully_randomized_lb} be improved? Our results prove optimality of \distalg{} only for constant $\extra$ in the incremental model, and only for $\extra=0$ in the fully dynamic model. Resolving these questions is necessary for a complete understanding of the optimality of \distalg{}.
\end{enumerate}


\appendix
\section{Running Times}
\label{appendix_running_times_analysis}

As emphasized in the introduction, the focus of this paper is on recourse. In this section we complement our algorithms with running time analysis. Recall that not every update incurs recourse, but every update takes $\Omega(1)$ time to process. Therefore, a trivial lower bound on the update time is $\Omega(1+R)$ where $R$ is the recourse incurred by the update.

We show that our main technical result, \distalg{} can be implemented optimally up to a constant factor in the fully dynamic model.

We then show an implementation of a variant of \colorfulpath{} with $O(\Delta)$ amortized running time. In the realm of \greedy{} and its variants, one should not naturally expect highly efficient algorithms, because greediness entails a global perspective which usually requires processing the whole instance of a problem, the whole forest in our case. We still derive polynomial algorithms for this family.

Throughout, we assume that standard pointer and arithmetic operations, as well as sampling a uniformly random integer from a range of size at most $\pal$, take $O(1)$ time.

\subsection{\distalg{} Running Time}
In this section we consider implementations for \distalg{}.

\begin{theorem}
\label{theorem_runtime_distalg_arrays}
\distalg{} can be implemented in the fully dynamic model so that every update that incurs recourse $R$ is processed in $O(1+R)$ time.
\end{theorem}

\begin{proof}
In the fully dynamic model, the proof of Theorem~\ref{theorem_recourse_randomized} does not rely on which tree is re-rooted and linked as the child of the other when an edge is inserted. Therefore, for the purpose of efficient implementation, we can use a variant of \distalg{} that chooses this orientation arbitrarily. It remains to describe how to implement the edge insertion and deletion steps of Algorithm~\ref{algorithm_randomized_main_rooted}, together with the helper functions in Algorithm~\ref{algorithm_randomized_helpers}. We show that there
is a data structure that supports all required local operations at a vertex in $O(1)$ time. The required operations are: sampling a uniformly available color, sampling a uniformly used color, inserting or deleting a colored edge, recoloring an existing edge to an unused color, and swapping the colors of two existing incident edges. Since each local fix at a vertex performs $O(1)$ such operations and $O(1)$ additional work, this gives $O(1+R)$ time for an update that incurs recourse $R$.

For every vertex $v$, we maintain two arrays of length $\pal$, denoted by $H_v$ and $L_v$. We also maintain the degree $d_v \equiv \deg(v)$. The array $L_v$ stores a permutation of the colors $[\pal]$, maintained so that the first $d_v$ entries are exactly the colors currently used on edges incident to $v$. 
The array $H_v$ is indexed by colors. For a color $\col{1}$, the entry $H_v[\col{1}]$ stores two fields. The first is a pointer to the neighbor of $v$ over an edge of color $\col{1}$, or null if no such edge exists. The second is the location (index) of $\col{1}$ in $L_v$. These arrays are initialized lazily. Initially, $L_v$ is interpreted as the identity permutation, so every color $\col{1}$ is interpreted as being located at position $\col{1}$ in $L_v$, and every neighbor pointer in $H_v$ is interpreted as null. Entries are initialized only when they are first changed.\footnote{This implementation requires $O(\Delta)$ space per vertex for a total of $O(\Delta \cdot n)$ space, although there can be at most $O(n)$ edges at any given time. 
The explicit arrays at $v$ can be replaced by implicit representations using hash tables of size $O(d_v)$, with an additional tweak to the invariant such that the permutation $L_v$ is maintained as a permutation of order at most $2$ (keeping explicitly only the swapped pairs). This gives an optimal $O(n)$ space, but in exchange the update time is expected and amortized $O(1+R)$.}

The invariant on $L_v$ lets us sample colors in $O(1)$ time. To sample a uniformly available color at $v$, we sample $x \sim_U [d_v+1,\pal]$ and return $L_v[x]$. To sample a uniformly used color at $v$, we sample $x \sim_U [d_v]$ and return $L_v[x]$.

We now describe how local updates are maintained at $v$. Suppose that an edge from $v$ to $u$ is inserted with color $\col{1}$. We first increment $d_v$. Then, using $H_v[\col{1}]$, we find the current location of $\col{1}$ in $L_v$.  Since $\col{1}$ was previously unused at $v$, this location is at least $d_v$ (by the invariant). We then swap $\col{1}$ with the color currently stored at $L_v[d_v]$. This makes $\col{1}$ one of the first $d_v$ entries of $L_v$, preserving the invariant that these entries are exactly the used colors. We also update the affected color locations in $H_v$, and set the neighbor pointer in $H_v[\col{1}]$ to $u$. This preserves the invariant and takes $O(1)$ time.

When an edge of color $\col{1}$ from $v$ to $u$ is deleted, we swap $\col{1}$ with the color currently stored at $L_v[d_v]$, update the affected color locations in $H_v$, set the neighbor pointer in $H_v[\col{1}]$ to null, and then decrement $d_v$. This also preserves the invariant and takes $O(1)$ time.

A recoloring of an edge $(v,u)$ from color $\col{1}$ to a color $\col{2}$ that is currently unused at $v$ is implemented by swapping the locations of $\col{1}$ and $\col{2}$ in $L_v$, updating their locations in $H_v$, setting the neighbor pointer of $H_v[\col{2}]$ to $u$ and the pointer of $H_v[\col{1}]$ to null. Finally, if two existing incident edges of $v$ swap their colors $\col{1}$ and $\col{2}$, then both colors remain used at $v$, so we keep $L_v$ unchanged. We only swap the corresponding neighbor pointers in $H_v[\col{1}]$ and $H_v[\col{2}]$. This also preserves the invariant and takes $O(1)$ time.

The edge insertion and deletion steps of Algorithm~\ref{algorithm_randomized_main_rooted} use only the operations described above: they sample a color, insert or delete the updated edge at its endpoints, and then call one of the helper functions in Algorithm~\ref{algorithm_randomized_helpers}. Similarly, each non-recursive step of the helper functions performs only $O(1)$ such local operations at a vertex. In particular, a recursive step of $\funcFixForbidden{v,\col{1},\col{2}}$ that moves the conflict from $v$ to its child (rather than finish the recoloring) is implemented by swapping at $v$ the color of the parent edge $\col{1}$ with the child edge of color $\col{2}$. Each recursive step corresponds to one unit of recourse: it recolors exactly one edge and continues with the same repair at the child vertex. Therefore, an update that incurs recourse $R$ performs $O(R)$ recoloring steps and $O(1)$ additional work, so the total update time is $O(1+R)$.
\end{proof}

\begin{theorem}
\label{theorem_runtime_distalg_arrays_bsts_variant}
\distalg{} can be implemented in the incremental model so that its expected amortized running time is $O(\alpha(n) + R)$ where $R$ is its expected amortized recourse and $\alpha(n)$ is the inverse Ackermann function.
\end{theorem}

\begin{proof}
We can repeat the proof of Theorem~\ref{theorem_runtime_distalg_arrays}, which gives $O(1 + R)$ running time. The one missing argument regards the re-rooting of the trees done by Algorithm~\ref{algorithm_randomized_main_unrooted}, which could be ignored when discussing the fully dynamic model.

To know the size of the trees being linked, we maintain an additional union-find data structure, initially each vertex is a singleton. When we link two trees, we \emph{unite} them and sum their size, and when we need to check a size we \emph{find} it from the vertex in question. If we determined that both trees are small, checking a degree takes $O(1)$ since we already maintain $\deg(v)$ for $v$. The added amortized running time due to the union-find data structure is $O(\alpha(n))$, so the total expected amortized bound is $O(\alpha(n) + R)$.
\end{proof}

\subsection{\colorfulpath{} Running Time}

\begin{theorem}
Consider a variant of \colorfulpath{} as detailed in the proof of this statement. We can implement it with $O(\Delta)$ amortized update time.
\end{theorem}

\begin{proof}
The variant of \colorfulpath{} differs as follows:
\begin{itemize}
    \item The proof of Theorem~\ref{theorem_rooted_trees_recourse_constant_UB} defined the concept of a \emph{ready} edge and used it for the analysis. A ready edge was guaranteed to have an available color if recoloring reaches it. It was also argued (mostly by definition) why every vertex can have at most one ready edge towards a child. Therefore, in Algorithm~\ref{algorithm_rooted_forest}, when we cannot finish recoloring on the current edge $(u,v)$, we don't consider every possible child $w$ of $v$ to see if $A(w) \cap A(v) \ne \emptyset$, we only consider its \emph{ready} child, if such exists, and otherwise we proceed as if no child $w$ is viable.
\end{itemize}

To implement the algorithm efficiently, every vertex maintains its edges in a list. For convenience each edge also points to its endpoints and its place in their lists, and holds its color. Insertion and deletion take $O(1)$ time. Every edge also maintains a pointer to its \emph{ready} child-edge, if such exists, with its ready to use color. Every vertex also marks its parent among its neighbours (recall that the forest is \emph{rooted}).

Revisiting the proof of Theorem~\ref{theorem_rooted_trees_recourse_constant_UB}, we see that every insertion, deletion or recoloring of an edge affects the readiness of $O(1)$ edges in the vicinity of the update. For example, when an edge $e$ is recolored it makes its grandparent edge $e' =(p,v)$ ready, if $e'$ was recolored during the same update. We can use the parenting information to update the readiness of $e'$ for $p$. More generally, since readiness of an edge is only affected by changes in its neighbour edges or a grandchild edge, we can use the parenting information to update readiness in $O(1)$.

The main difficulty is to pick a color for an edge, and we differentiate a few cases:
\begin{enumerate}
    \item In the final \emph{else-case} of the algorithm, in the call to function \funcChooseColor{}, we must choose a used color of a child-edge of $v$ that is different than a specific color. We can simply consider the first three neighbours in the list of neighbours: at most one is the parent (and the edge $(\parent(v),v)$ is currently uncolored), and at most another uses the forbidden color. Then it takes $O(1)$ time to determine the color and neighbour over which to execute \funcRecursiveStep{}. Note that $v$ is guaranteed to have three neighbours: $\Delta \ge 3$ and $\pal = 2\Delta-2$, so if $\deg(v) < 3$ then the edge $(\parent(v),v)$ must have had an available color, which contradicts the need to shift a color from a child of $v$.

    \item When recoloring reaches a ready edge $e$, it simply colors it with the associated ready color that was known and was stored when the edge became ready. By design, this is the color of the parent of $e$ at the moment that $e$ became ready.

    \item The bottleneck of the running time is checking if a newly inserted edge $(u,v)$ has an available color, $\col{1} \in A(u) \cap A(v)$. We can use a temporary array of size $\pal$ to mark colors used by $u$ and $v$ and then re-scan it and see if any color remains unmarked. \qedhere
    \end{enumerate}
\end{proof}

\subsection{Greedy Family Running Times}
In this section we analyze the running times of the greedy variants. The greedy variants are simple to compute on trees in the polynomial sense. We emphasize that the greedy variants are natural to consider for recourse but are not necessarily natural for running time considerations because of their global nature that could be affected by non-local changes in the coloring of a given graph.

\begin{theorem}
\label{theorem_runtime_greedy_variants}
\greedy{}, \greedyshift{} and \greedypath{} can be implemented in $O(\Delta^4 \cdot n)$,  $O(\Delta \cdot n)$,  $O(n)$ update times, respectively.
\end{theorem}

\begin{proof}
We argue in increasing order of complexity, starting from \greedypath{}. When an edge $(u_1,u_2)$ is inserted between trees $u_1 \in T_1$ and $u_2 \in T_2$, assume that $T_i$ is rooted at $u_i$, \greedypath{} simply checks all the paths from $u_i$ to $w \in T_i$ to see if shifting the colors on this path ends with $(\parent(w),w)$ having an available color. Every leaf of $T_i$ satisfies that a color is available, so \greedypath{} just has to find the closest such $w$. It is possible to scan and compare all the shiftable paths in $O(n)$ time. Observe that we don't have to check paths separately but can use the overlap by doing a standard BFS scan. When we check if an edge $(x,y)$ has an available color we might pay $O(\Delta)$ time, but this time is amortized on extending the shiftable path to the other edges of $y$ (assuming we reached $(x,y)$ from $x$). By scanning using BFS, we are guaranteed to stop on the shortest shiftable path, in $O(n + \Delta) = O(n)$ time.

For \greedyshift{} we also consider the path to every $w$. The difference is that \greedypath{} immediately rules out paths with consecutive edges $e_1 = (x,y)$ and $e_2 = (y,z)$ if the color of $e_2$ cannot be shifted to $e_1$, while \greedyshift{} checks if this transition is possible when we perform a \emph{fan} around $y$. Computing the smallest fan that lets the chain proceed from $e_1$ to $e_2$ is taken into account when considering the total recourse of the path, and \greedyshift{} picks the shortest chain. An alternative perspective is this: We compute the shortest path in the line-graph, subject to being able to shift colors. The line graph has $O(\sum_{v \in V}{\deg(v)^2}) = O(\Delta \cdot n)$ edges (maximized when $\approx \frac{n}{\Delta}$ vertices have degree $\Delta$), so this would be the running time. While we can't simply do BFS on the line graph of a general edge colored graph, in a forest we can. To see why, note that because this is a forest, we can reach every vertex $y$ from a fixed specific neighbour (parent) $x$, so the history of the path to $y$ does not affect which color becomes free at $y$ due to shifts, it is always the color of $(x,y)$ that is freed at $y$ when we reach it.

Finally, \greedy{} is more complicated, but is still polynomially computable on trees, using dynamic-programming (DP) as follows. Given a newly inserted edge $e=(u_1,u_2)$, define the natural rooting on the tree of $u_i$ with root $u_i$ ($i \in \{1,2\}$). To clarify, in the tree of $u_i$, an edge $(x,y)$ is defined such that $x=\parent(y)$ if $x$ is closer than $y$ to $u_i$. For every vertex $v$ denote $\col{1}(v) \equiv \colorf(\parent(v),v)$. We define subproblems $P(v,\col{2})$ as minimizing the recoloring within the subtree rooted at $v$ such that the color of its edge toward its parent is colored by $\col{2}$. We abuse notation and use $P(v,\col{2})$ to refer to both the recourse and the actual recoloring. If $v$ is a leaf, then $P(v,\col{2}) = (1-\delta_{\col{2},\col{1}(v)})$ where $\delta$ is the Kronecker delta function. For a non-leaf, $P(v,\col{2}) = \min_{\col{1}'} \Big \{ (1-\delta_{\col{2},\col{1}(v)}) + \sum_{w \in \ch(v)} {P(w,\col{1}'(w))} \Big \}$ where $\col{1}'(w)$ must be unique for each child $v$, and not equal to $\col{2}$. Finding an assignment $\col{1}'$ that minimizes the expression for $P(v,\col{2})$ is the \emph{Assignment Problem}, which can be solved, for example, by the \emph{Hungarian algorithm}. There are at most $\Delta-1$ children, and $\Delta+\extra-1$ colors that are not $\col{2}$, so we can compute $P(v,\col{2})$ in $O(\Delta^3)$ time. To pick the smallest recourse, we should color $e$ by $\col{2}$ such that $P(u_1,\col{2}) + P(u_2,\col{2})$ is minimized. There are $O(\Delta \cdot n)$ problems $P(\cdot,\cdot)$ (per vertex, per color), so the total running time is $O(\Delta^4 \cdot n)$.
\end{proof}

We note that Theorem~\ref{theorem_amortized_insertion_only_improved_UB} is proven using Theorem~\ref{theorem_logCp1_recourse_shiftable_UB} and Theorem~\ref{theorem_trees_delta_colors_sublinear_WC_UB} which both don't require us to compute the absolute greedy recourse, so the running times can be improved while maintaining the same proven bounds. For example, in the incremental model, if we restrict the computation to be greedy but only among recolorings of the small tree being linked, this semi-greedy approach means that every vertex participates in $O(\Delta \cdot \log n)$ DP problems, and the amortized running time improves to $O(\Delta^4 \cdot \log n)$.

\section{Additional Technical Analyses}
\label{section_technical_analysis_appendix}

This section provides the remaining technical details that were deferred to maintain the flow of the core amortized recourse analysis.

\subsection{Proving Lemma~\ref{lemma_calculus_summation}}
\label{section_subsection_calculus_proof}
We restate Lemma~\ref{lemma_calculus_summation} below for convenience.
    \lemmaSummations*
    
    \begin{proof}
    Consider the function $f(x)$ in the domain $x \ge 1$. Since $f'(x) = \frac{1}{\ln 2} \cdot \frac{1 - \ln x}{x^2}$ the function is decreasing for $x \ge 3$. Let $j$ be smallest index such that $x_j \ge 3$, then $f(x)$ is decreasing from $x_j$. $j \le 3$ since $x_3 \ge 4 \cdot x_1 \ge 4$. Now:
    \begin{gather*}
    \sigma_j \equiv \sum_{i \ge j}{f(x_i)}
    \le 
    \sum_{i \ge j}{f(2^{i-j} \cdot x_j)}
    =
    \sum_{i \ge j}{\frac{i-j + \lg x_j}{2^{i-j} \cdot x_j}}
    = \frac{2 + 2 \lg x_j}{x_j}
    \end{gather*}
    
    Since $x_j \ge 3$, we have that $\sigma_j = O(\frac{\lg x_j}{x_j}) = O(f(x_j))$. To conclude, we divide to cases:
    \begin{enumerate}
        \item If $\Delta \ge 6$: Then $x_1 \ge 3$ which implies $f(x_1) \le f(\frac{\Delta}{2}) = O(\frac{\lg \Delta}{\Delta})$, and $j=1$. Then $\sigma_1 = O(f(x_1)) = O(\frac{\lg \Delta}{\Delta})$.
    
        \item If $\Delta < 6$: Then $\frac{\lg \Delta}{\Delta} = \Theta(1)$. For every specific $x_i$ we have that $f(x_i) = O(1)$. Since $j \le 3$, we have that $\sigma_1 \le f(x_1) + f(x_2) + O(f(x_j)) = O(1) = O(\frac{\lg \Delta}{\Delta})$.
    \end{enumerate}

    Next, consider $g(x)$ in the domain $x \ge 1$. $g(x)$ is decreasing for $x \ge 3.08$. To see that, denote $t(x) \equiv \sqrt{2 \lg x}$ so $g(x) = \frac{t \cdot 2^t}{2^{t^2/2}}$. Then $g'(x) = 2^{t-t^2/2} \cdot (1 + \ln 2 \cdot t(1-t)) \cdot t'(x)$ and since $t(x)$ is monotone we get a maximum of $g(x)$ at $t^*$ such that $1 + \ln 2 \cdot t^*(1-t^*)) = 0$. This gives $t^* \approx 1.80$, or $x \approx 3.08$. Let $j$ be the smallest index such that $x_j \ge 4$, then $g(x)$ is decreasing from $x_j$. We have $j \le 3$ since $x_3 \ge 4 \cdot x_1 \ge 4$. Now:
    \begin{gather*}
    \sigma_{j+1} \equiv \sum_{i \ge j+1}{g(x_i)}
    \le 
    \sum_{i \ge j+1}{g(2^{i-j} \cdot x_j)}
    \le
    \sum_{k \ge 1}{\frac{\sqrt{2(k + \lg x_j)} \cdot 2^{\sqrt{2(k + \lg x_j)}}}{2^k \cdot x_j}}
    \end{gather*}
    Using the identities for $a,b \ge 2$: $\sqrt{a+b} \le \sqrt{a} \cdot \sqrt{b}$ and $\sqrt{a+b} \le \sqrt{a} + \sqrt{b}$ where $a=2k$ and $b = 2 \lg x_j$, we can proceed:
    \begin{gather*}
    \le
    \sum_{k \ge 1}{\frac{\sqrt{2k} \cdot \sqrt{2 \lg x_j} \cdot 2^{\sqrt{2k}} \cdot 2^{\sqrt{2 \lg x_j}}}{2^k \cdot x_j}}
    =
    g(x_j) \cdot \sum_{k \ge 1}{\frac{\sqrt{2k} \cdot 2^{\sqrt{2k}}}{2^k}}
    = O(g(x_j))
    \end{gather*}
    Where the last step is because for $k \ge 8$ we have $\sqrt{2k} \le \frac{k}{2}$ and $\lg \sqrt{2k} \le \frac{k}{4}$ so in total for $k \ge 8$ the terms are each bounded by $\frac{1}{2^{k/4}}$. To conclude, we divide to cases:
    
    \begin{enumerate}
        \item If $\Delta \ge 8$: Then $x_1 \ge 4$, $g(x_1) \le g(\frac{\Delta}{2}) = \Theta(g(\Delta))$ and $j=1$. Then $\sigma_1 = g(x_1) + \sigma_2 = O(g(x_1)) = O(g(\Delta))$.
    
        \item If $\Delta < 8$: Then $g(\Delta) = \Theta(1)$. For every specific $x_i$, we have $g(x_i) = O(1)$. 
        Since $j \le 3$, we have $\sigma_1 \le g(x_1) + g(x_2) + g(x_3) + \sigma_{j+1} = O(1)$ thus $\sigma_1 = O(g(\Delta))$. \qedhere
    \end{enumerate}
\end{proof}

\subsection{Sublinear Worst-Case Recourse for \texorpdfstring{$\Delta$}{Delta} Edge Coloring (used by Theorem~\ref{theorem_amortized_insertion_only_improved_UB})}
\label{section_subsection_sunlinear_wc_recourse_Czero}

Theorem~2.11 of \cite{sadeh2026vizingchainsimprovedrecourse} shows a linear worst-case recourse lower bound of $\Omega(n/\Delta)$ for shift-based algorithms when $\extra = 0$. A similar lower bound can also be seen in the proof of Theorem~\ref{theorem_insertion_only_Czero_high_LB} (Section~\ref{section_appendix_greedy_variants_differences}) where most of the amortized recourse is due to few expensive insertions. In this section we analyze an algorithm that circumvents these bounds and achieves a sublinear worst-case recourse, by doing more than shifting colors.

Loosely speaking, the algorithm implied by our analysis has a few decision points in which it is allowed to color an edge such that it conflicts with \emph{two} of its neighbours, in order to reduce the overall recourse. Note that the analysis is tight if the adversary is allowed to present a graph with a given coloring (then this would be the worst-case being analyzed). We restate the claim for convenience, but first a lemma we will require.

\begin{lemma}
\label{lemma_overlap_at_most_two}
\looseness=-1
Let $T$ be a tree of degree at most $\Delta$ that is $\Delta$-edge-colored, and assume that we shift colors along some \emph{incomplete} $(\col{1},\col{2})$-bicolored path of $k \ge 1$ edges starting from the root, then color the edge uncolored by the shift, by $\col{3}$, uncolor its $\col{3}$ neighbours if such exist and proceed to fix them by shifting either $(\col{3},\col{2})$ or $(\col{3},\col{1})$ bicolored paths of length $k' \ge 1$. Then the two paths reach disjoint subtrees, and each reachable subtree can be reached at most twice for different choices of $(k,\col{3})$.
\end{lemma}

\begin{figure*}[!t]
	\centering

    \begin{subfigure}{0.3\textwidth}
    \includegraphics[width=0.9\textwidth]{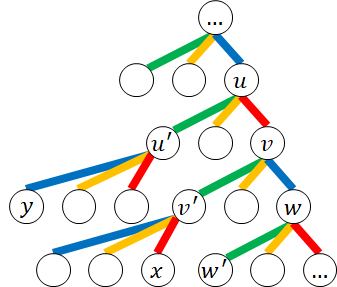}
    \subcaption{}\label{figure_conflict_doubling_and_overlap1}
    \end{subfigure}
    \begin{subfigure}{0.3\textwidth}
    \includegraphics[width=0.9\textwidth]{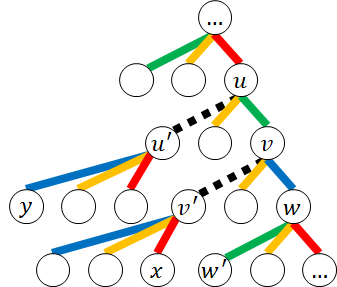}
    \subcaption{}\label{figure_conflict_doubling_and_overlap2}
    \end{subfigure}
    \begin{subfigure}{0.3\textwidth}
    \includegraphics[width=0.9\textwidth]{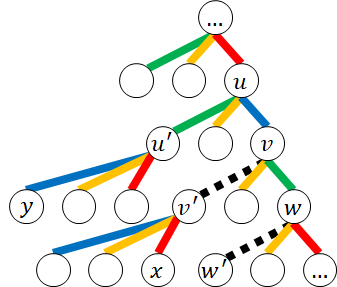}
    \subcaption{}\label{figure_conflict_doubling_and_overlap3}
    \end{subfigure}

	\caption{\small{An example for Lemma~\ref{lemma_overlap_at_most_two} where we recolor an edge by choosing a color unavailable at both ends. This choice converts a single problem into two, but may be useful in the long-run. In this example, let $\Delta=4$ and use only $4$ colors. Assume that we apply a bicolored flip on the red-blue right spine of the initial coloring (a), and that we choose to truncate the flip using green. We can choose to stop the flip on $(u,v)$ (b) or on $(v,w)$ (c). In each case, we get two new conflicts, $(v',v)$, and either $(u,u')$ or $(w,w')$, respectively. Both choices overlap in the problem over $(v,v')$ which can no longer be green, and in both cases the new bicolored path proceeds on a red-green path towards $x$.}}
	\label{figure_conflict_doubling_and_overlap}
\end{figure*}

\begin{proof}
Fix $k \ge 1$, and denote the edge over which we truncated the $(\col{1},\col{2})$-bicolored path by $e_1=(v_1,v_2)$ and its parent $e_0=(v_0,v_1)$ (there is a parent because $k \ge 1$). Without loss of generality assume that the color of $e_1$ was originally $\col{1}$, so its parent was colored by $\col{2}$. After shifting colors, $\col{2} \in A(v_1)$, and therefore coloring $e_1$ by $\col{3}$ implies that if there was a color conflict at $v_1$ due to $\col{3}$, we proceed to shift colors into a subtree of a sibling of $v_2$, over a $(\col{3},\col{2})$-bicolored path. Since $k' \ge 1$, the subtree which we reach is that of the ``$\col{3}$-sibling`` of $v_2$, or a deeper subtree within it. Similarly, after shifting colors on the first path, $\col{1} \in A(v_2)$ because even if the $(\col{1},\col{2})$-bicolored path proceeds and $e_1$ has a child-edge $(v_2,v_3)$ on this path, its color is $\col{2}$. By a similar argument, if $\col{3}$ was already used at $v_2$ then we fix the coloring over a $(\col{3},\col{1})$-bicolored path from $v_2$ to a sibling of $v_3$, and since $k' \ge 1$ we reach (at least) the subtree of the ``$\col{3}$-sibling`` of $v_3$. Because $T$ is a tree, indeed the subtrees are disjoint. See Figure~\ref{figure_conflict_doubling_and_overlap1} (initial) versus Figure~\ref{figure_conflict_doubling_and_overlap2} (recolored) for a visual clarification by example, where $(v_0,v_1,v_2,v_3)=(\parent(u),u,v,w)$ and $(\col{1},\col{2},\col{3}) = (\text{\color{red} red}, \text{\color{blue} blue}, \text{\color{OliveGreen} green})$. In this example one of the new paths to recolor is green-blue and is rooted at $u'$ (proceeds via $y$), and the other is green-red and is rooted at $v'$ (proceeds via $x$).

Next we argue that there are at most two ways to reach a specific subtree. Since the subtrees we reach are rooted at siblings of the $(\col{1},\col{2})$-bicolored path (or deeper), denote the subtree $T'$ and its lowest ancestor on the path by $v$. For a specific choice of $k$ we get at most two subtrees for the continuation, then conversely, there are at most two choices of $k$ to visit $T'$: we need $v \in e_1$. Revisit Figure~\ref{figure_conflict_doubling_and_overlap} for concreteness: either $k = k^*$ is such that $e_1 = (u,v)$ or $k = k^* + 1$ such that $e_1 = (v,w)$. The identity of the root of $T'$ also depends on the choice of $\col{3}$. So in total, there can be at most two choices of $(k,\col{3})$ that reach the same subtree.
\end{proof}

\theoremSubLinearRecourseTreesDeltaColors*

\begin{proof}
Because the graph is always a forest, when a new edge $e=(u,v)$ is inserted we can think of it as linking two trees $T_u$ and $T_v$ (one may be empty). If there is an available color shared by both $u$ and $v$, we use it without recourse at all. Otherwise, let $\col{1}_u$ and $\col{1}_v$ be available colors in $u$ and $v$ respectively, we can swap colors on the $(\col{1}_u,\col{1}_v)$-bicolored-path from $v$ to a leaf (or possibly a closer vertex), and color $(u,v)$ with the color $\col{1}_u$ which becomes available at $v$ after the color swap. This is guaranteed to work because the graph has no cycles, however the recourse might be large.

In order to reduce the recourse, we may decide to color-swap only up to $d_1$ edges, and truncate the process at some point, by uncoloring the $i$th edge on this bicolored path for $i \le d_1$, denote it $(u',v')$. This moves the problem down in $T_v$. Since we uncolored $(u',v')$ (temporarily) we may consider the available colors of $u'$ and $v'$ and repeat the process. However, since we only have $\Delta$ colors, if $u'$ and $v'$ are of degree $\Delta$ we might have no choice but to pick again the same bicolored path. So we must use a different approach... which would be to choose a color not available at either end, and now we need to fix two color conflicts, see Figure~\ref{figure_conflict_doubling_and_overlap} for a visual example. We increased the number of problems, but perhaps we can get for each of them a short bicolored path that reaches a leaf of $T_v$ such that the overall recoloring is smaller. Note that since the graph is a tree, each of the two new problems has a disjoint path of edges to leaves of $T_v$.\footnote{In fact, we can apply the choice of unavailable color already to the initial edge, and get that we have to fix the coloring in both $T_u$ and $T_v$ (one edge each). For simplicity we do not analyze this case, which is similar up to a constant factor, and depends on the sizes of both trees.}

We can generalize further the technique as follows: let each of the new bicolored paths be of length at most $d_2$. If we can finish it, be done. Otherwise, choose an index $i \le d_2$ where we truncate the color-swap, and then again we choose a color that results in two conflicts. Next, again, define a cap of at most $d_3$, etc. Overall, at level $\ell$ we may have up to $2^\ell$ unresolved color conflicts that we try to fix using bicolored paths of length at most $d_{\ell+1}$. Note that when we choose a color, there are $\Delta-2 \ge 1$ options ($2$ colors are actively part of the bicolored path).

Recall that $T_v$ is a tree, so whenever the color conflict diverges (when we color an edge $(x,y)$ with a color unavailable at both $x$ or $y$), we proceed to different parts of the tree. By Lemma~\ref{lemma_overlap_at_most_two}, each edge can be reached through at most two truncation choices, by virtue of reaching at most twice to the subtree that contains it. Therefore, we can choose a truncation out of the first $1 \le i \le d_i$ options, that proceed into two subtrees with total size that is at most $\frac{2}{d_i \cdot (\Delta-2)}$ fraction of the current tree. The term $d_i$ is due to our $d_i$ options for truncation, $(\Delta-2)$ is for the number of new colors to choose from, each leads us to different disjoint subtrees, and the $2$ is due to the possible overlap.\footnote{If we assume that every choice results in larger trees, and sum the number of edges over all the choices, we get more than double the count of edges. However every edge is reachable at most twice, which is a contradiction.} So, by choosing $i$ deterministically based on subtree sizes, we can ensure that at level $\ell$ of the process, each subtree with a color-conflict has size at most $S_\ell \equiv \frac{N \cdot 2^\ell}{(\Delta-2)^\ell \cdot d_1 \cdot \ldots \cdot d_\ell}$.

Once we have $S_\ell = O(1)$ we don't have to truncate, and this final sub-problem can be resolved with at most $S_\ell$ recourse. Then the total recourse is bounded by $R_\ell \equiv \sum_{i=1}^{\ell}{2^{i-1} \cdot d_i} + 2^{\ell} \cdot S_\ell$, because there are at most $2^{i-1}$ uncolored edges to fix at level $i$ of the process, and at this level we recolor (bicolor-swap) at most $d_i$ edges per fix. Before we optimize, let us re-write the expressions in a cleaner way. Define $D_i \equiv 2^{i-1} \cdot d_i$, then $R_\ell = \sum_{i=1}^{\ell}{D_i} + O(2^\ell)$
and $S_\ell = \frac{N \cdot 2^\ell \cdot 2^{\ell(\ell-1)/2}}{(\Delta-2)^\ell \cdot D_1\cdot \ldots \cdot D_{\ell}}$.
Finally, it remains to minimize $R_\ell$ by choosing $\ell$ and the values $D_1,\ldots,D_{\ell}$.

For a fixed value of $\ell$, observe that minimizing $R_\ell$ for a fixed $S_\ell$ is analogues to maximizing $D_1 \cdot \ldots \cdot D_{\ell}$ for a fixed sum $\sum_{i=1}^{\ell}{D_i}$, thus the best choice is a uniform value $\forall i: D_i \equiv D$. Then $S_\ell = \frac{N \cdot 2^{\ell(\ell+1)/2}}{(\Delta-2)^\ell \cdot D^{\ell}}$. Assuming that $S_\ell = \Theta(1)$ we get $D = \Theta \Big (\frac{N^{1/\ell} \cdot 2^{\ell/2}}{\Delta-2} \Big )$. We choose this value for $D$ even if $S_\ell = o(1)$ due to a too-large fixed value of $\ell$. Recall that $\forall i: 1 \le d_i \in \mathbb{N}$ and $D_i = 2^{i-1} \cdot d_i$, so given this desired value of $D$, we define $d_i \equiv \ceil{\frac{D}{2^{i-1}}} \le \frac{D}{2^{i-1}} + 1$. Then we get that:
$R_\ell = \sum_{i=1}^{\ell}{2^{i-1} \cdot d_i} + O(2^\ell) \le \sum_{i=1}^{\ell}{(D + 2^{i-1})} + O(2^\ell) = O \Big (\ell \cdot \frac{N^{1/\ell} \cdot 2^{\ell/2}}{\Delta-2} + 2^{\ell} \Big )
=
O \Big ( (\frac{\ell \cdot N^{1/\ell}}{\Delta-2} + 2^{\ell/2}) \cdot 2^{\ell/2} \Big )
$.

To balance $N^{1/\ell} = 2^{\ell/2}$ we need to choose $\ell = \sqrt{2 \lg N}$. Since $\ell$ is integer, we choose $\ell = \ceil{\sqrt{2 \lg N}}$, which makes $N^{1/\ell}$ slightly smaller than $2^{\ell/2} = O(2^{\sqrt{2 \lg N}})$. (If $N < 2$ choose $\ell=1$ so $N^{1/\ell}$ is well-defined.) Then we get: $R_\ell = O \Big ( (\frac{\ell }{\Delta-2} + 1) \cdot 2^{\ell} \Big ) = 
O \Big ( (\frac{\sqrt{2 \lg N}}{\Delta-2} + 1) \cdot 2^{\sqrt{2 \lg N}} \Big )$.

Finally, consider the running time. Our logic requires to know subtree sizes, and we cannot maintain those trivially. However, upon insertion of $(u,v)$ we get a natural rooting on the trees $T_u$, $T_v$ rooted at $u$ and $v$ respectively, and can compute subtree sizes for $T_i$ in $O(|T_i|)$ time, while also determining which tree is smaller. Once we have subtree sizes, the recoloring logic dictates that at each point we consider diverging the bicolored path, and when we do, we have to check $O(\Delta)$ colors for the subtree sizes. Overall, the running time is $O(N + \Delta \cdot R_\ell) = O(N + \Delta \cdot 2^{\sqrt{2 \lg N}})$.
\end{proof}

\begin{remark}
\label{remark_can_also_do_logn_on_balanced_tree}
It is important to note that Theorem~\ref{theorem_trees_delta_colors_sublinear_WC_UB} gives a guarantee for the worst-case tree. However, $\tilde{O}(2^{\sqrt{2\lg n}})$ might be far from being tight, for example consider a balanced tree, then we should expect only $O(\log n)$ recourse. This is indeed what we get: In terms of the proof of Theorem~\ref{theorem_trees_delta_colors_sublinear_WC_UB}, $d_1 = D_1 = D \approx \frac{2^{\sqrt{2 \lg n}}}{\Delta-2} = \omega(\log n)$ (assuming a fixed $\Delta$), therefore we will not have the ``opportunity'' to truncate the first bicolored path, instead we would just reach its end, and finish recoloring with only $O(\log n)$ recourse.
\end{remark}

\subsection{Theorems~\ref{theorem_amortized_DplusC_LB}, \ref{theorem_greedy_fails_LB}, \ref{theorem_greedy_fails_LB_shift_based}: The Relation Between \texorpdfstring{$n$}{n} to \texorpdfstring{$\Delta$}{Delta} and \texorpdfstring{$\extra$}{c} }
\label{section_dependence_of_n}

Some of our lower bounds require $n$ to be sufficiently large. In this section we explicitly analyze these quantities. We have the following debt:

\begin{theorem}
\label{theorem_exact_analysis_of_n_function_of_Delta}
In Theorem~\ref{theorem_amortized_DplusC_LB}, $n_0(\Delta,\extra) = \Theta((\extra+1) \cdot \Delta \cdot \binom{\Delta + \extra}{\extra+1})$. In Theorem~\ref{theorem_greedy_fails_LB} and Lemma~\ref{lemma_construct_initial_layered_tree}, $n_1(\Delta,\extra) = \Theta((\extra+1) \cdot \Delta^2 \cdot \binom{\Delta+\extra}{\extra+1} / \binom{\Delta}{\extra+1})$. In Theorem~\ref{theorem_greedy_fails_LB_shift_based}, $n_2(\Delta,\extra) = \max \{n'_2(\Delta,\extra) , n_1(\extra+2,\extra) \}$ where $n'_2(\Delta,\extra) = \Omega \Big ( \Delta \cdot x \cdot (\Delta+\extra-x) \cdot \binom{\Delta+\extra}{x} / \binom{\Delta}{x} \Big )$ for $x = \min \{ 2\extra+2, \Delta-(\extra+2) \}$.
\end{theorem}

\begin{proof}
The common thread of all $n_i(\Delta,\extra)$ is that we need to construct $N$ stars, each of degree $\Delta-d$, until $M$ of them use the same sub-palette. The values of $N$, $M$ and $d$ differ according to the scenario. 

\textbf{Theorem~\ref{theorem_amortized_DplusC_LB}:} $n_0(\Delta,\extra) = \Theta((\extra+1) \cdot \Delta \cdot \binom{\Delta+\extra}{\extra+1})$.
Let $1 \le d \le \Delta-1 - \extra$. We need $M = \extra+d+1$, therefore $N = (M-1) \cdot B + 1$ so that we can apply the pigeonhole principle, where $B \equiv \binom{\Delta+\extra}{\Delta-d} = \binom{\Delta+\extra}{\extra+d}$ is the number of different possible sub-palettes. One can verify that $d=1$ minimizes $B$.\footnote{The binom function $B(d) = \binom{\Delta+\extra}{\extra + d}$ has a single maximum, thus minimizing $B$ for $d \in [d_1,d_2] \equiv [1,\Delta-1-\extra]$ is achieved for either $d_1$ or $d_2$. $d_2$ is the unique minimizer if and only if $(\Delta+\extra) - (\extra+d_2) < (\extra+d_1) - 0$, which is false.
} Even though $d=1$ maximizes $M$, because $B$ is binomial while $M$ is linear, it is still best to optimize for $B$. In total, including $B$ additional vertices that are not part of the stars (the \emph{owners} in the proof), we get the constraint that:
$
n \ge N \cdot (\Delta-d+1) + B
= \Theta(N \cdot \Delta + B)
= \Theta(M \cdot B \cdot \Delta)
= \Theta((\extra+1) \cdot \Delta \cdot \binom{\Delta+\extra}{\extra+1})
$.

\textbf{Theorem~\ref{theorem_greedy_fails_LB} and Lemma~\ref{lemma_construct_initial_layered_tree}:} $n_1(\Delta,\extra) = \Theta((\extra+1) \cdot \Delta^2 \cdot \binom{\Delta+\extra}{\extra+1} / \binom{\Delta}{\extra+1})$. We need $M = (\extra+1) \cdot (\Delta-1)$ for the number of $P$-stars in the proof to bootstrap. Also, $d=\extra+1$ because each $P$-star is a star of degree $\extra+1$, and there are $B = \binom{\Delta+\extra}{\extra+1}$ different sub-palettes of this size. For the pigeonhole principle to apply in this proof, we require $N \ge (M-1) \cdot B + 1$. Then naively we would get $n_1 = \Theta(M \cdot B \cdot \Delta)$ as with $n_0$ (note that $M$ has a different value here). However, we can optimize and actually construct stars of degree $\Delta$, from which we later delete edges to get stars of degree $\extra+1$. Given the freedom to delete the edges that we want, each $\Delta$-star represents $\binom{\Delta}{\extra+1}$ different $(\extra+1)$-stars, so we only need $N'$ $\Delta$-stars such that $N' \cdot \binom{\Delta}{\extra+1} = (M-1) \cdot B + 1$. In conclusion, $n_1(\Delta,\extra) = \Theta(N' \cdot \Delta) = \Theta((\extra+1) \cdot \Delta^2 \cdot \binom{\Delta+\extra}{\extra+1} / \binom{\Delta}{\extra+1})$.

\textbf{Theorem~\ref{theorem_greedy_fails_LB_shift_based}:} $n_2(\Delta,\extra)$ is best described by two constraints that must both be satisfied:
\begin{enumerate}
    \item The proof of this theorem first reduces the problem by grouping all the vertices to form stars of degree $\Delta-\extra-2$ that use the exact same sub-palette. This can be achieved by either forming and replicating stars of degree $\Delta-\extra-2$ (directly), or by forming and replicating stars of the complementing sub-palette of size $2\extra+2$, which can be complemented later. The process of replicating stars and their complement is explained in the proof of Lemma~\ref{lemma_construct_initial_layered_tree}. So denote $x \equiv \min (2\extra+2,\Delta-(\extra+2))$, the first condition is that $n_2(\Delta,\extra) = \Omega \Big ( \Delta \cdot x \cdot (\Delta+\extra-x) \cdot \binom{\Delta+\extra}{x} / \binom{\Delta}{x} \Big )$.\footnote{The expression is dominated by the division of binomial terms, and for them choosing $x$ as the smallest is best: Let $y \equiv \max (2\extra+2,\Delta-(\extra+2))$. Therefore $x + y = \Delta+\extra$, then $\binom{\Delta+\extra}{x} = \binom{\Delta+\extra}{y}$ and we need to show that $\binom{\Delta}{x} \ge \binom{\Delta}{y}$. It cannot be that $x \le y \le \frac{\Delta}{2}$. If $\frac{\Delta}{2} \le x \le y$ then $\binom{\Delta}{x} \ge \binom{\Delta}{y}$. Finally, if $x \le \frac{\Delta}{2} \le y$, then $(\frac{\Delta}{2} - x)-(y - \frac{\Delta}{2}) = -\extra < 0$ which shows that $x$ is closer to $\frac{\Delta}{2}$ than $y$, thus again, $\binom{\Delta}{x} \ge \binom{\Delta}{y}$.} (Note that this expression and $n_1$ are both derived from the common expression $f(\Delta,c,z) \equiv \Delta \cdot z \cdot (\Delta+\extra-z) \cdot \binom{\Delta+\extra}{z} / \binom{\Delta}{z}$ where we either substitute $z=\extra+1$ for $n_1$ or $z=x$ in this case.)

    \item Once the reduction to the scenario of Theorem~\ref{theorem_greedy_fails_LB} occurs, the second condition is simply $n_2(\Delta,\extra) \ge n_1(\extra+2,\extra) = \Theta((\extra+1)^2 \cdot \binom{2\extra+2}{\extra+1})$ if $\extra \ge 1$, and $n_2(\Delta,\extra) \ge 1$ if $\extra = 0$. \qedhere
\end{enumerate}
\end{proof}

\begin{remark}
\label{remark_bootstrap_with_randomization}
The conditions $n \ge n_i(\Delta,\extra)$ in Theorem~\ref{theorem_exact_analysis_of_n_function_of_Delta} for $i \in \{1,2\}$ are only necessary to guarantee adversarial success under any circumstances. However, if recoloring choices are less arbitrary, for example chosen by the algorithm at random when possible, the conditions might not be required. To clarify, in such randomized case, instead of attempting to collect simultaneously a large number of stars that use the same sub-palette, the adversary can create the stars one by one, recreating each of them until (eventually) the star will have the desired sub-palette.
\end{remark}

\begin{remark}
\label{remark_n_Delta_relation_examples}
To get an idea how $n_0(\Delta,\extra)$ and  $n_1(\Delta,\extra)$ in Theorem~\ref{theorem_exact_analysis_of_n_function_of_Delta} behave, consider the extreme values of $\extra$:
\begin{enumerate}
    \item $\extra=O(1)$: Then $n_0 = \Theta(\Delta^{\extra+2})$.
    Also, $\binom{\Delta+\extra}{\extra+1} / \binom{\Delta}{\extra+1} \approx 1$, so $n_1 = \Theta(\Delta^2)$. Conversely, the corresponding statements apply for $\Delta = O(n_0^{1/(\extra+2)})$ and $\Delta = O(n_1^{1/2})$.

    \item $\extra=\Delta-2$: Then $\binom{\Delta}{\extra+1} = \Delta$, and using Stirling's approximation, $\binom{\Delta+\extra}{\extra+1} = \binom{2(\Delta-1)}{\Delta-1} \approx \frac{4^{\Delta-1}}{\sqrt{\pi (\Delta-1)}}$. So both $n_0$ and $n_1$ are exponential in $\Delta$. Conversely, the corresponding statements apply for $\Delta = O(\log n_i)$ for both $i \in \{0,1 \}$.
\end{enumerate}
\end{remark}

\bibliographystyle{alpha} 
\bibliography{reference}

\end{document}